\documentclass[a4paper,UKenglish,cleveref,autoref,thm-restate]{lipics-v2021}

\pdfoutput=1 \hideLIPIcs  

\title{A Rank-Preserving Gaifman Normal Form for First-Order Logic on Weighted Structures}

\titlerunning{Rank-Preserving Gaifman Normal Form for First-Order Logic on Weighted Structures}

\author{Steffen {van Bergerem}}
       {Humboldt-Universität zu Berlin, Germany}
       {steffen.van.bergerem@hu-berlin.de}
       {https://orcid.org/0000-0002-5212-8992}
       {This work was funded by the Deutsche Forschungsgemeinschaft
        (DFG, German Research Foundation) -- project number 541000908
        (gefördert durch die Deutsche Forschungsgemeinschaft (DFG) -- Projektnummer 541000908).}
\author{Martin Grohe}{RWTH Aachen University, Germany}{grohe@informatik.rwth-aachen.de}{https://orcid.org/0000-0002-0292-9142}{Funded by the European Union (ERC, SymSim,
101054974). Views and opinions expressed are however those of the
author(s) only and do not necessarily reflect those of the European
Union or the European Research Council. Neither the European Union
nor the granting authority can be held responsible for them.}
\author{Charlotte Lenz}{Humboldt-Universität zu Berlin, Germany}{charlotte.lenz@hu-berlin.de}{}{}
\author{Nicole Schweikardt}
       {Humboldt-Universität zu Berlin, Germany}
       {schweikn@hu-berlin.de}
       {https://orcid.org/0000-0001-5705-1675}
       {This work was funded by the Deutsche Forschungsgemeinschaft
        (DFG, German Research Foundation) -- project number 541000908
        (gefördert durch die Deutsche Forschungsgemeinschaft (DFG) -- Projektnummer 541000908).}

\authorrunning{S. van Bergerem, M. Grohe, C. Lenz, and N. Schweikardt} 

\Copyright{Steffen van Bergerem, Martin Grohe, Charlotte Lenz, and Nicole Schweikardt} 



\category{} 

\relatedversiondetails[cite=GroheSchweikardt_2026_Locality]
{This paper supersedes the following preprint, see \cref{sec:intro} for details}
{https://arxiv.org/abs/2606.11993}

\useofai{No AI tools were used for conducting the research presented in this paper.}

\nolinenumbers

\EventEditors{John Q. Open and Joan R. Access}
\EventNoEds{2}
\EventLongTitle{42nd Conference on Very Important Topics (CVIT 2016)}
\EventShortTitle{CVIT 2016}
\EventAcronym{CVIT}
\EventYear{2016}
\EventDate{December 24--27, 2016}
\EventLocation{Little Whinging, United Kingdom}
\EventLogo{}
\SeriesVolume{42}
\ArticleNo{23}

\usepackage{mathtools} \usepackage{stmaryrd} \usepackage{xcolor}

\crefname{conjecture}{Conjecture}{Conjectures}
\crefname{claim}{Claim}{Claims}

\newcommand{\ie}{i.\,e.}

\newcommand{\NN}{\mathbb{N}}
\newcommand{\NNpos}{\ensuremath{\NN_{\scriptscriptstyle \geq 1}}}
\newcommand{\QQ}{\mathbb{Q}}
\newcommand{\QQpos}{\ensuremath{\QQ_{\scriptscriptstyle > 0}}}
\newcommand{\RR}{\mathbb{R}}
\newcommand{\ZZ}{\mathbb{Z}}

\newcommand{\Structure}[1]{\ensuremath{\mathcal{#1}}}
\newcommand{\A}{\Structure{A}}
\newcommand{\B}{\Structure{B}}
\newcommand{\CN}{\Structure{N}}

\newcommand{\Class}[1]{\ensuremath{\mathcal{#1}}}
\newcommand{\CC}{\Class{C}}
\newcommand{\CG}{\Class{G}}

\newcommand{\I}{\ensuremath{\mathcal{I}}}

\newcommand{\X}{\ensuremath{\mathcal{X}}}

\newcommand{\deff}{\coloneqq}
\newcommand{\ffed}{\eqqcolon}

\newcommand{\neighb}[3]{\ensuremath{N_{#1}^{#2}(#3)}} \newcommand{\neighbr}[2]{\neighb{r}{#1}{#2}} \newcommand{\Neighb}[3]{\ensuremath{\mathcal{N}_{#1}^{#2}(#3)}} \newcommand{\Neighbr}[2]{\Neighb{r}{#1}{#2}} \newcommand{\nrA}[1]{\ensuremath{\neighb{r}{\A}{#1}}}
\newcommand{\NrA}[1]{\ensuremath{\Neighb{r}{\A}{#1}}}

\newcommand{\NrB}[1]{\ensuremath{\Neighb{r}{\B}{#1}}}

\newcommand{\abs}[1]{\left\lvert#1\right\rvert}
\newcommand{\bigabs}[1]{\bigl\lvert#1\bigr\rvert}

\newcommand{\norm}[1]{\left\lVert#1\right\rVert}

\DeclareMathOperator*{\ar}{ar}
\DeclareMathOperator{\dist}{dist}
\DeclareMathOperator*{\free}{free}
\DeclareMathOperator*{\qr}{qr}
\DeclareMathOperator*{\mdOperator}{md}
\newcommand{\md}[1]{\mdOperator \paren{{#1}}}

\renewcommand{\leq}{\leqslant}
\renewcommand{\geq}{\geqslant}

\renewcommand{\phi}{\varphi}
\renewcommand{\epsilon}{\varepsilon}

\newcommand{\set}[1]{\ensuremath{\{#1\}}}
\newcommand{\setc}[2]{\ensuremath{\set{#1 : #2}}}
\newcommand{\bigset}[1]{\ensuremath{\bigl\{ #1 \bigr\}}}
\newcommand{\bigmid}{\ \mathrel{:} \ }
\newcommand{\bigsetc}[2]{\bigset{#1 \bigmid #2}}

\newcommand{\tuple}[1]{\bar{#1}}
\newcommand{\ta}{\tuple{a}}
\newcommand{\tb}{\tuple{b}}
\newcommand{\tc}{\tuple{c}}
\newcommand{\tx}{\tuple{x}}
\newcommand{\ty}{\tuple{y}}
\newcommand{\tz}{\tuple{z}}

\newcommand{\tv}{\tuple{v}}

\newcommand{\tX}{\tuple{X}}
\newcommand{\tY}{\tuple{Y}}

\newcommand{\vars}{\ensuremath{\textsf{\upshape vars}}}
\newcommand{\varsof}{\ensuremath{\textup{\upshape vars}}}
\newcommand{\sem}[1]{\left\llbracket #1 \right\rrbracket} \newcommand{\Land}{\ensuremath{\bigwedge}}
\newcommand{\Lor}{\ensuremath{\bigvee}}

\newcommand{\bigO}{\mathcal{O}}

\newcommand{\AWstar}{\ensuremath{\mathsf{AW[\ast]}}}

\newcommand{\CollectionOfRingsFont}[1]{\mathbb{#1}}
\newcommand{\SC}{\CollectionOfRingsFont{S}} \newcommand{\Weights}{\ensuremath{\mathbf{W}}}
\newcommand{\weight}{\ensuremath{\mathtt{w}}}

\newcommand{\wtype}{\ensuremath{\textup{type}}}
 \newcommand{\ptype}{\ensuremath{\textup{type}}}
\newcommand{\plusSR}[1]{\ensuremath{+_{#1}}}
\newcommand{\minusSR}[1]{\ensuremath{-_{#1}}}
\newcommand{\malSR}[1]{\ensuremath{{\cdot}_{#1}}}
\newcommand{\nullSR}[1]{\ensuremath{0_{#1}}}

\newcommand{\plus}{\plusSR{}}
\newcommand{\mal}{\malSR{}}
\newcommand{\minus}{\minusSR{}}
\newcommand{\plusS}{\plusSR{S}}

\newcommand{\malS}{\malSR{S}}
\newcommand{\nullS}{\nullSR{S}}

\newcommand{\LogicFont}[1]{\ensuremath{\mathsf{#1}}}
\newcommand{\FO}{\LogicFont{FO}}
\newcommand{\FOMOD}{\LogicFont{FO{+}MOD}}
\newcommand{\FOMODplus}[1]{\ensuremath{\FOMOD^{+}_{#1}}}
\newcommand{\FOplus}{\FO^+}

\newcommand{\FOWA}{\LogicFont{FOWA}}

\newcommand{\FOW}{\ensuremath{\LogicFont{FOW}}}
\newcommand{\FOWplus}{\ensuremath{\FOW^+}}
\newcommand{\FOWun}{\ensuremath{\FOW_1}}
\newcommand{\FOWPSR}[2]{\ensuremath{\FOW(#1)[#2]}}
\newcommand{\FOWplusPSR}[2]{\ensuremath{\FOWplus(#1)[#2]}}
\newcommand{\FOWPS}{\FOWPSR{\PP}{\sigma,\SC,\Weights}}
\newcommand{\FOWplusPS}{\FOWplusPSR{\PP}{\sigma,\SC,\Weights}}
\newcommand{\ngFOW}{\ensuremath{\LogicFont{ngFOW}}}
\newcommand{\ngFOWplus}{\ensuremath{\ngFOW^+}}

\newcommand{\PP}{\ensuremath{\mathbb{P}}}
\newcommand{\Pred}{\ensuremath{\mathsf{P}}}

\newcommand{\Fparam}[3]{\ensuremath{\LogicFont{F}(#1,#2,#3)}}
\newcommand{\Fqkd}{\Fparam{q}{k}{d}}

\newcommand{\FOplusParam}[3]{\FOplus(#1,#2,#3)}
\newcommand{\FOplusqkd}{\FOplusParam{q}{k}{d}}

\newcommand{\FOMODplusParam}[3]{\FOMODplus{M}(#1,#2,#3)}
\newcommand{\FOMODplusqkd}{\FOMODplusParam{q}{k}{d}}

\DeclareMathOperator{\locRadOperator}{r}
\newcommand{\locRad}[3]{\locRadOperator ({#1}, {#2}, {#3})}
\newcommand{\myr}{\locRadOperator}

\DeclareMathOperator{\locRadFOOperator}{r_1}
\newcommand{\locRadFO}[3]{\locRadFOOperator ({#1}, {#2}, {#3})}

\DeclareMathOperator{\locOperator}{L}
\newcommand{\loc}[3]{\locOperator ({#1}, {#2}, {#3})}
\DeclareMathOperator{\sentOperator}{S}
\newcommand{\sent}[3]{\sentOperator ({#1}, {#2}, {#3})}

\newcommand{\Lqkd}{\loc{q}{k}{d}}
\newcommand{\Sqkd}{\sent{q}{k}{d}}

\newcommand{\Fraisse}{Fra\"{\i}ss{\'e}}

\usepackage{csquotes}

\DeclarePairedDelimiter{\paren}{(}{)}
\DeclarePairedDelimiter{\brac}{[}{]}
\newcommand{\setpadding}{\mathchoice{ \ }{ \, }{}{}}
\DeclarePairedDelimiter{\setS}{\{}{\}}

\DeclarePairedDelimiterX{\setcompS}[2]{\{}{\}}{\setpadding {#1} \setpadding : \setpadding {#2} \setpadding}

\newcommand{\enud}[2]{ #1, \dots, #2 }

\newcommand{\setd}[2]{\setS*{\enud{#1}{#2}}}
\newcommand{\tup}[1]{\paren*{ #1 }}
\newcommand{\tupd}[2]{\tup{\enud{#1}{#2}}}

\newcommand{\isdef}{\coloneqq}

\newcommand{\concat}{\mathrel{\kern-0.15em | \kern-0.1em | \kern-0.15em} }

\newcommand{\intsUpTo}[1]{\brac{#1}}

\newcommand{\proj}[2]{{#1}_{#2}}
\newcommand{\projTup}[2]{\proj{\myvec{#1}}{#2}}

\newcommand{\eval}[2]{#1 (#2)}
\newcommand{\annSym}[2]{\stackrel{\text{#1}}{#2}}

\newcommand{\annLeq}[1]{\annSym{#1}{\leq}}

\newcommand{\bigLand}{\bigwedge}
\newcommand{\bigLor}{\bigvee}

\newcommand{\bigLandClPad}[2]{
    \hspace{#2}
    \bigLand_{
        \mathclap{
            \substack{
                #1
            }
        }
    }
    \hspace{#2}
}
\newcommand{\bigLorClPad}[2]{
    \hspace{#2}
    \bigLor_{
        \mathclap{
            \substack{
                #1
            }
        }
    }
    \hspace{#2}
}

\newcommand{\relationpadding}{\mathchoice{\ \; }{}{}{}}
\newcommand{\deffp}{\relationpadding \isdef \relationpadding}
\newcommand{\equivp}{\relationpadding \equiv \relationpadding}

\newcommand{\modelsp}{\relationpadding \models \relationpadding}
\newcommand{\binarypadding}{\mathchoice{\ }{}{}{}}

\newcommand{\landp}{\mathbin{\binarypadding \land \binarypadding}}
\newcommand{\lorp}{\mathbin{\binarypadding \lor \binarypadding}}
\newcommand{\undp}{\landp}

\newcommand{\refsubstat}[2]{\hyperref[#2]{\ref*{#1}\ref*{#2}}}
\newcommand{\autorefsubstat}[2]{\hyperref[#2]{\autoref*{#1}\,\ref*{#2}}}

\newcommand{\baseCases}{Base Cases.\ }

\newcommand{\indStep}{Induction Step.\ }

\newcommand{\BaseCasesP}{\proofsubparagraph{\baseCases}}

\newcommand{\IndStepP}{\proofsubparagraph{\indStep}}

\DeclareMathOperator{\arOperator}{ar}
\DeclareMathOperator{\qrOperator}{qr}
\DeclareMathOperator{\freeOperator}{free}
\DeclareMathOperator{\distOperator}{dist}
\DeclareMathOperator{\predOperator}{\mathsf{P}}

\DeclarePairedDelimiterXPP{\arS}[1]{\arOperator}{(}{)}{}{#1}
\DeclarePairedDelimiterXPP{\qrS}[1]{\qrOperator}{(}{)}{}{#1}
\DeclarePairedDelimiterXPP{\mdS}[1]{\mdOperator}{(}{)}{}{#1}
\DeclarePairedDelimiterXPP{\freeS}[1]{\freeOperator}{(}{)}{}{#1}
\DeclarePairedDelimiterXPP{\distS}[2]{\distOperator^{#1}}{(}{)}{}{#2}
\DeclarePairedDelimiterXPP{\wtypeS}[1]{\wtype}{(}{)}{}{#1}
\DeclarePairedDelimiterXPP{\ptypeS}[1]{\ptype}{(}{)}{}{#1}
\DeclarePairedDelimiterXPP{\predS}[1]{\predOperator}{(}{)}{}{#1}

\newcommand{\nc}[1]{\newcommand{#1}}
\newcommand{\rnc}[1]{\renewcommand{#1}}

\nc{\myvec}{\bar}
\rnc{\vec}{\myvec}

\nc{\wmaxar}{\ensuremath{\omega}}
\nc{\myk}{\ensuremath{{k'}}}

\nc{\und}{\ensuremath{\wedge}}
\nc{\Und}{\ensuremath{\bigwedge}}
\nc{\oder}{\ensuremath{\vee}}
\nc{\Oder}{\ensuremath{\bigvee}}
\nc{\nicht}{\ensuremath{\neg}}
\nc{\impl}{\ensuremath{\to}}
\nc{\gdw}{\ensuremath{\leftrightarrow}}

\nc{\true}{\ensuremath{\textsf{\upshape true}}}
\nc{\false}{\ensuremath{\textsf{\upshape false}}}

\nc{\GG}[1]{\ensuremath{G_{#1}}}
\nc{\GGA}{\GG{\A}}

\nc{\myAlphas}{\ensuremath{\mathbb{A}}}
\nc{\myBetas}{\ensuremath{\mathbb{B}}}
\nc{\tildeDelta}{\ensuremath{\tilde{\Delta}}}
\nc{\tildealpha}{\ensuremath{\tilde{\alpha}}}
\nc{\tildebeta}{\ensuremath{\tilde{\beta}}}
\nc{\tildealphaStrich}{\ensuremath{\tilde{\alpha}{}'}}
\nc{\tildebetaStrich}{\ensuremath{\tilde{\beta}{}'}}
\nc{\altDelta}{\ensuremath{\nabla}}
\nc{\MnotNull}{\ensuremath{{M_0}}}
\nc{\NnotNull}{\ensuremath{{N_0}}}

\definecolor{ibm-ultramarine}{HTML}{648fff}
\definecolor{ibm-indigo}{HTML}{785ef0}
\definecolor{ibm-magenta}{HTML}{dc267f}
\definecolor{ibm-orange}{HTML}{fe6100}
\definecolor{ibm-gold}{HTML}{ffb000}
\usepackage{tikz}

\newcommand{\removalsig}[3]{\tilde{#1}^{#3}_{#2}}

\newcommand{\removalsigrs}[1]{\removalsig{#1}{r}{s}}
\newcommand{\sigmars}{\removalsigrs{\sigma}}
\newcommand{\removalsigma}[2]{\removalsig{\sigma}{#1}{#2}}
\newcommand{\distAtom}[3]{\dist(#1,#2) \,{\leq}\, #3}
\newcommand{\distAtomNeg}[3]{\dist(#1,#2) \,{>}\, #3}

\usetikzlibrary{backgrounds,calc,shadows}
\tikzset{dropshadow/.style={drop shadow={opacity=.4, shadow xshift=.25ex, shadow yshift=-.25ex}},
  vertex template/.style={draw, semithick, circle, inner sep=.3ex, minimum width=3.5ex, fill=white},
  vertex/.style={vertex template, dropshadow},
  edge/.style={thick},
}
 
\begin{document}

\setcounter{page}{0}

\maketitle

\begin{abstract}
  We prove a \emph{rank-preserving} version of Gaifman's Theorem.
  Compared to earlier rank-preserving locality theorems (in
  particular, [Grohe, Kreutzer, Siebertz, JACM~2017]), our theorem is
  much simpler and yields formulas in exactly the same
  normal form as Gaifman's original theorem.
  Furthermore, it holds not only for first-order logic, but also for first-order
  logic with modulo-counting quantifiers and, more generally, for the
  \emph{first-order logic on weighted structures}
  $\ngFOWplus$ that is introduced in this article.

  As an application of our theorem, we give a simplified proof of the
  algorithmic meta-theorem of [Grohe, Kreutzer, Siebertz, JACM~2017] stating that
  first-order properties of nowhere dense structures can be decided in
  almost-linear time. Our locality theorem for the weight logic $\ngFOWplus$ can be seen as an essential step toward such a meta-theorem for this logic.
\end{abstract}
\clearpage

\section{Introduction}
\label{sec:intro}
\emph{Locality} is a fundamental property of first-order logic with many applications,
mainly in proving expressivity lower bounds
(e.g., \cite{Fagin75,Hanf65,KuskeS17,Libkin97,Libkin00}) and algorithmic meta-theorems
(e.g., \cite{DreierMS23,FrickG01,GroheKreutzerSiebertz_2017_NowhereDense,SchweikardtSV22,Seese96}).
Locality comes in different forms (cf.~\cite{Libkin97,Libkin-FMT}),
but arguably the most important is the one captured by Gaifman's Theorem~\cite{GaifmanPaper}.
We say that a formula of first-order logic is in \emph{Gaifman normal form}
if it is a Boolean combination of local formulas and of basic local sentences,
i.e.\ sentences of the form
\begin{equation}\label{eq:intro-bl}
  \exists y_1\cdots\exists y_\ell\ \Big(\bigwedge_{1\le i<j\le\ell}
  \distAtomNeg{y_i}{y_j}{2r}\ \wedge\bigwedge_{1\le j\le \ell}\lambda(y_j)\Big),
\end{equation}
where $\distAtomNeg{y_i}{y_j}{2r}$ is a formula expressing that the distance
between $y_i$ and $y_j$ is greater than $2r$,
and the formula $\lambda(y)$ is $r$-local around $y$
(cf.~\cref{sec:preliminaries,sec:FO-Gaifman} for all necessary definitions).
Gaifman's Theorem states that every formula of first-order logic
is equivalent to a formula in Gaifman normal form.
The currently known proofs of this theorem
(cf.~\cite{EF-FMT,GaifmanPaper,KeislerLotfallah2004})
transform $\phi$ into a formula in Gaifman normal form that does not preserve
the quantifier rank, i.e.\ the local formulas and the formulas $\lambda(y)$
in the basic local sentences may have a larger quantifier rank than $\phi$.
In fact, we give an example
-- to the best of our knowledge the first such example --
showing that an increase in the quantifier rank
is necessary~(cf.~\cref{ex:fo-gaifman}).

For some applications of Gaifman's Theorem,
specifically for proving algorithmic meta-theorems for nowhere dense classes
\cite{DreierMS23,GroheKreutzerSiebertz_2017_NowhereDense,GroheSchweikardt_2018_FOC1,SchweikardtSV22},
this increase in quantifier rank is a serious problem.
The reason is that Gaifman's Theorem needs to be applied recursively,
and the depth of the recursion depends on the quantifier rank of the formulas
that still need to be processed.
To resolve this issue, Grohe, Kreutzer, and Siebertz~\cite{GroheKreutzerSiebertz_2017_NowhereDense}
introduced a new, two-parameter rank measure (\enquote{$q$-rank at most $\ell$})
and proved a locality theorem preserving the rank.
For other, more general algorithmic meta-theorems,
variants of this rank-preserving locality theorem have been
given~\cite{DreierT25,GroheSchweikardt_2018_FOC1}.
All these rank-preserving locality theorems,
including the original one by Grohe et al.~\cite{GroheKreutzerSiebertz_2017_NowhereDense},
are very complex, much more complex than Gaifman's original theorem.
This not only made them difficult to apply, but also error-prone.
Indeed, it recently turned out that the proof of~\cite[Theorem~7.1]{GroheSchweikardt_2018_FOC1},
which had been applied in \cite{GroheSchweikardt_2018_FOC1,SchweikardtSV22},
contains a subtle error.
The present paper grew out of our attempt to fix this error
(and hence salvage the applications of the theorem,
notably the algorithmic
meta-theorems of \cite{GroheSchweikardt_2018_FOC1,SchweikardtSV22}).
We not only succeeded in fixing the error,
but also ended up with a much simpler and arguably nicer rank-preserving locality theorem
that actually is a rank-preserving version of Gaifman's original theorem.
As an application of this theorem, we give a simplified proof of the main result
of~\cite{GroheKreutzerSiebertz_2017_NowhereDense}
that first-order properties of nowhere dense structures can be decided
in almost-linear time (see \cref{section:nowhere-dense}).
A similar but incomparable rank-preserving locality theorem for first-order logic
was obtained independently by Dreier and Toru\'nczyk~\cite{DreierT26a}.

We then realised that our techniques are actually applicable to a much broader class of logics,
extensions of first-order logic with the ability to count,
and on weighted structures, with the ability to aggregate values.
Logics with counting operators were first studied
in descriptive complexity theory~\cite{Immerman87,ImmermanL90,GradelO92}.
Their locality was later studied
in~\cite{HellaLN99,KuskeSchweikardt_ICALP18,Libkin00,Libkin01,LibkinN99}.
Grädel and Gurevich~\cite{GradelG98} introduced a general framework for logics
with weights and aggregation that they dubbed \emph{metafinite model theory}.
It forms the basis for many other such logics that have been introduced and analysed later,
for example~\cite{vanBergeremSchweikardt_2021_FOWA,Geerts23,GeertsMRVV21,GradelM95,Grohe26,GroheSSV25,KuskeS17,Torunczyk20}.

The specific logic we study here is based on \cite{vanBergeremSchweikardt_2021_FOWA}.
In general, logics with counting or aggregation are not local.
Intuitively, this is because vertices of a graph that are far apart may share a numeric value
(for example, the number of neighbours, or a weight) that appears nowhere else,
and this allows the logic to establish a \enquote{virtual} connection between these two vertices.
Thus, locality theorems for such logics can only be established for fragments
that do not allow such long-range virtual connections.
Van Bergerem and Schweikardt~\cite{vanBergeremSchweikardt_2021_FOWA}
proved a Gaifman locality theorem for the fragment $\FOWun$
that forbids comparisons between terms associated with different vertices
and allows aggregations only for values that come from a finite set.
However, while useful in other applications,
this theorem does not help in proving algorithmic meta-theorems for nowhere dense structures,
because it is not rank-preserving.
We establish a rank-preserving locality theorem for a larger fragment $\ngFOW$,
\enquote{neighbourhood-guarded $\FOW$},
which allows comparisons between terms associated with several variables,
but only if they are close together.
The logic not only subsumes $\FOWun$, but also first-order logic with
modulo-counting quantifiers.

A first step towards rank-preserving locality
is to separate the quantifier rank induced by assertions about distances between vertices
from the rank induced by the remaining quantifiers.
This can be done by introducing distance atoms of the form
$\distAtom{x}{y}{r}$ (for constants $r$) into the logic.
The resulting logic $\ngFOWplus$ has the same expressiveness as $\ngFOW$,
but as distance atoms are considered to be quantifier-free, the quantifier rank may change.
For $\ngFOWplus$-formulas $\phi$,
the distances $r$ occurring in distance atoms $\distAtom{x}{y}{r}$
provide a new natural resource to be restricted.
Instead of defining a complicated multi-parameter rank measure
(as~\cite{GroheKreutzerSiebertz_2017_NowhereDense} did),
we found it easier to stratify our logic $\ngFOWplus$
into a family of fragments $\Fqkd$ whose union is the full logic.
Roughly speaking, the parameter~$q$ controls the quantifier rank,
$k$ the number of free variables,
and $d$ (together with $q$ and $k$) controls the values allowed in the distance atoms.
However, the actual values permitted in a distance atom
also depend on the position of the atom in a formula.
One of our main new technical insights is how to refine the interaction
between the distances in atoms and their depth within the formula.

Our main result (\cref{thm:GNF_ngFOWplus}) states that every formula $\phi\in\Fqkd$
is equivalent to a formula in Gaifman normal form
that increases neither the quantifier rank $q$ nor the bounds on the distance atoms.
More precisely, $\phi$ is equivalent to a Boolean combination of local formulas in $\Fqkd$,
of basic local sentences of the form~\eqref{eq:intro-bl}
where $\lambda(y)$ is a local formula in $\Fparam{q{-}1}{k{+}\wmaxar}{d}$
(for some constant $\wmaxar$ depending on the arity of the weight symbols occurring in $\phi$),
and \enquote{local aggregation sentences} that aggregate the weights of tuples $\bar y$
satisfying a local formula $\lambda(\bar y) \in \Fparam{q{-}1}{k{+}\wmaxar}{d}$.
While locality is defined semantically in this theorem,
we actually prove a refined syntactic version (\cref{thm:GNF_ngFOWplus_refined})
where locality is enforced by the structure of the formulas.
Our proof combines ideas from the Ehrenfeucht--\Fraisse-style proof
of~\cite{GroheKreutzerSiebertz_2017_NowhereDense} with Gaifman's original,
more syntactic proof ideas \cite{GaifmanPaper}
and generalises these to the logic $\ngFOWplus$.
A key ingredient is a Feferman--Vaught decomposition for the logic (\cref{lem:fv_decomp})
that is compatible with all relevant parameters.
Furthermore, the combinatorial core of our construction differs significantly
from that of Gaifman's original proof;
and this is what enabled us to avoid an increase of the quantifier rank.

The structure of the remainder of the paper is as follows.
\cref{sec:preliminaries} provides basic notations.
\cref{sec:FO-Gaifman} shows that in Gaifman's Theorem for first-order logic,
an increase of the quantifier rank is unavoidable.
\cref{sec:fow} defines the logic $\ngFOWplus$.
\cref{sec:gaifman} proves our main result, a rank-preserving Gaifman normal form for $\ngFOWplus$.
\cref{sec:conclusion} discusses applications of our theorem and points out directions for future research.
\cref{sec:rpGNFforFOandFOMOD} translates our main result into rank-preserving versions
of a Gaifman normal form for first-order logic and for first-order logic with modulo counting,
both extended by distance atoms.
Further details are provided in the appendix,
including a simplified proof of the main result of~\cite{GroheKreutzerSiebertz_2017_NowhereDense}.

The present paper supersedes the preprint~\cite{GroheSchweikardt_2026_Locality},
which proves an analogous result for first-order logic
with an alternative proof based on Ehrenfeucht--\Fraisse\ games.
 \section{Preliminaries}
\label{sec:preliminaries}
We let \(\ZZ\), \(\NN\), \(\NNpos\), \(\QQpos\) denote the sets of
integers, non-negative integers, positive integers, and positive rationals,
respectively.
For \(m, n \in \ZZ\),
we let \([m,n] \deff \setc{\ell \in \ZZ}{m \leq \ell \leq n}\)
and \([n] \deff [1,n]\).
For a \(k\)-tuple \(\ta = (a_1, \dots, a_k)\), we write \(\abs{\ta}\)
to denote its \emph{length} \(k\).
We denote the power set of a set \(S\) by \(2^S\).
We assume that readers are familiar with basic notions of algebra such
as groups and rings, and with basic notions
of graph theory.
When referring to an abelian group (or ring), we will usually write
\((S,\plusS)\) (or \((S,\plusS,\malS)\)),
we denote the neutral element of
\((S,\plusS)\)
by \(\nullS\),
and \(\minus s\) denotes the inverse of an element \(s\) in
\((S,\plusS)\).
In this paper, graphs usually are undirected and simple (no loops or
parallel edges). We write $V(G)$ and $E(G)$ to denote
the set of vertices and the set of edges of a graph $G$. For graphs $G$ and
$H$ we write $H\subseteq G$ and, equivalently, $G\supseteq H$ to
indicate that $H$ is a subgraph of $G$. For a set $X\subseteq V(G)$,
we write $G[X]$ for the subgraph induced by $G$ on $X$, and we let
$G\setminus X\deff G[V(G)\setminus X]$.

A \emph{(relational) signature} is a finite set of relation symbols,
each relation symbol \(R\) coming with an \emph{arity} \(\ar(R) \in
\NN\).
Note that we admit 0-ary relation symbols, which is convenient in some
applications (cf.\ Section~\ref{section:nowhere-dense} and
\cite{GroheSchweikardt_2018_FOC1}).
Let \(\sigma\) be a signature.
A \emph{\(\sigma\)-structure \(\A=(A,(R^{\A})_{R\in\sigma})\)}
consists of a finite non-empty set \(A\), called the \emph{universe} of \(\A\),
and a relation \(R^{\A} \subseteq A^{\ar(R)}\) for every \(R \in \sigma\).
A \emph{(relational) structure} is a \(\sigma\)-structure for some signature \(\sigma\).
For a set \(X \subseteq A\), we let the \emph{induced substructure of \(\A\) on \(X\)}
be the \(\sigma\)-structure \(\A[X]\) with universe \(X\)
and \(R^{\A[X]} \deff R^{\A} \cap X^{\ar(R)}\) for every \(R \in \sigma\).

The \emph{Gaifman graph} $\GGA$ of a $\sigma$-structure $\A$ is the
graph with vertex set $A$ and edges $\set{a,b}$ for all distinct
$a,b\in A$ such that $a,b\in\set{a_1,\ldots,a_{\ar(R)}}$ for some
tuple $(a_1,\ldots,a_{\ar(R)})\in R^{\A}$ for some $R\in\sigma$.
We generalise graph-theoretic notions such as connectivity and degree
to $\sigma$-structures via their Gaifman graphs. In particular, the
\emph{distance} $\dist^{\A}(a,b)$ between two elements $a,b\in A$ is
the length (i.e., the number of edges) of a shortest path from $a$ to
$b$ in the Gaifman graph $\GGA$, or $\infty$ if no such path exists.
For every $r\in\NN$, the \emph{$r$-neighbourhood} of an element $a\in
A$ in a $\sigma$-structure $\A$ is the set $\nrA{a}\deff\setc{b\in
A}{\dist^{\A}(a,b)\leq r}$. More generally, for a $k\in\NNpos$ and a
tuple $\bar a=(a_1,\ldots,a_k)\in A^k$, we let $\nrA{\bar a}\deff
\bigcup_{i\in[k]}\nrA{a_i}$. Moreover, the \emph{$r$-neighbourhood
structure of $\bar a$} is the induced substructure $\NrA{\bar
a}\deff \A[\nrA{\bar a}]$.

We assume familiarity with first-order logic $\FO$ (for
background, cf.\ \cite{EbbinghausFT21}).
$\FO$ formulas are built from \emph{relational atoms},
(including equalities and atomic sentences $\true$ and $\false$)
by using Boolean combinations and existential quantifications.
The set of free variables of a formula $\phi$ is denoted by
$\free(\phi)$.
\emph{Sentences} are formulas $\phi$ with $\free(\phi)=\emptyset$.
We write $\phi(x_1,\ldots,x_k)$ to stipulate that
$\free(\phi)\subseteq\set{x_1,\ldots,x_k}$.
For a structure $\A$ and elements $a_1,\ldots,a_k\in A$ we
write $\A\models\phi(a_1,\ldots,a_k)$ if $\A$ with the variable
assignment $x_i\mapsto a_i$ for $i\in[k]$ satisfies $\phi$.
The notation $\phi(x_1,\ldots,x_k)$ is also convenient for
substitutions;
 by $\phi(y_1,\ldots,y_k)$ we denote the formula obtained from $\phi$ by
substituting $x_i$ with $y_i$ for all $i$ (and renaming bound
variables if necessary, cf.~\cite[Section~III.8]{EbbinghausFT21}).
We write $\phi\equiv \psi$ to express that two formulas $\phi$ and
$\psi$ are equivalent.
The \emph{quantifier rank} $\qr(\phi)$ of $\phi$ is
the maximum nesting depth of quantifiers in $\phi$.

Let \(r \in \NN\).
For a formula \(\phi\) and a non-empty tuple \(\tx = (x_1,
\dots, x_k)\) of variables with
$\free(\phi)\subseteq\set{x_1,\ldots,x_k}$, we say that $\phi(\tx)$
is \emph{\(r\)-local} (around \(\tx\))
if for every structure \(\A\) and all \(\ta \in A^k\)
we have \(\A \models \phi(\ta) \iff \NrA{\ta} \models
\phi(\ta)\); note that this implies $r'$-locality for all
$r'\geq r$.
By the \emph{coincidence lemma},
if $\phi(\bar x)$ is
$r$-local for some non-empty tuple $\bar x$ that contains all the free variables
of $\phi$, then $\phi(\bar x)$ is $r$-local for \emph{every} non-empty
tuple
$\bar x$ that contains all the free variables of $\phi$ (cf.\ \cref{appendix:LocalityNotion}).
A formula $\phi$ is \emph{$r$-local} if $\phi(\bar x)$ is $r$-local for some
non-empty  $\tx =(x_1,\ldots, x_k)$ with $\free(\phi)\subseteq
\set{x_1,\ldots,x_k}$, and
it is \emph{local} if it is \(r\)-local for some \(r \in \NN\).
In particular, for a relation symbol $R$ of
arity 0, the sentence $R()$ is 0-local.
 \section{Gaifman's Theorem for \texorpdfstring{\(\FO\)}{FO} is Not Rank-Preserving}\label{sec:FO-Gaifman}

For every signature $\sigma$ and $d\in\NN$, there is an $\FO$ formula
$\delta^{\sigma}_{\leq d}(x,y)$ such that, for all
$\sigma$-structures $\A$ and $a,b\in A$,
we have $\A\models\delta^{\sigma}_{\leq d}(a,b)\iff\dist^{\A}(a,b)\le
d$. The formula $\delta^{\sigma}_{\leq d}(x,y)$ depends on 
$\sigma$ and $d$ and can be chosen to be 
of quantifier rank $\lceil \log_2 d \rceil + \alpha-2$
for $\alpha\deff\max\setc{\ar(R)}{R\in\sigma}$, if $d\geq 1$ and
$\alpha\geq 2$, and of quantifier rank 0 if $d=0$ or $\alpha\leq 1$.
A \emph{basic local sentence} (for $\FO$ of signature $\sigma$) is an \(\FO\) sentence of
the form
\eqref{eq:intro-bl}, where \enquote{$\distAtomNeg{y_i}{y_j}{2r}$}  is the
formula $\nicht\,\delta^{\sigma}_{\leq 2r}(y_i,y_j)$, and
where $\ell\geq 1$, $r\geq 0$, and $\lambda(y)\in\FO$
is $r$-local and has only one free variable.
An \(\FO\) formula is said to be in \emph{Gaifman normal form}
if it is a Boolean combination of local formulas and basic local sentences.

\begin{theorem}[Gaifman \cite{GaifmanPaper}]\label{thm:gaifman}
  Every \(\FO\) formula \(\phi\) is equivalent to a formula \(\phi'\) in Gaifman normal form.
  Moreover, \(\free(\phi') = \free(\phi)\),
  and $\phi'$ can be computed from $\phi$.
\end{theorem}

Let \(\phi'\) be a formula in Gaifman normal form,
and let \(S\) and \(L\) be the sets of basic local sentences and of local formulas,
resp., such that \(\phi'\) is a Boolean combination of the
formulas in \(S \cup L\).
The \emph{outer quantifier rank} of \(\phi'\) is the maximum of the quantifier ranks
of all formulas in \(L\).
The next example shows that Gaifman's Theorem
does not preserve the outer quantifier rank:
it provides an \(\FO\) formula \(\phi(x)\)
with $\qr(\phi)=2$
that is not equivalent to any Boolean combination
of local formulas of quantifier rank  \(\leq 2\) and
of basic local sentences.

\begin{example}\label{ex:fo-gaifman}
  Let \(\sigma \deff \set{E, \mathsf{orange}, \mathsf{blue}}\)
  and consider the \(\FO\) formula
  \[
    \phi(x) \ \deff \ \exists y\, \Big(\mathsf{orange}(y)
    \land \neg\exists z\, \big(E(x,z) \land E(z,y)\big)
    \land \neg\exists z\, \big(\mathsf{blue}(z) \land E(y,z)\big)\Big).
  \]
  This formula states that there exists an orange node \(y\)
  that is not reachable from \(x\) by a walk of length \(2\)
  and that does not have a blue neighbour.

  \begin{figure}
    \centering
    \begin{tikzpicture}[scale=0.9]
      \def\nodedist{4ex}
\def\nodedistSquishedX{1.4*\nodedist}
\def\nodedistSquishedY{0.6*\nodedist}
\node[vertex, fill=ibm-orange] (3) {\(3\)};
\foreach \x in {1,2} {
  \node[vertex] at ($(3) + {3-\x}*(-\nodedistSquishedX, -\nodedistSquishedY)$) (\x) {\(\x\)};
}
\foreach \x in {4,5} {
  \node[vertex] at ($(3) + {\x-3}*(-\nodedistSquishedX, \nodedistSquishedY)$) (\x) {\(\x\)};
}
\foreach \x in {6} {
  \node[vertex] at ($(5) + {(\x-5)}*(\nodedistSquishedX, \nodedistSquishedY)$) (\x) {\(\x\)};
}
\node[vertex, fill=ibm-orange] at ($(6) + (\nodedistSquishedX, 0ex)$) (7) {\(7\)};
\foreach \x in {8} {
  \node[vertex] at ($(7) + (\nodedistSquishedX, 0ex)$) (\x) {\(\x\)};
}
\foreach \x in {10,11} {
  \node[vertex] at ($(5) + {\x-9}*(-\nodedistSquishedX, -\nodedistSquishedY)$) (\x) {\(\x\)};
}
\node[vertex] at ($(1) + (-\nodedistSquishedX, \nodedistSquishedY)$) (9) {\(9\)};
\foreach \x in {12,...,15} {
   \node[vertex] at ($(11) + {1.5*(\x-11)}*(-\nodedist, 0ex)$) (\x) {\(\x\)};
}
\node[vertex, fill=ibm-orange] at ($(15) + 1.5*(-\nodedist, 0ex)$) (16) {\(16\)};
\node[vertex, fill=ibm-ultramarine] at ($(16) + 1.5*(-\nodedist, 0ex)$) (17) {\(17\)};

\begin{scope}
  \clip ($(17) + (0, -2*\nodedist)$) rectangle ($(14) + (0, 2.5*\nodedist)$);
  \node[draw, circle, line width=.4ex, ibm-magenta!30, minimum width={14.65*\nodedist}] at (11) {};
  \node[draw, circle, line width=.4ex, ibm-magenta, minimum width={14.8*\nodedist}] at (11) {};
\end{scope}
\node[ibm-magenta] at ($(14) + (-.7*\nodedist,2*\nodedist)$) {\(\neighb{7}{\CG}{v} = \neighb{7}{\CG}{v'}\)};

\draw[edge] (1) -- (2) -- (3) -- (4) -- (5) -- (6) -- (7) -- (8)
  (1) -- (9) -- (11) -- (12) -- (13) -- (14) -- (15) -- (16) -- (17)
  (5) -- (10) -- (11);

\node at ($(1) + (-2.25ex, -2.5ex)$) (v) {\(v\)};
\node at ($(5) + (-2.ex, 3ex)$) (vs) {\(v'\)};
     \end{tikzpicture}
    \caption{Structure \(\CG=\CG_7\) of \cref{ex:fo-gaifman}.
      Undirected edges represent directed edges in both directions,
      \(\mathsf{orange}^{\CG} = \set{3,7,16}\),
      and \(\mathsf{blue}^{\CG}=\set{17}\).
    }
    \label{fig:fo-gaifman}
  \end{figure}

  Let \(\CG\) be the \(\sigma\)-structure depicted in \cref{fig:fo-gaifman},
  and let \(v \deff 1\) and \(v' \deff 5\)
  (as indicated in the figure).
  Clearly, \(\CG \models \phi(v)\) and \(\CG \not \models \phi(v')\).
  Observe that \(\neighb{7}{\CG}{v} = \neighb{7}{\CG}{v'}\),
  and let \(\CN_7 \deff \Neighb{7}{\CG}{v} = \Neighb{7}{\CG}{v'}\).
  In \cref{fig:fo-gaifman}, \(\CN_7\) is the substructure of \(\CG\)
  induced on the pink part.

  Now consider an arbitrary \(r \geq 7\),
  and let \(\CG_r\) be the \(\sigma\)-structure that looks exactly like \(\CG\),
  with the only difference that the path from vertex \(11\) to the leftmost orange vertex
  (vertex \(16\) for \(r=7\))
  has length \(r{-}2\)
  (hence, the distance between the leftmost orange vertex and \(v\) as well as \(v'\) is \(r\)).
  The structure \(\CG\) depicted in \cref{fig:fo-gaifman} is \(\CG_7\).
  Clearly,
  \(\CG_r \models \phi(v)\) and \(\CG_r \not \models \phi(v')\).
  Furthermore, \(\neighbr{\CG_r}{v} = \neighbr{\CG_r}{v'}\).
  Let \(\CN_r \deff \Neighbr{\CG_r}{v} = \Neighbr{\CG_r}{v'}\).
  We use Ehrenfeucht--\Fraisse\ games (cf., \cite{Libkin-FMT,EF-FMT}).
  An exhaustive exploration proves the following claim,
  which, combined with the Ehrenfeucht--\Fraisse\ Theorem
  (cf., \cite{Libkin-FMT,EF-FMT}),
  yields that \(\phi(x)\) is not equivalent to a Boolean combination \(\phi'(x)\)
  of local formulas of quantifier rank \(\leq 2\)
  and of basic local sentences (see \cref{appendix:fo-gaifman} for details).
\begin{restatable}{claim}{GaifmanCounterExample}
    \label{claim:GaifmanCounterExample}
    For every \(r \in \NN\) with \(r \geq 7\),
    Duplicator has a winning strategy for the 2-round
    Ehrenfeucht--\Fraisse\ game on \((\CN_r, v)\) and \((\CN_r, v')\).
  \end{restatable}
\end{example}
 \section{First-Order Logic on Weighted Structures:
\texorpdfstring{$\FOWplus$}{FOW+} and
\texorpdfstring{$\ngFOWplus$}{ngFOW+}}
\label{sec:fow}

We give a brief introduction into the weight aggregation logic \(\FOWplus\),
based on \cite{vanBergeremSchweikardt_2021_FOWA}
(details can be found in \cref{appendix:fow}).
The logic operates on \emph{weighted structures},
which are standard relational structures expanded by weight functions
that map tuples of elements to some abelian group or ring \(S\)
that is part of a collection \(\SC\) that is fixed as part of the vocabulary.
The (finite) set of available weight symbols \(\Weights\) is also fixed as part of the vocabulary.
It is required that each tuple with non-zero weight is contained in
some
tuple of some
relation of the structure (\enquote{\emph{locality condition}}).
Typical rings in \(\SC\) are \((\RR, +, \cdot)\),
\((\ZZ, +, \cdot)\),
or \((\ZZ/m\ZZ, +, \cdot)\).

The logic \(\FOW\) is a first-order logic that has the usual
relational atoms (including equality),
and it also has \emph{weight atoms} of the form \(\bigl(s = \weight(\ty)\bigr)\)
for a weight symbol \(\weight \in \Weights\) and an element $s$ in the
range $S$ of $\weight$.
Furthermore, elements in $S$ and weight symbols can be combined into terms
by using the group operation \(\plusS\),
and possibly \(\malS\) if \(S\) is a ring.
New atoms can be formed by relating such terms \(t\)
using predicates on the groups or rings \(S \in \SC\).
The collection \(\PP\) of predicates that can be used is part of the vocabulary.
Typically, these predicates include the binary equality predicate
and the natural order on rings such as \((\RR,+,\cdot)\).
However, we may also include more esoteric predicates,
for example, a predicate for the transcendental numbers.
The atoms can be combined using the usual Boolean connectives and quantifiers of first-order logic.
In addition, we have an \emph{aggregation} operator (or quantifier),
forming new formulas of the form \(\bigl(s = \sum \weight(\ty).\phi \bigr)\),
where \(\phi\) is a formula,
\(\weight\) is a weight symbol whose range is a
\emph{finite} group or ring \(S \in \SC\), and \(s \in S\).

The semantics of the logic is defined in the straightforward way.
Only the semantics of a formula
\(\psi(\tx) \deff \bigl(s = \sum \weight(\ty).\phi(\tx,\ty)\bigr)\) requires explanation:
suppose the range of \(\weight\) is \(S\).
Let \(\A\) be a structure and \(\ta \in A^{\abs{\tx}}\).
Then \(\A \models \psi(\ta)\)
iff \(s\) is equal to the sum in \(S\)
over all values \(\weight^\A(\tb)\),
where \(\tb\) ranges over all tuples \(\tb \in A^{\ty}\)
such that \(\A \models \phi(\ta,\tb)\).

Finally, \(\FOWplus\) is the extension of \(\FOW\)
where we also allow \emph{distance atoms} of the form \(\distAtom{x}{y}{d}\)
for constants \(d \in \NN\), with the obvious semantics.
The next example presents a concrete choice of $\SC$, $\Weights$,
a predicate collection \(\PP\),
and an $\FOWplus$ formula $\phi(x)$
that is not equivalent to any Boolean combination of
sentences and local formulas.

\begin{example}\label{example:noGNFforFOW}
Let $\sigma=\emptyset$, let $\SC=\set{S}$ for an \emph{infinite}
abelian group $S$, and let $\Weights=\set{\weight}$
for a unary weight symbol \(\weight\) with range \(S\).
Let \(\PP\)
contain the \emph{equality predicate for $S$}, denoted
by $\Pred_{=}$ with semantics $\sem{\Pred_=}=\setc{(s,s)}{s\in S}$.
Then,
\(
  \phi(x)  \deff
  \exists y\,\bigl(\neg\, x{=}y \ \land \
  \Pred_{=}\bigl(\weight(x), \weight(y)\bigr)\bigr)
\)
is an $\FOW$ formula that expresses
in
a $(\sigma,\Weights)$-structure
$\A$ that $x$ is assigned with an
$a\in A$ for which there exists
a $b\in A$ that is
different from $a$ but has the same weight as $a$.
It is not too difficult to prove that $\phi(x)$ is not
equivalent to any Boolean combination of
sentences and local formulas of $\FOWplus$ (see \cref{appendix:fow}).
\end{example}

The previous example shows that Gaifman's Theorem (\cref{thm:gaifman})
cannot be extended to the general logic $\FOWplus$.
Henceforth, we will consider the logic \(\ngFOWplus\),
obtained from \(\FOWplus\) by restricting the use of predicates from \(\PP\)
to the \emph{neighbourhood-guarded} version
\[\textstyle\Big(
   \Pred(t_1, \dots, t_m)
   \land
   \Land_{y \in \bigcup_{i=1}^{m} \free(t_i) \setminus \set{x}} \dist(x,y) \leq d
\Bigr),\]
where \(\Pred\) is an \(m\)-ary predicate from \(\PP\),
\(t_1, \dots, t_m\) are terms,
and $\free(t_i)$ denotes the set of variables occurring in $t_i$.

\begin{example}
  Let \(\SC \) contain the abelian group $(\ZZ,+)$,
  let $\Weights$ contain a unary weight symbol $\weight$
  with range $\ZZ$,
  and let \(\PP\) contain a predicate $\Pred_<$ with
  \(\sem{\Pred_<} \deff \setc{(i,j) \in \ZZ^2}{i < j}\).
Let \(\sigma \deff \set{E}\) be the signature of directed graphs.
  The \(\ngFOWplus\) sentence
  \(
  \phi
  \deff
  \forall x \forall y\,
  \Bigl(
    \neg E(x,y) \lor
    \big(\,
      \Pred_<(\weight(x), \weight(y))
      \und
      \dist(x,y)\,{\leq}\, 1
    \,\big)
    \Bigr)\)
  expresses in a vertex-weighted digraph that all edges are directed towards the
  vertex with the larger weight.
\end{example}

\begin{example}\label{ex:FOplus}
The logic $\FOplus$ of
\cite[Section~7.2]{GroheKreutzerSiebertz_2017_NowhereDense} extends
$\FO$ by the use of distance atoms.
This coincides with $\ngFOWplus$
(and with $\FOWplus$)
for the special case where
\(\SC = \Weights = \PP = \emptyset\), because then,
no weight symbols or predicates can be used.
\end{example}

\begin{example}\label{ex:FOMODplus}
For a finite set $M$ of integers $m\geq 2$, let $\FOMOD_M$ denote the
extension of
$\FO$
with modulo-counting quantifiers
$\exists^{i \text{\,mod\,}m}x$ for $m\in M$ and $i\in[0,m{-}1]$.
A formula  $\exists^{i \text{\,mod\,}m}x\,\phi$ expresses
in a $\sigma$-structure $\A$ that the number of elements $a\in A$ that
can be assigned to $x$ such that $\phi$ is satisfied is congruent to $i$
modulo $m$.  Extending this logic by the use of distance atoms, one
obtains the logic $\FOMODplus{M}$. This logic is embedded in
$\ngFOWplus$ for the special case where $\PP=\emptyset$, $\SC$
consists of the abelian groups $\ZZ/m\ZZ$ with addition modulo
$m$, for all $m\in M$, and $\Weights=\setc{\weight_m}{m\in M}$ with
\(\weight_m\) being a unary weight symbol with range \(\ZZ/m\ZZ\).
For this, we can turn any $\sigma$-structure $\A$ into the
$(\sigma,\Weights)$-structure $\hat{\A}$ by letting $\weight_m(a)=1\in
\ZZ/m\ZZ$ for all $m\in M$ and $a\in A$. Then, the formula $\exists^{i
  \text{\,mod\,}m}x\,\phi$ corresponds to the formula
$\big(i=\sum\weight_m(x).\phi\big)$. In fact, since $\PP=\emptyset$,
no predicates can be used, and there is a
1-to-1 correspondence that maps formulas $\psi(\bar x)\in\FOMODplus{M}$ to
formulas $\psi'(\bar x)\in\ngFOWplus$ (for this choice of $\SC,\Weights,\PP$)
such that, for every $\sigma$-structure $\A$
and every assignment $\bar a\in A^{|\bar x|}$,
we have $\A\models\psi(\bar a)\iff\hat{\A}\models\psi'(\bar a)$.
\end{example}
 \section{A Rank-Preserving Gaifman Normal Form for \texorpdfstring{\(\ngFOWplus\)}{ngFOW+}}
\label{sec:gaifman}
We fix a signature \(\sigma\),
a collection \(\SC\) of rings and/or abelian groups,
a finite set \(\Weights\) of weight symbols, and
a collection \(\PP\) of predicates.
Let $x_1,x_2,\ldots$ be a fixed enumeration of the available variables.
The \emph{quantifier rank $\qr(\phi)$} of an $\ngFOWplus$ formula
$\phi$ is the maximum nesting depth of constructs using
existential quantification or aggregation
in order to construct $\phi$.
The \emph{maximum distance} \(\md{\phi}\)
is the least \(r \in \NN\) such that $r'\leq r$ holds
for all constructs \(\dist(x,y)\,{\leq}\, r'\) appearing in \(\phi\).
We will write
\enquote{$\dist(x,y)\,{>}\,r$} for the formula
\enquote{$\nicht\,\dist(x,y)\,{\leq}\,r$} and, slightly abusing notation, we
will also call this a \enquote{distance atom}.

Let
$\wmaxar\deff\max(\set{1}\cup\setc{\ar(\weight)}{\weight\in\Weights})$,
i.e., $\wmaxar$ is the maximum arity of the weight symbols in
$\Weights$ (or 1 if $\Weights$ is empty or contains only symbols
of arity 0).
We define the \emph{radius function}
\begin{equation}\label{eq:def:radius_function}
 \locRad{0}{k}{d} \deff \max\set{d,1}
    \quad\text{and}\quad
 \locRad{q{+}1}{k}{d} \deff 4k\cdot \locRad{q}{k{+}\wmaxar}{d}
    \,, ~~
    \text{for all \(q,k,d \in \NN\)}
    .
\end{equation}
Note that
\(\locRad{q}{k}{d} =  4^q \cdot \prod_{i=0}^{q-1} (k{+}i\wmaxar)\cdot \max\set{d,1}\).
Using this function $\myr$,
for $q,k,d\in\NN$, we define a particular set $\Fqkd$
of $\ngFOWplus$ formulas $\phi$ as follows.

For all $k,d\in\NN$, we let
$\Fparam{0}{k}{d}$ be the set of all formulas $\phi$ in $\ngFOWplus$
with $\qr(\phi)=0$,
$\free(\phi)\subseteq \set{x_1,\ldots,x_k}$, and $\md{\phi}\leq d$.
For all $q\in\NN$, we let $\Fparam{q{+}1}{k}{d}$ be the set of all formulas
$\phi$ with $\free(\phi)\subseteq \set{x_1,\ldots,x_k}$ such that
$\phi$ is a Boolean combination of
  \begin{enumerate}[(i)]
    \item
    \label{def:fparam:previous}
      formulas in $\Fparam{q}{k{+}\wmaxar}{d}$ (with free variables
      among $\set{x_1,\ldots,x_k}$),
    \item
    \label{def:fparam:distance_atom}
      distance atoms $\dist(x_i,x_j)\leq \tilde{d}$ with $i,j\in[k]$ and
      $\tilde{d}\leq \locRad{q{+}1}{k}{d}$,
    \item
    \label{def:fparam:existential_with_distance}
      formulas of the form $\exists x_{k+1}\big(\dist(x_i,x_{k+1})\,{\leq}\, \hat{d}
      \land \psi \big)$ with $\psi\in\Fparam{q}{k{+}\wmaxar}{d}$, $i\in[k]$, and $\hat{d}\leq
      \locRad{q{+}1}{k}{d}-\locRad{q}{k{+}\wmaxar}{d}$,
    \item
    \label{def:fparam:existential_plain}
      formulas of the form $\exists x_{k+1}\,\psi$ with
      $\psi\in\Fparam{q}{k{+}\wmaxar}{d}$,
      and
    \item
    \label{def:fparam:aggregation}
      formulas of the form
      $\big(s=\sum\weight(\bar y).\psi\big)$, where
      $\weight\in\Weights$,
      the range \(S\) of \(\weight\) is finite,
      $s\in S$,
      $\ell=\ar(\weight)$,
      $\bar y=(x_{k+1},\ldots,x_{k+\ell})$, and $\psi\in \Fparam{q}{k{+}\wmaxar}{d}$.
  \end{enumerate}

For every formula $\phi$ in $\ngFOWplus$, there exist
$q,k,d\in\NN$ such that $\phi\in\Fparam{q}{k}{d}$;
and
every $\phi\in \Fparam{q}{k}{d}$ has $\qr(\phi)\leq q$,
$\free(\phi)\subseteq\set{x_1,\ldots,x_k}$, and
$\md{\phi}\leq \locRad{q}{k}{d}$.

We define notions similar to the ones used in the literature on Gaifman normal
form for $\FO$, $\FOMOD$, and $\FOWun$
(cf.\ \cite{GaifmanPaper,KuskeSchweikardt_ICALP18,vanBergeremSchweikardt_2021_FOWA}).
A \emph{basic local sentence for
$\Fparam{q}{k}{d}$} (of
 \emph{width} $\ell$, \emph{inner distance} $2r$, \emph{inner radius}
 at most $\tilde{r}$, and
 \emph{inner quantifier rank} $\qr(\lambda)$) is of the
 form~\eqref{eq:intro-bl},
 where \enquote{$\dist(y_i,y_j)\,{>}\,2r$}  is
a distance atom,
and where $\ell\geq 1$, $y_1,\ldots,y_\ell$ are $\ell$
variables, $r\geq \tilde{r}\geq 0$,
and
$\lambda(x_1)\in\Fparam{q{-}1}{k{+}\wmaxar}{d}$ is
\emph{$\tilde{r}$-local} and has only one free variable.
A \emph{local aggregation sentence for $\Fparam{q}{k}{d}$} (of
\emph{inner radius} at most $\tilde{r}$ and
\emph{inner quantifier rank} $\qr(\lambda)$) is of the form
$
\big(
  s = \sum \weight(\bar y).\lambda
\big)
$
where
$\weight\in\Weights$,
$\ell=\ar(\weight)\geq 1$,
the range \(S\) of \(\weight\) is finite,
$s\in S$,
$\bar y=(y_1,\ldots,y_\ell)$ is a tuple of $\ell$ pairwise distinct variables,
$\tilde{r}\in\NN$,
and $\lambda(\bar y)\in \Fparam{q{-}1}{k{+}\wmaxar}{d}$ is
$\tilde{r}$-local
(with  $\free(\lambda)\subseteq\set{y_1,\ldots,y_\ell}$).

A formula in \emph{Gaifman normal form for $\Fparam{q}{k}{d}$}
is a Boolean combination of \emph{local} formulas in
$\Fparam{q}{k}{d}$ and basic local sentences and local aggregation
sentences for $\Fparam{q}{k}{d}$.
Given such a formula $\phi'$, let $\mathcal{S}$ and $\mathcal{L}$ be the sets of
sentences and local formulas, resp., such that $\phi'$ is a Boolean combination
of the formulas in $\mathcal{S}\cup \mathcal{L}$. The
\emph{outer quantifier rank} of $\phi'$ is the maximum of the
quantifier ranks of the formulas in $\mathcal{L}$; and we say that $\phi'$ has
\emph{outer radius at most $r$} if every formula in $\mathcal{L}$ is $r$-local.
The
\emph{inner quantifier rank} (\emph{radius}, \emph{distance},
\emph{width}) of $\phi'$ is the maximum of the inner
quantifier ranks (radii and, for basic local sentences,
distances, widths) of the sentences in $\mathcal{S}$.
The main theorem of this section is
the following.

\begin{theorem}\label{thm:GNF_ngFOWplus}
Let $q,k,d\in\NN$. Every $\phi\in\Fqkd$ is equivalent to a
formula $\phi'$ in Gaifman normal form for $\Fqkd$.
Furthermore, $\free(\phi')=\free(\phi)$,
the
outer (inner) quantifier rank of $\phi'$
is at most $\qr(\phi)$ (resp.\ $\qr(\phi){-}1$),
the outer (inner) radius of $\phi'$
is at most $\locRad{q}{k}{d}$
(resp.\ $\locRad{q{-}1}{k{+}\wmaxar}{d}$),
and the width (inner distance)
is at most $k{+}1{+}(q{-}1)\wmaxar$
(resp.\ $\locRad{q}{k}{d}+2\locRad{q{-}1}{k{+}\wmaxar}{d}$).
Moreover,
such a $\phi'$ can be computed from $\phi,q,k,d$.
\end{theorem}

This \enquote{Gaifman normal form theorem} is \emph{rank preserving} in the
sense that the \emph{local formulas} of the normal form formula $\phi'$ belong
to the same class $\Fqkd$ as the original formula~$\phi$, and
their quantifier rank is at most $\qr(\phi)$.
Considering \cref{ex:FOplus,ex:FOMODplus}, this
yields corollaries that provide rank-preserving
Gaifman normal forms for $\FOplus$ and for $\FOMODplus{M}$
(see \cref{sec:rpGNFforFOandFOMOD} for details). Using the normal form
for $\FOplus$, in \cref{section:nowhere-dense}, we give a simplified proof
of the algorithmic meta-theorem
of~\cite{GroheKreutzerSiebertz_2017_NowhereDense}.

The remainder of this section is devoted to the proof of
\cref{thm:GNF_ngFOWplus}.
Similarly as in
\cite{GaifmanPaper,KuskeSchweikardt_ICALP18,vanBergeremSchweikardt_2021_FOWA},
we proceed by induction on $q$, and handling atomic formulas and
Boolean combinations of formulas is trivial.
Handling formulas built using the aggregation quantifier
is achieved similarly as modulo-counting
and aggregation quantifiers are handled in
\cite{KuskeSchweikardt_ICALP18,vanBergeremSchweikardt_2021_FOWA}.
Handling formulas built using existential quantification
is the most intricate case,
and in order to avoid an increase of the quantifier rank,
we contribute a new construction that differs significantly
from that of Gaifman's original proof.
Independently of our work,
Dreier and Toru\'{n}czyk
have recently published a note
\cite{DreierT26a} that provides a \emph{rank
preserving locality theorem} for first-order logic. The combinatorial
core of their treatment of the $\exists$ quantifier is similar to
ours, but their normal form is different, with a different radius
function, different sentences, and without considering weighted structures or aggregation quantifiers.

In fact, we prove a more refined version of
\cref{thm:GNF_ngFOWplus} that is based on particular sets
\(\Lqkd\) and \(\Sqkd\) of specific local formulas in \(\Fqkd\) and
specific $\ngFOWplus$ sentences,
respectively.
We proceed with defining these sets.

For all $q,d\in\NN$,
we let $\loc{q}{0}{d}\deff\emptyset$.
For all $k,d\in\NN$ with $k\geq 1$ we let
$\loc{0}{k}{d}\deff\Fparam{0}{k}{d}$, and for all $q\in \NN$ we let
$\loc{q{+}1}{k}{d}$ be the set of all formulas $\phi$ with
$\free(\phi)\subseteq\set{x_1,\ldots,x_k}$ such that $\phi$ is a
Boolean combination of
\begin{enumerate}
\item
\label{itm:def:loc:kind_previous_rank}
formulas in $\loc{q}{k{+}\wmaxar}{d}$
(with free variables among $\set{x_1,\ldots,x_k}$),
\item
\label{itm:def:loc:kind_distance_atom}
distance atoms $\dist(x_i,x_j)\leq \tilde{d}$ \ with $i,j\in[k]$ and $\tilde{d}\leq \locRad{q{+}1}{k}{d}$,
and
\item
\label{itm:def:loc:kind_recursive}
  formulas of the following form, where
  $r'\deff \locRad{q}{k{+}\wmaxar}{d}$ and $\hat{r}\deff \locRad{q{+}1}{k}{d}-r'$:
\begin{enumerate}
\item
\label{itm:def:loc:kind_near}
      \(
        \exists x_{k+1}
        \,
        \paren[\big]{
          \
          \dist(x_i,x_{k+1}) \,{\leq}\,
          \hat{d}
          \;\undp\;
          \lambda
          \
        }
      \)
      \ with $i\in[k]$, \(\lambda \in \loc{q}{k{+}\wmaxar}{d}\), and
      $\hat{d}\leq\hat{r}$,
\item
\label{itm:def:loc:kind_weight_aggregation}
$\displaystyle
\Big(
  s =
  \sum \weight(\bar y).
  \big(\;
     \lambda \; \und \!\!
     \Oder_{i\in I, j\in[\ell]}
 \!\!        \dist(x_i,y_j)\,{\leq}\, r'
  \,\big)
\Big)$,
  where $\weight\in\Weights$,
  the range \(S\) of \(\weight\) is finite,
  $s\in S$, $\ell=\ar(\weight)\geq 1$,  $\bar
  y=(y_1,\ldots,y_\ell)=(x_{k+1},\ldots,x_{k+\ell})$,
  $\emptyset\neq I\subseteq [k]$
  and
  $\lambda\in\loc{q}{k{+}\wmaxar}{d}$.
\end{enumerate}
\end{enumerate}

\begin{restatable}{lemma}{LocalityLemma}\label{lemma:LocalityLemma}
  For all \(q, k, d \in \NN\),
  it holds that \(\Lqkd \subseteq \Fqkd\),
  and every formula \(\phi \in \Lqkd\) is \(\locRad{q}{k}{d}\)-local.
\end{restatable}
The proof is by a straightforward induction
(cf.\ \cref{appendix:LocalityLemma};
the locality of formulas of the form \ref{itm:def:loc:kind_weight_aggregation}
relies on the \emph{locality condition} of weighted structures,
which ensures that if $\weight(\bar y)\neq 0_S$ and one of the $y_j$ is close to $x_i$,
then \emph{all} the $y_j$ are close to $x_i$).

For all $k,d\in\NN$, we let $\sent{0}{k}{d}$ be the set of all atomic
formulas in $\Fparam{0}{k}{d}$ that are sentences. Hence,
\(\sent{0}{k}{d} = \sent{0}{0}{0}\), which consists of sentences of the form
$\true$, $\false$,
$R()$ for 0-ary $R\in\sigma$,
$\big(s=\weight()\big)$ for 0-ary $\weight\in\Weights$,
and $\Pred(t_1,\ldots,t_m)$ for $\Pred\in\PP$
and terms
$t_1,\ldots,t_m$ with $\bigcup_{i\in[m]}\free(t_i)=\emptyset$.
We say that  these sentences have \emph{inner quantifier rank} -1 and
\emph{outer quantifier rank} 0.
For all $q\in\NN$ we let $\sent{q{+}1}{k}{d}$ consist of all
\begin{itemize}
\item
  sentences in $\sent{q}{k{+}\wmaxar}{d}$,
\item
  sentences of the form \
  \(\displaystyle
    \exists y_1\cdots\exists y_\ell\ \Big(
      \Und_{1 \leq j < j' \leq \ell} \dist(y_j,y_{j'}) > 2(2c{+}1)r' \undp
      \Und_{j \in [\ell]} \lambda(y_j)
    \ \Big)
  \)
  with \(\ell \in [k{+}1]\), $\ell$
  variables $y_1,\ldots,y_\ell$,
  \(c \in [0,k]\),
  \(r' = \locRad{q}{k{+}\wmaxar}{d}\),
  and \(\lambda \paren{x_1} \in \loc{q}{k{+}\wmaxar}{d}\),
\item
  sentences of the form
$
\big(
  s = \sum \weight(y_1,\ldots,y_\ell).\lambda
\big)
$
where
$\weight\in\Weights$,
$\ell=\ar(\weight)\geq 1$,
the range \(S\) of \(\weight\) is finite,
$s\in S$,
$y_1,\ldots,y_\ell$ are $\ell$
distinct variables,
and
$\lambda(y_1,\ldots,y_\ell)\in\loc{q}{k{+}\wmaxar}{d}$.
\end{itemize}
By using \cref{lemma:LocalityLemma}, an easy induction on $q$ shows
that every formula in $\sent{q}{k}{d}$ is an atomic sentence in
$\sent{0}{0}{0}$
or a basic local sentence or a local aggregation sentence for
$\Fqkd$ of inner radius at most $\locRad{q{-}1}{k{+}\wmaxar}{d}$ and,
moreover, the basic local sentences in
$\sent{q}{k}{d}$ have
width (inner distance)
at most $k{+}1{+}(q{-}1)\wmaxar$
(resp.~$\locRad{q}{k}{d}+2\locRad{q{-}1}{k{+}\wmaxar}{d}$);
cf.~\cref{appendix:PropertiesOfSqkd}.
We can now state the refinement of Theorem~\ref{thm:GNF_ngFOWplus}.

\begin{restatable}{theorem}{GNFngFOWplusRefined}
 \label{thm:GNF_ngFOWplus_refined}
Let $q,k,d\in\NN$. Every $\phi\in \Fqkd$
is equivalent to a
formula $\phi'$
that is a Boolean combination of formulas in
$\Sqkd\cup\Lqkd$. Furthermore, $\free(\phi')=\free(\phi)$,
the
outer (inner) quantifier rank of $\phi'$ is at most $\qr(\phi)$
(resp.\ $\qr(\phi){-}1$),
and there is an algorithm that computes such a $\phi'$ from $\phi,q,k,d$.
\end{restatable}

Theorem~\ref{thm:GNF_ngFOWplus} is an immediate
consequence of Theorem~\ref{thm:GNF_ngFOWplus_refined}
(for this, use
\cref{lemma:LocalityLemma}, the properties of $\sent{q}{k}{d}$ stated above, and note that
every $\chi\in\Sqkd$ that is an atomic sentence in
$\sent{0}{0}{0}$
is a local formula, because $\chi(z)$ is
0-local around $z$ for every variable $z$).

The rest of this section is devoted to the proof of Theorem~\ref{thm:GNF_ngFOWplus_refined}.
A crucial tool for our proof is the next lemma that provides
specific \emph{Feferman--Vaught decompositions} for formulas in $\Lqkd$.
For this, we need further notation.

For \(\vec{z} = (z_1,\ldots,z_{k})\) and  \(K\subseteq\intsUpTo{k}\)
we write \(\projTup{z}{K}\) to denote the projection of \(\vec{z}\) on the indices in \(K\),
i.e., if \(K = \set{i_1,\ldots,i_\ell}\) with \(i_1 < \cdots < i_\ell\)
then \(\projTup{z}{K} = (z_{i_1},\ldots,z_{i_\ell})\).
For $I\subseteq K$ and \(d \in \NN\) we let
$\dist(\projTup{z}{I};\projTup{z}{K\setminus I}) >d$ be the conjunction
of the atoms $\dist(z_i,z_j) > d$ for all $i\in I$ and $j\in
K\setminus I$.
We use similar notations for
$(\sigma,\Weights)$-structures \(\A\) and tuples
\(\ta=(a_1,\ldots,a_{k}) \in A^{k}\) and write
$\dist^{\A}(\ta_I;\ta_{K\setminus I})>d$
to express that $\dist^{\A}(a_i,a_j)>d$
for all $i\in I$ and $j\in K\setminus I$.

Recall that for a formula $\phi$ and a tuple $\ty=(y_1,\ldots,y_\ell)$ of
variables, $\phi(\ty)$ indicates that
$\free(\phi)\subseteq \set{y_1,\ldots,y_\ell}$.
Our Feferman--Vaught decomposition lemma reads as follows.

\begin{restatable}{lemma}{FV}
  \label{lemma:fv}\label{lem:fv_decomp}
  Let \(q,\myk,d \in \NN\) with \(\myk \geq 2\).
Let \(\vec{x} \deff (x_1,\ldots,x_{\myk})\),
\(r' \deff \locRad{q}{\myk}{d}\),
\(K \subseteq \intsUpTo{\myk}\),
and
  \(I \subseteq K\) with \(\emptyset \neq I \neq K\).
For every \(\lambda \paren{\projTup{x}{K}} \in \loc{q}{\myk}{d}\),
  there exists a finite and non-empty set \(\Delta\) of pairs of formulas
  \(\big(\alpha\paren{\projTup{x}{I}},\beta\paren{\projTup{x}{K \setminus I}}\big)\)
  such that \(\alpha\paren{\projTup{x}{I}}\) and
  \(\beta\paren{\projTup{x}{K \setminus I}}\) belong to \(\loc{q}{\myk}{d}\),
  their
  quantifier rank is at most $\qr(\lambda)$,
  and
  \begin{equation}\label{eq:lem:fv_decomp}
    \lambda(\tx_K)
    \; \und\;
    \dist(\tx_I;\tx_{K \setminus I}) \,{>}\, r'
    \ \ \equiv \ \Oder_{(\alpha,\beta) \in \Delta}
   \paren[\Big]{
      \alpha \paren{\projTup{x}{I}}
      \und
      \beta \paren{\projTup{x}{K \setminus I}}
    }
     \undp
    \dist(\projTup{x}{I};\projTup{x}{K \setminus I}) \,{>}\, r'
    .
  \end{equation}
  Moreover, we can ensure that the first entries in $\Delta$ are \emph{mutually exclusive},
  i.e., for every two distinct $(\alpha,\beta)$ and $(\alpha',\beta')$
  in $\Delta$, the formula $(\alpha\und\alpha')$ is unsatisfiable.
  Furthermore, there is an algorithm that computes \(\Delta\)
  on input \(q, \myk, d, K, I, \lambda\).
\end{restatable}

The proof of \cref{lemma:fv} is provided in \cref{sec:fv};
it is based on the same techniques as previous Feferman--Vaught
decomposition results (cf.\ e.g.\
\cite{FefermanVaught,KuskeSchweikardt_ICALP18,vanBergeremSchweikardt_2021_FOWA}),
but special care needs to be taken to check that the set $\loc{q}{\myk}{d}$ and
the radius function $\locRad{q}{\myk}{d}$ have the properties needed
for the proof.
Finally, we are ready for proving \cref{thm:GNF_ngFOWplus_refined}.

\begin{proof}[Proof of \cref{thm:GNF_ngFOWplus_refined}]
We proceed by induction on \(q\).
Here, we focus on the main issues;
a detailed proof is provided in \cref{appendix:DetailedGNFproof}.
The induction base with \(q = 0\) is trivial.
For the induction step from \(q\) to \(q{+}1\),
consider \(q, k, d \in \NN\)
and a formula \(\phi \in \Fparam{q{+}1}{k}{d}\);
in particular, \(\free(\phi) \subseteq \set{x_1, \dots, x_k}\).
The two crucial cases are those where
\(\phi\) is of the form \(\exists x_{k+1}\,\phi_1\)
and where
\(\phi\) is of the form \(\bigl(s = \sum\weight(\ty).\phi_1\bigr)\),
for some \(\phi_1 \in \Fparam{q}{k'}{d}\) with \(k' \deff k{+}\wmaxar\).
The latter case can be handled similarly as the analogous case in the
proof of~\mbox{\cite[Theorem~4.6]{vanBergeremSchweikardt_2021_FOWA}},
but we rely on \cref{lemma:fv} instead of the Feferman--Vaught decompositions
of~\cite{vanBergeremSchweikardt_2021_FOWA}.
The case \(\phi = \exists x_{k+1}\,\phi_1\) is the most intricate case;
in order to avoid an increase of the (outer/inner) quantifier rank,
we cannot proceed similarly as previous proofs of
Gaifman normal form theorems in
\cite{GaifmanPaper,EF-FMT,KuskeSchweikardt_ICALP18,vanBergeremSchweikardt_2021_FOWA},
but we have to come up with a different construction.
For the remainder of this proof, we focus on this case.

By applying the induction hypothesis to \(\phi_1\)
and  performing standard transformations,
we can reduce to the following task:
for a \(\lambda \in \loc{q}{k'}{d}\) with
\(\free(\lambda) \subseteq \free(\phi) \cup \set{x_{k+1}}\)
and \(\qr(\lambda) \leq \qr(\phi){-}1\),
transfer \(\exists x_{k+1}\,\lambda\) into an equivalent
Boolean combination of formulas in
\(\sent{q{+}1}{k}{d} \cup \loc{q{+}1}{k}{d}\)
with outer (inner) quantifier rank
\(\leq \qr(\phi)\) (resp.~\(\qr(\phi)-1\)).

If \(\free(\phi) = \emptyset\),
then \(\exists x_{k+1}\,\lambda\) is a sentence in \(\sent{q{+}1}{k}{d}\),
and we are done.
Thus, consider the case that \(\free(\phi) \neq \emptyset\),
and hence, \(k \geq 1\).
Let \(\tx \deff (x_1, \dots, x_{k'})\),
\(I \deff \setc{i \in [k]}{x_i \in \free(\phi)}\), and
\(r' \deff \locRad{q}{k'}{d}\).
Clearly, \(\exists x_{k+1}\,\lambda\) is equivalent to
\(\psi \lor \Lor_{i \in I} \vartheta_i\)
with
\(\vartheta_i \deff \exists x_{k+1}\bigl(\distAtom{x_i}{x_{k+1}}{r'}
\landp \lambda\bigr)\) for all \(i\in I\),
and
\(\psi \deff \exists x_{k+1}\bigl(\dist(\tx_I; x_{k+1}) \,{>}\, r'
\landp \lambda\bigr)\).
Since the \(\vartheta_i\) belong to \(\loc{q{+}1}{k}{d}\),
all that remains to be done is to transfer \(\psi\) into an
equivalent Boolean combination of formulas in
\(\sent{q{+}1}{k}{d} \cup \loc{q{+}1}{k}{d}\)
with outer (inner) quantifier rank
\(\leq \qr(\phi)\) (resp.~\(\qr(\phi)-1\)).
Applying \cref{lemma:fv} with \(K \deff I \cup \set{k{+}1}\) to \(\lambda\)
yields a finite, non-empty set \(\Delta\),
which we can utilize to reduce to the following task:
for a formula
\(\beta \in \loc{q}{k'}{d}\) with \(\free(\beta) = \set{x_{k+1}}\)
and \(\qr(\beta) \leq \qr(\phi){-}1\),
transform the formula
\[
  \mu(\tx_I)
  ~~ \deff ~~
  \exists x_{k+1} \bigl(
    \dist(\tx_I; x_{k+1}) \,{>}\, r' \landp \beta(x_{k+1})
  \big)
\]
into an equivalent formula \(\mu'(\tx_I)\)
that is a Boolean combination of formulas in
\(\sent{q{+}1}{k}{d} \cup \loc{q{+}1}{k}{d}\)
with outer (inner) quantifier rank at most
\(\qr(\phi)\) (resp.~\(\qr(\phi)-1\)).
Let \(\tilde{r} \deff \locRad{q{+}1}{k}{d}\)
and \(\hat{r} \deff \tilde{r} - r'\).
For constructing \(\mu'(\tx_I)\),
we use the following formulas:
\begin{itemize}
  \item
    \(\displaystyle
    \delta_i(\tx_I) \deff \exists x_{k+1} \big(
      \distAtom{x_i}{x_{k+1}}{\hat{r}}
      \landp
      \dist(\tx_I; x_{k+1}) \,{>}\, r'
      \landp
      \beta(x_{k+1})
    \big)\)
    for all \(i \in I\),
  \item
    \(\displaystyle
    \psi_\ell^{(c)}(\tx_I) \deff
    \Lor_{\substack{J \subseteq I,\\ \abs{J} = \ell}}
    \paren[\Big]{
      \Land_{\substack{j,j' \in J,\\ j \neq j'}}
      \!\!\!\dist(x_j,x_{j'}) \,{>}\, 4cr'
      \landp
      \Land_{j \in J} \exists x_{k+1} \,
      \paren[\big]{
        \distAtom{x_j}{x_{k+1}}{r'} \landp \beta(x_{k+1})
      }
    }\)
    for all \(\ell, c \in [\abs{I}]\),
    and
  \item
    \(\displaystyle
    \xi_\ell^{(c)} \deff \exists y_1 \cdots \exists y_\ell \Big(\!
      \Land_{1 \leq j < j' \leq \ell}\!\!\!
      \dist(y_j, y_{j'}) \,{>}\, 2(2c{+}1)r'
      \land
      \Land_{j \in [\ell]} \!\beta(y_j)
    \Big)\)
    for all \(\ell \in [\abs{I}{+}1]\) and \(c \in [0,\abs{I}]\).
\end{itemize}

\noindent
Note that
\(\delta_i(\tx_I) \in \loc{q{+}1}{k}{d}\) for all \(i \in I\),
that \(\psi_\ell^{(c)}(\tx_I) \in \loc{q{+}1}{k}{d}\)
for all \(\ell, c \in [\abs{I}]\),
and that \(\xi_\ell^{(c)} \in \sent{q{+}1}{k}{d}\)
for all \(\ell \in [\abs{I}{+}1]\) and \(c \in [0,\abs{I}]\).
We let
\[
  \mu'(\tx_I)
  \quad \deff \quad
  \xi^{(0)}_{\abs{I}+1} \ \lor \ \; \Lor_{c,\ell \in [\abs{I}]} \ \Big(
    \xi^{(c)}_\ell \landp \neg\,\xi^{(c{-}1)}_{\ell+1}
    \landp
    \big(\Lor_{i \in I} \delta_i(\tx_I)
    \lorp \neg\,\psi^{(c)}_\ell(\tx_I)\big)
  \Big).
\]
Clearly, \(\mu'\) is a Boolean combination of formulas in
\(\sent{q{+}1}{k}{d} \cup \loc{q{+}1}{k}{d}\)
with
\(\free(\mu') = \setc{x_i}{i \in I} = \free(\phi)\),
and \(\mu'\) has outer (inner) quantifier rank
\(\leq \qr(\phi)\) (resp.~\(\qr(\phi)-1\)).
Next, we prove that \(\mu'(\tx_I)\) is equivalent to \(\mu(\tx_I)\).
Consider a \((\sigma,\Weights)\)-structure \(\A\) and a tuple \(\ta \in A^k\).
We have to show that
\(\A \models \eval{\mu'}{\ta_I} \iff \A \models \eval{\mu}{\ta_I}\).
We will use the following observation,
which is an immediate consequence of a double use of the triangle inequality.
\begin{restatable}{claim}{distanceByDoubleTriangle}
  \label{claim:distance_by_double_triangle}
Let \(s, t \in \NN\),
  and let \(c_1, c_2, d_1, d_2 \in A\)
  such that
  \(\distS{\A}{d_1, d_2} \,{>}\, s{+}2t\)
  and
  \(\distS{\A}{c_j, d_j} \,{\leq}\, t\) for all \(j \in \set{1,2}\).
  Then we have \(\distS{\A}{c_1, c_2} \,{>}\, s\).
\end{restatable}

We now focus on proving that
\(\A \models \eval{\mu'}{\ta_I} \iff \A \models \eval{\mu}{\ta_I}\).
Throughout the remainder of this proof,
we let \(B \deff \setc{b \in A}{\A \models \eval{\beta}{b}}\).

\proofsubparagraph{Case 1:}
\(\A \models \xi^{(0)}_{\abs{I}+1}\).
Then, \(\A \models \eval{\mu'}{\ta_I}\).
We have to show that \(\A \models \eval{\mu}{\ta_I}\).

Since \(\A \models \xi^{(0)}_{\abs{I}+1}\),
there exist \(\enud{b_1}{b_{\abs{I}+1}} \in B\) of pairwise distance \(> 2r'\).
Let \(\pi\) be a mapping that maps every
\(b \in \set{\enud{b_1}{b_{\abs{I}+1}}}\)
to an element \(a \in \setc{a_i}{i \in I}\) that is closest to \(b\),
that is, an element that minimises \(\distS{\A}{b,a}\),
where ties are broken arbitrarily.
Then, by the pigeonhole principle,
there are two distinct \(j, j' \in \intsUpTo{\abs{I}{+}1}\)
with \(\pi(b_j) = \pi(b_{j'}) = a_i\) for some \(i \in I\).
Since \(\distS{\A}{b_j, b_{j'}} > 2r'\),
by the triangle inequality,
we know that \(\distS{\A}{a_i, b_j} > r'\)
or \(\distS{\A}{a_i, b_{j'}} > r'\).
Say, w.l.o.g., \(\distS{\A}{a_i, b_j} > r'\).
By the choice of \(\pi\),
we obtain that \(\distS{\A}{\ta_I; b_j} > r'\).
This \(b_j \in B\) witnesses that \(\A \models \eval{\mu}{\ta}\).
 
\proofsubparagraph{Case 2:}
\(\A \not\models \xi^{(0)}_1\).
Then, there does not exist any \(b \in A\) such that
\(\A \models \eval{\beta}{b}\), so \(B = \emptyset\).
Hence, \(\A \not\models \eval{\mu}{\ta_I}\).
Furthermore, \(\A \not\models \xi^{(0)}_{\abs{I}+1}\) and, moreover,
for all \(c, \ell \in \intsUpTo{\abs{I}}\),
we have that \(\A \not\models \xi^{(c)}_\ell\).
Thus, \(\A \not\models \eval{\mu'}{\ta_I}\),
so we have that \(\A \not\models \eval{\mu'}{\ta_I}\)
and \(\A \not\models \eval{\mu}{\ta_I}\).
 
\proofsubparagraph{Case 3:}
We are neither in Case~1 nor in Case~2.
Then, we have that
\(\A \models \xi^{(0)}_1 \land \neg\,\xi^{(0)}_{\abs{I}+1}\),
so \(B \neq \emptyset\),
but \(B\) does not contain \(\abs{I}{+}1\) elements
of pairwise distance \(> 2r'\).

\begin{restatable}{claim}{rpGaifmanScatteredThreshold}
  \label{claim:proof:rp_gaifman:scattered_threshold}
  There exist \(c, \ell \in \intsUpTo{\abs{I}}\)
  such that \(\A \models \xi^{(c)}_\ell \land \neg\,\xi^{(c-1)}_{\ell+1}\).\end{restatable}
\begin{claimproof}
  For every \(c \in [0,\abs{I}]\),
  let \(\ell^{(c)} \in \NN \cup \set{\infty}\) be maximal
  such that there exists a set \(X \subseteq B\) of \(\ell^{\paren{c}}\) elements
  of pairwise distance \(> 2(2c{+}1)r'\).
From \(B \neq \emptyset\),
  we obtain that \(\ell^{\paren{c}} \geq 1\) for each \(c \in \intsUpTo{0,\abs{I}}\).
Since $B$ does not contain $\abs{I}{+}1$ elements of pairwise distance
  $>2r'$, we know that
  \(\ell^{\paren{c}} \leq \abs{I}\) for each \(c \in \intsUpTo{0,\abs{I}}\).
  Furthermore, for all \(c \geq 1\), we have \(\ell^{\paren{c}} \leq \ell^{(c-1)}\),
  because \(2(2c{+}1)r' \geq 2(2(c{-}1){+}1)r'\).
Hence, we have
  $\abs{I} \geq \ell^{\paren{0}} \geq \ell^{\paren{1}}  \geq
    \cdots \geq  \ell^{\paren{\abs{I}}} \geq 1$.
  By the pigeonhole principle,
  there exists a \(c \in \intsUpTo{\abs{I}}\)
  such that \(\ell^{(c-1)} = \ell^{\paren{c}}\).
We choose such a \(c\) and let \(\ell \deff \ell^{\paren{c}}\).
  Since \(\ell = \ell^{\paren{c}}\), we have that \(\A \models \xi^{(c)}_\ell\).
  Since \(\ell = \ell^{(c-1)}\) and \(\ell^{(c-1)}\) is maximal,
  we have that  \(\A \models \neg\,\xi^{(c-1)}_{\ell+1}\).
\end{claimproof}
 
Consider arbitrary \(c, \ell \in \intsUpTo{\abs{I}}\)
such that \(\A \models \xi^{(c)}_\ell \land \neg\,\xi^{(c-1)}_{\ell+1}\),
which exist due to \cref{claim:proof:rp_gaifman:scattered_threshold}.

\proofsubparagraph{Case 3.1:}
\(\A \models \eval{\delta_i}{\ta_I}\) for some \(i \in I\).
This implies that \(\A \models \eval{\mu'}{\ta_I}\)
and \(\A \models \eval{\mu}{\ta_I}\).
 
\proofsubparagraph{Case 3.2:}
\(\A \models \eval{\neg\,\psi^{(c)}_\ell}{\ta_I}\).
Then \(\A \models \eval{\mu'}{\ta_I}\).
We have to show that \(\A \models \eval{\mu}{\ta_I}\).

From \(\A \models \xi^{(c)}_\ell\),
we know that there exist \(\enud{b_1}{b_\ell} \in B\)
of pairwise distance \(> 2(2c{+}1)r'\).
If there is a \(\nu \in \intsUpTo{\ell}\)
such that \(\distS{\A}{\ta_I; b_\nu} \,{>}\, r'\),
then \(b_\nu\) serves as a witness certifying that
\(\A \models \eval{\mu}{\ta_I}\).
\emph{For contradiction},
assume that, for each \(\nu \in \intsUpTo{\ell}\),
we do \emph{not} have
\(\distS{\A}{\ta_I;b_\nu} \,{>}\, r'\).
Then, for every \(\nu \in \intsUpTo{\ell}\),
there is an \(i(\nu) \in I\)
such that \(\distS{\A}{a_{i(\nu)}, b_\nu} \,{\leq}\, r'\).
 
\begin{restatable}{claim}{rpGaifmanPwFarNearElements}
  \label{claim:proof:rp_gaifman:pw_far_near_elements}
  For all distinct \(\nu,\nu' \in \intsUpTo{\ell}\),
  we have \(\distS{\A}{a_{i(\nu)}, a_{i(\nu')}} > 4cr'\).
\end{restatable}
\begin{claimproof}
  For distinct \(\nu, \nu' \in [\ell]\),
  apply \cref{claim:distance_by_double_triangle}
  with \(s = 4cr'\),
  \(t = r'\),
  and \((c_1, c_2, d_1, d_2) = (a_{i(\nu)}, a_{i(\nu')}, b_\nu, b_{\nu'})\).
  This yields
  \(\distS{\A}{c_1, c_2} \,{>}\, s\),
  so
  \(\distS{\A}{a_{i(\nu)}, a_{i(\nu')}} \,{>}\, 4cr'\).
\end{claimproof}
 
\Cref{claim:proof:rp_gaifman:pw_far_near_elements} implies that
\(i(\nu) \neq i(\nu')\) holds for all distinct \(\nu, \nu' \in [\ell]\).
Thus, for \(J \deff \setc{i(\nu)}{\nu\in[\ell]}\),
we have \(\abs{J} = \ell\).
Furthermore, for each \(\nu \in \intsUpTo{\ell}\),
the element \(b_\nu\) serves as a witness certifying that
\(
(\A,\ta_I) \models \exists x_{k+1}\,\paren[\big]{
  \distAtom{x_{i(\nu)}}{x_{k+1}}{r'}
  \landp
  \beta(x_{k+1})
}
\).
Combining this with
\cref{claim:proof:rp_gaifman:pw_far_near_elements} proves
that \(J\) serves as a witness certifying that
\(\A \models \eval{\psi^{(c)}_\ell}{\ta_I}\).
However, recall that we currently consider the case that
\(\A \models \neg\,\eval{\psi^{(c)}_\ell}{\ta_I}\),
a contradiction!
Thus, there must exist a \(\nu \in \intsUpTo{\ell}\)
such that \(\distS{\A}{\ta_I; b_\nu} \,{>}\, r'\).
This \(b_\nu\) serves as a witness certifying that
\(\A \models \eval{\mu}{\ta_I}\).
 
\proofsubparagraph{Case 3.3:}
For all \(c, \ell \in \intsUpTo{\abs{I}}\)
with \(\A \models \xi^{(c)}_\ell \land \neg\,\xi^{(c-1)}_{\ell+1}\),
we have \(\A \not\models \eval{\delta_i}{\ta_I}\) for all \(i \in I\),
and \(\A \models \eval{\psi^{(c)}_\ell}{\ta_I}\).
This is the only remaining case.
Since we are still in Case~3,
we have \(\A \not\models \eval{\mu'}{\ta_I}\).
In the following,
we show that also \(\A \not\models \eval{\mu}{\ta_I}\).

Let us fix \(c, \ell \in \intsUpTo{\abs{I}}\)
according to \cref{claim:proof:rp_gaifman:scattered_threshold}.
Thus, we have \(\A \models \xi^{(c)}_\ell \land \neg\,\xi^{(c-1)}_{\ell+1}\)
and \(\A \not\models \eval{\delta_i}{\ta}\) for all \(i \in I\),
and \(\A \models \eval{\psi^{(c)}_\ell}{\ta}\).
\emph{For contradiction}, assume that \(\A \models \eval{\mu}{\ta_I}\).
Then there exists a \(b_0 \in B\) such that \(\dist^{\A}(\ta_I; b_0) > r'\).
Since, for all \(i \in I\),
we have \(\A \not\models \delta_i(\ta_I)\),
we obtain
\(\distS{\A}{\ta_I; b_0} > \hat{r}\).
From \(\A \models \eval{\psi^{(c)}_\ell}{\ta_I}\),
we obtain that there exists a set \(J \subseteq I\)
with \(\abs{J} = \ell\)
such that
\(\distS{\A}{a_j, a_{j'}}
  >
  4cr'
\)
holds for all distinct \(j, j' \in J\);
and for each \(j \in J\),
there exists a \(b_j \in B\) with
\(\distS{\A}{a_j, b_j}
  \leq
  r'
\).
We would like to use the elements \(b_0, (b_j)_{j \in J}\) as witnesses
certifying that \(\A \models \xi^{(c-1)}_{\ell+1}\).
All that remains to be done
is to show that they have pairwise distance \(> 2(2(c{-}1){+}1)r'\).
Note that \(2(2(c{-}1){+}1)r' = 4cr'-2r'\).

For distinct \(j,j' \in J\),
we apply \cref{claim:distance_by_double_triangle}
with \(s = 4cr'-2r'\),
\(t=r'\),
and \((c_1, c_2, d_1, d_2) = (b_j, b_{j'}, a_j, a_{j'})\).
This yields
\(\distS{\A}{c_1, c_2} > s\),
so \(\distS{\A}{b_j, b_{j'}} > 4cr'-2r'\).

All that remains to be done is to show that
\(\distS{\A}{b_j, b_0} > 4cr'-2r'\) for all \(j \in J\).
Consider a \(j \in J\).
We have
\(
  \hat{r}
  <
  \distS{\A}{a_j, b_0}
  \leq
  \distS{\A}{a_j, b_j} + \distS{\A}{b_j, b_0}
  \leq
  r' + \distS{\A}{b_j, b_0}
\).
Thus,
\(
  \distS{\A}{b_j, b_0}
  >
  \hat{r} - r'
\).
Furthermore, by our choice of \(\tilde{r}, \hat{r}, r'\),
we have
\(\hat{r} - r' = \tilde{r} - 2r' = \locRad{q{+}1}{k}{d} - 2r'\),
which, by \eqref{eq:def:radius_function},
is equal to \(4kr' - 2r'\),
and this is \(\geq 4cr'-2r'\),
since \(c \leq \abs{I} \leq k\).

Finally, we have shown that \(b_0, (b_j)_{j \in J}\) serve as witnesses
certifying that \(\A \models \xi^{(c-1)}_{\ell+1}\).
However, this is a contradiction to the fact that
\(\A \models \neg\,\xi^{(c-1)}_{\ell+1}\).
Therefore, we obtain that \(\A \not \models\eval{\mu}{\ta_I}\).
This completes the proof that \(\mu'(\tx_I)\)
is equivalent to \(\mu(\tx_I)\).
\end{proof}
  \section{Applications and Outlook}
\label{sec:conclusion}
By \cref{ex:FOplus,ex:FOMODplus}, the logics $\FOplus$ and $\FOMODplus{M}$
can be viewed as special cases of the logic $\ngFOWplus$.
Consequently, we can translate \cref{thm:GNF_ngFOWplus} into
results that provide \emph{rank-preserving} versions of the known Gaifman
normal forms for $\FO$ and $\FOMOD$;
see \cref{cor:GNF_FOplus,cor:GNF_FOMODplus}
in \cref{sec:rpGNFforFOandFOMOD} for details.
To demonstrate the use of our rank-preserving normal form for
$\FOplus$, we apply it to obtain a
simplified proof of the algorithmic meta-theorem
of \cite{GroheKreutzerSiebertz_2017_NowhereDense} for
$\FO$ model checking on nowhere dense classes; see
\cref{section:nowhere-dense}.
We believe that,
by using our rank-preserving normal form for $\FOMODplus{M}$,
this proof can be generalised to obtain an analogous result
for $\FOMOD$ model checking on nowhere dense classes.

Canonical next steps are to first try to use \cref{thm:GNF_ngFOWplus}
to achieve a model-checking result
for (a suitable fragment of) $\ngFOWplus$ on nowhere dense classes, and then
try to combine this with methods of
\cite{GroheSchweikardt_2018_FOC1,vanBergeremSchweikardt_2021_FOWA}
to lift it to a suitable fragment of the more general weight
aggregation logic $\FOWA$ of \cite{vanBergeremSchweikardt_2021_FOWA}.
This research agenda, however, comes with a \emph{caveat}:
the next example shows
that $\FO$ definable problems over
arbitrary graphs can be translated into $\ngFOWplus$ definable problems
over vertex-weighted trees of height $1$, provided that the
vertex weights range over $\ZZ$,
and $\PP$ contains a predicate for checking if two such
weights are equal.
This shows that for this particular choice of $\SC,\Weights,\PP$, the
\(\ngFOWplus\) model-checking problem is
\(\AWstar\)-hard
already on vertex-weighted trees of height $1$. This implies that,
in order to achieve efficient model checking on nowhere dense
classes, one needs to find suitable restrictions of $\ngFOWplus$ and
$\SC,\Weights,\PP$.

\begin{example}\label{ex:GraphsAsTrees}
  This example shows how to transform \(\FO\) definable problems over graphs
  into \(\ngFOWplus\) definable problems over vertex-weighted trees of height \(1\).

  Let \(\sigma \deff \set{E}\) for a binary relation symbol \(E\),
  let \(\SC\) consist of a single abelian group,
  the group \((\ZZ, +)\) of integers with natural addition,
  let \(\PP \deff \set{\Pred_=}\) with \(\sem{\Pred_=} \deff \setc{(i,i)}{i \in \ZZ}\),
  and let \(\Weights \deff \set{\weight_1, \weight_2}\)
  for two unary weight symbols with range $\ZZ$.
For variables $x,y,e$, we let
  \begin{align*}
    \phi_V(x)
    &\ \deff\ \big(\,0 = \weight_2(x)\,\big),\\
    \phi_{E,1}(x,e)
    &\ \deff\ \Big(\,\Pred_=\bigl(\weight_1(x), \weight_1(e)\bigr) \
      \und \ \dist(x,e) \,{\leq}\, 2 \,\Big),\\
    \phi_{E,2}(y,e)
    &\ \deff\ \Big(\, \Pred_=\bigl(\weight_1(y), \weight_2(e)\bigr) \
      \und \ \dist(y,e) \,{\leq}\, 2\,\Big),\\
    \intertext{and}
    \phi_E(x,y)
    &\ \deff\
    \phi_V(x) \land \phi_V(y) \land
    \exists e\,\Bigl(
    \bigr(\phi_{E,1}(x,e) \land \phi_{E,2}(y,e)\bigr)
    \lor
    \bigr(\phi_{E,1}(y,e) \land \phi_{E,2}(x,e)\bigr)
    \Bigr).
  \end{align*}

  For a formula \(\phi(\tx) \in \FO\) of signature $\sigma$,
  we let \(\psi(\tx) \) be the $\ngFOWplus$ formula obtained from \(\phi\)
  by replacing every subformula of the form \(\exists x\, \phi'\)
  by \(\exists x\, (\phi_V(x) \land \phi')\)
  and replacing every atomic formula of the form \(E(x,y)\)
  by \(\phi_E(x,y)\).

  Now let \(G\) be an undirected graph with vertex set \(V(G) = [n]\) for some \(n \in \NNpos\).
  We let \(\mathcal{T}\) be the \((\sigma, \Weights)\)-structure with
  universe \(T \deff [0,n] \cup E(G)\),
  edges
  \(E^{\mathcal{T}} \deff \set{0} \times \bigl([n] \cup E(G)\bigr)
  \cup \bigl([n] \cup E(G)\bigr) \times \set{0}\),
  and weights \(\weight_1(i) \deff i\) for all \(i \in [0,n]\),
  \(\weight_2(0) \deff 1\),
  \(\weight_2(j) \deff 0\) for all \(j \in [n]\),
  and \(\weight_1(\set{i,j}) \deff i\) and \(\weight_2(\set{i,j}) \deff j\)
  for all \(\set{i,j} \in E(G)\) with \(i < j\).
It can easily be verified that,
  for all \(\tv \in [n]^{\abs{\tx}}\),
  it holds that \(G \models \phi(\tv)\)
  if and only if \(\mathcal{T} \models \psi(\tv)\).
\end{example}

\section{Application:\texorpdfstring{\newline}{} Rank-Preserving Gaifman Normal Forms
for \texorpdfstring{$\FOplus$}{FO+}
and \texorpdfstring{$\FOMODplus{}$}{FO+MOD+}}\label{sec:rpGNFforFOandFOMOD}

The purpose of this section is to formulate corollaries to \cref{thm:GNF_ngFOWplus} for
first-order logic and for first-order logic with modulo-counting quantifiers.
Recall from Examples~\ref{ex:FOplus} and \ref{ex:FOMODplus} that both
logics, enriched with distance atoms, can be viewed as special cases
of $\ngFOWplus$ with $\PP=\emptyset$ and suitable choices of $\SC$ and
$\Weights$. In particular, $\Weights$ is empty for first-order logic,
and it contains only unary symbols for first-order logic with
modulo-counting quantifiers. In both cases, we have $\wmaxar=1$ for
the parameter $\wmaxar$ introduced at the beginning of \cref{sec:gaifman};
and the radius function defined in \cref{eq:def:radius_function} is
\begin{equation}\label{eq:specialRadiusFunction}
 \locRadFO{0}{k}{d} = \max\set{d,1}
    \quad\text{and}\quad
 \locRadFO{q{+}1}{k}{d} = 4k\cdot \locRadFO{q}{k{+}1}{d}
    \,, \ \
    \text{for all \(q,k,d \in \NN\)}
    .
\end{equation}

\subsection{A Rank-Preserving Gaifman Normal Form for \texorpdfstring{$\FOplus$}{FO+}}
\label{sec:rpGNFforFOplus}

Recall from \cref{ex:FOplus} that $\FOplus$ coincides with
$\ngFOWplus$ for the particular case that
$\SC=\Weights=\PP=\emptyset$. We write $\FOplusqkd$ to denote the set
$\Fqkd$ for this particular setting.
That is, $\FOplusParam{0}{k}{d}$ (for $k,d\in\NN$) is
the set of all formulas $\phi$ in $\FOplus$
with $\qr(\phi)=0$,
$\free(\phi)\subseteq \set{x_1,\ldots,x_k}$, and $\md{\phi}\leq d$;
and $\FOplusParam{q{+}1}{k}{d}$ (for all $q,k,d\in\NN$) is the set of all formulas
$\phi$ with $\free(\phi)\subseteq \set{x_1,\ldots,x_k}$ such that
$\phi$ is a Boolean combination of
  \begin{enumerate}[(i)]
    \item\label{item:FOplusParam_One}
      formulas in $\FOplusParam{q}{k{+}1}{d}$ (with free variables
      among $\set{x_1,\ldots,x_k}$),
    \item\label{item:FOplusParam_Two}
      distance atoms $\dist(x_i,x_j)\leq \tilde{d}$ with $i,j\in[k]$ and
      $\tilde{d}\leq \locRadFO{q{+}1}{k}{d}$,
    \item\label{item:FOplusParam_Three}
      formulas of the form $\exists x_{k+1}\big(\dist(x_i,x_{k+1})\,{\leq}\, \hat{d}
      \land \psi \big)$ with $\psi\in\FOplusParam{q}{k{+}1}{d}$, $i\in[k]$, and $\hat{d}\leq
      \locRadFO{q{+}1}{k}{d}-\locRadFO{q}{k{+}1}{d}$, and
    \item\label{item:FOplusParam_Four}
      formulas of the form $\exists x_{k+1}\,\psi$ with
      $\psi\in\FOplusParam{q}{k{+}1}{d}$.
  \end{enumerate}
Since $\Weights=\emptyset$,
\emph{aggregation quantification}
is not available.
Consequently, \enquote{local aggregation sentences} do not exist for
$\FOplusParam{q}{k}{d}$.
Note that \emph{basic local sentences for $\FOplusParam{q}{k}{d}$}
(of \emph{width} $\ell$, \emph{inner distance} $2r$, \emph{inner radius}
 at most $\tilde{r}$,  and
 \emph{inner quantifier rank} $\qr(\lambda)$) are of the form
\eqref{eq:intro-bl}, where \enquote{$\dist(y_i,y_j)\,{>}\,2r$}  is a
distance atom,
and where $\ell\geq 1$, $y_1,\ldots,y_\ell$ are $\ell$
variables, $r\geq \tilde{r}\geq 0$,
and
$\lambda(x_1)\in\FOplusParam{q{-}1}{k{+}1}{d}$ is
\emph{$\tilde{r}$-local} and has only one free variable.
A formula in \emph{Gaifman normal form for $\FOplusParam{q}{k}{d}$}
is a Boolean combination of \emph{local} formulas in
$\FOplusParam{q}{k}{d}$ and basic local sentences for $\FOplusParam{q}{k}{d}$.

As an immediate consequence of \cref{thm:GNF_ngFOWplus}, we obtain the
following \emph{rank-preserving Gaifman normal form for $\FOplusParam{q}{k}{d}$}.

\begin{corollary}\label{cor:GNF_FOplus}
Let $q,k,d\in\NN$. Every $\phi\in\FOplusqkd$ is equivalent to a
formula $\phi'$ in Gaifman normal form for $\FOplusqkd$.
Furthermore, $\free(\phi')=\free(\phi)$,
the outer (inner) quantifier rank of $\phi'$
is at most $\qr(\phi)$ (resp.\ $\qr(\phi){-}1$),
the outer (inner) radius of $\phi'$
is at most $\locRadFO{q}{k}{d}$
(resp.\ $\locRadFO{q{-}1}{k{+}1}{d}$),
and the width (inner distance)
is at most $k{+}q$
(resp.\ $\locRadFO{q}{k}{d}+2\locRadFO{q{-}1}{k{+}1}{d}$).
Moreover,
such a $\phi'$ can be computed from $\phi,q,k,d$.
\end{corollary}

\subsection{A Rank-Preserving Gaifman Normal Form for \texorpdfstring{$\FOMODplus{}$}{FO+MOD+}}
\label{sec:rpGNFforFOMODplus}

Let $M$ be a finite, non-empty set of integers $m\geq 2$, and consider
the logic $\FOMODplus{M}$ of \cref{ex:FOMODplus}. Following
\cref{ex:FOMODplus}, we let $\PP=\emptyset$, we let $\SC$ consist of
the abelian groups $\ZZ/m\ZZ$ with addition modulo $m$, for all $m\in
M$, and we let $\Weights=\setc{\weight_m}{m\in M}$ consist of a unary
weight symbol with range $\ZZ/m\ZZ$, for every $m\in M$.
We turn every $\sigma$-structure $\A$ into a
$(\sigma,\Weights)$-structure $\hat{\A}$ by letting
$\weight_m(a)=1\in\ZZ/m\ZZ$, for every $m\in M$ and every $a\in A$.
Recall that a \enquote{modulo $m$}-quantification
$\exists^{i\text{\,mod\,}m}\,x\,\phi$ in $\FOMODplus{M}$ corresponds to an
aggregation quantification $\big(i=\sum \weight_m(x).\phi \big)$, and
using this, we can rewrite $\FOMODplus{M}$ formulas $\psi(\bar x)$
into $\ngFOWplus$ formulas $\psi'(\bar x)$ (and vice versa, for our particular choice of
$\SC,\Weights,\PP$) such that for every $\sigma$-structure $\A$ and
every assignment $\bar a\in A^{|\bar x|}$ we have
$\A\models\psi(\bar a)\iff\hat{\A}\models\psi'(\bar a)$.

Regarding this setting, the set $\Fqkd$ can be rewritten into the set
$\FOMODplusqkd$, which is defined as follows.
The set $\FOMODplusParam{0}{k}{d}$ (for $k,d\in\NN$) consists
of all formulas $\phi$ in $\FOplus$
with $\qr(\phi)=0$,
$\free(\phi)\subseteq \set{x_1,\ldots,x_k}$, and $\md{\phi}\leq d$;
and $\FOMODplusParam{q{+}1}{k}{d}$ (for all $q,k,d\in\NN$) is the set of all formulas
$\phi$ with $\free(\phi)\subseteq \set{x_1,\ldots,x_k}$ such that
$\phi$ is a Boolean combination of
  \begin{enumerate}[(i)]
    \item
      formulas in $\FOMODplusParam{q}{k{+}1}{d}$ (with free variables
      among $\set{x_1,\ldots,x_k}$),
    \item
      distance atoms $\dist(x_i,x_j)\leq \tilde{d}$ with $i,j\in[k]$ and
      $\tilde{d}\leq \locRadFO{q{+}1}{k}{d}$,
    \item
      formulas of the form $\exists x_{k+1}\big(\dist(x_i,x_{k+1})\,{\leq}\, \hat{d}
      \land \psi \big)$ with $\psi\in\FOMODplusParam{q}{k{+}1}{d}$, $i\in[k]$, and $\hat{d}\leq
      \locRadFO{q{+}1}{k}{d}-\locRadFO{q}{k{+}1}{d}$,
    \item
      formulas of the form $\exists x_{k+1}\,\psi$ with
      $\psi\in\FOMODplusParam{q}{k{+}1}{d}$, and
    \item
      formulas of the form $\exists^{i\text{\,mod\,}m}\,x_{k+1}\,\psi$
      with $\psi\in\FOMODplusParam{q}{k{+}1}{d}$, $m\in M$, and $i\in[0,m{-}1]$.
  \end{enumerate}

The notion of \emph{local aggregation sentences for $\Fqkd$}, in our
particular setting, translates into the following notion.
A \emph{local modulo-counting sentence for $\FOMODplusParam{q}{k}{d}$} (of
\emph{inner radius} at most $\tilde{r}$ and
\emph{inner quantifier rank} $\qr(\lambda)$) is of the form
$
   \exists^{i\text{\,mod\,}m}\, y\,\lambda(y),
$
where
$m\in M$, $i\in[0,m{-}1]$,
$y$ is a variable,
$\tilde{r}\in\NN$,
and $\lambda(y)\in \FOMODplusParam{q{-}1}{k{+}1}{d}$ is
$\tilde{r}$-local
(with  $\free(\lambda)\subseteq\set{y}$).

\emph{Basic local sentences for $\FOMODplusParam{q}{k}{d}$}
(of \emph{width} $\ell$, \emph{inner distance} $2r$, \emph{inner radius}
 at most $\tilde{r}$,  and
 \emph{inner quantifier rank} $\qr(\lambda)$) are of the form
\eqref{eq:intro-bl}, where \enquote{$\dist(y_i,y_j)\,{>}\,2r$}  is a
distance atom,
and where $\ell\geq 1$, $y_1,\ldots,y_\ell$ are $\ell$
variables, $r\geq \tilde{r}\geq 0$,
and
$\lambda(x_1)\in\FOMODplusParam{q{-}1}{k{+}1}{d}$ is
\emph{$\tilde{r}$-local} and has only one free variable.

Finally, a formula in \emph{Gaifman normal form for $\FOMODplusParam{q}{k}{d}$}
is a Boolean combination of \emph{local} formulas in
$\FOMODplusParam{q}{k}{d}$, of basic local sentences for
$\FOMODplusParam{q}{k}{d}$, and of
local modulo-counting sentences for $\FOMODplusParam{q}{k}{d}$.

As an immediate consequence of \cref{thm:GNF_ngFOWplus}, we obtain the
following \emph{rank-preserving Gaifman normal form for
  $\FOMODplusParam{q}{k}{d}$}; it can be viewed as a
\enquote{rank-preserving} version of the Gaifman normal form for
$\FOMOD_M$ provided in \cite{KuskeSchweikardt_ICALP18}.

\begin{corollary}\label{cor:GNF_FOMODplus}
Let $q,k,d\in\NN$. Every $\phi\in\FOMODplusqkd$ is equivalent to a
formula $\phi'$ in Gaifman normal form for $\FOMODplusqkd$.
Furthermore, $\free(\phi')=\free(\phi)$,
the outer (inner) quantifier rank of $\phi'$
is at most $\qr(\phi)$ (resp.\ $\qr(\phi){-}1$),
the outer (inner) radius of $\phi'$
is at most $\locRadFO{q}{k}{d}$
(resp.\ $\locRadFO{q{-}1}{k{+}1}{d}$),
and the width (inner distance)
is at most $k{+}q$
(resp.\ $\locRadFO{q}{k}{d}+2\locRadFO{q{-}1}{k{+}1}{d}$).
Moreover,
such a $\phi'$ can be computed from $\phi,q,k,d$.
\end{corollary}

\bibliography{main}

\appendix
\crefalias{section}{appendix}
\crefalias{subsection}{appendix}
\section*{APPENDIX}

\section{Details Omitted in Section~\ref{sec:preliminaries}}
\label{appendix:LocalityNotion}

\begin{lemma}\label{lemma:LocalityNotion}
  Let $r\in\NN$, let $\phi$ be a formula, and let $\bar x=(x_1,\ldots,x_k)$ be a
  non-empty tuple of variables such that $\free(\phi)\subseteq\set{x_1,\ldots,x_k}$.
  If $\phi(\bar x)$ is $r$-local, then $\phi(\bar y)$ is $r$-local for every non-empty tuple $\bar y=(y_1,\ldots,y_\ell)$ of variables such that $\free(\phi)\subseteq \set{y_1,\ldots,y_\ell}$.
\end{lemma}
\begin{proof}
We assume that $\phi(\bar x)$ is $r$-local.

Let $\bar y=(y_1,\ldots,y_\ell)$ be a non-empty tuple of variables such that $\free(\phi)\subseteq\set{y_1,\ldots,y_\ell}$.
Consider an arbitrary structure $\A$ and a tuple $\bar b=(b_1,\ldots,b_\ell)\in A^\ell$ to which we assign the variables in $\bar y$. We have to show that
\begin{equation}\label{eq:LocalityNotion:goal}
 \A\models\phi(\bar b) \ \iff \ \NrA{\tb} \models \phi(\tb).
\end{equation}
Let us choose an element $c\in \set{b_1,\ldots,b_\ell}$ that will be fixed henceforth. Let $\bar b'$ be the tuple obtained from $\bar b$ by replacing $b_j$ with $c$ for every $j\in[\ell]$ with $y_j\not\in\free(\phi)$.

By the \emph{coincidence lemma},
the assignment of variables that are not free in a formula is irrelevant.
Hence, we have
\begin{equation}\label{eq:fuenf}
  \A\models\phi(\bar b) \ \iff \ \A\models\phi(\bar b'),
\end{equation}
and
\begin{equation}\label{eq:sechs}
\NrA{\tb} \models \phi(\tb) \ \iff \ \NrA{\tb} \models \phi(\tb').
\end{equation}
Now, let $\bar a=(a_1,\ldots,a_k)$ be the assignment
for $\bar x=(x_1,\ldots,x_k)$,
where, for every $i\in[k]$, we let
\begin{itemize}
\item $a_i\deff c$ if $x_i\not\in\free(\phi)$, and
\item $a_i \deff b_j$ if $x_i\in \free(\phi)$ and $j\in[\ell]$ such that $y_j=x_i$.
\end{itemize}
Due to the coincidence lemma,
for the evaluation of the formula $\phi$,
assigning $\bar a$ to $\bar x$ has the same effect
as assigning $\bar b'$ to $\bar y$.
In particular, we obtain that
\begin{equation}\label{eq:sieben}
  \A\models\phi(\bar b') \ \iff \ \A\models\phi(\bar a),
\end{equation}
and
\begin{equation}\label{eq:acht}
\NrA{\tb} \models \phi(\tb') \ \iff \ \NrA{\tb} \models \phi(\ta).
\end{equation}
According to the assumption of the lemma, $\phi(\bar x)$ is $r$-local. Thus, we obtain that
\begin{equation}\label{eq:neun}
\A\models\phi(\bar a) \ \iff \ \NrA{\ta}\models\phi(\bar a),
\end{equation}
and, using the $r$-locality also for the structure $\B\deff\NrA{\tb}$,
we obtain that
$
  \B\models\phi(\bar a) \iff
  \NrB{\ta}\models\phi(\bar a)$.
Since all entries of the tuple $\bar a$ also occur in the tuple $\bar b$,
it is straightforward to see that
\(\NrA{\ta} \subseteq \NrA{\tb} = \B\),
and thus,
$\NrB{\ta}=\NrA{\ta}$.
Hence, we have
\begin{equation}\label{eq:zehn}
  \NrA{\tb}\models\phi(\bar a) \ \iff \
  \NrA{\ta}\models\phi(\bar a).
\end{equation}
In summary, we obtain the following equivalences:
\[
\begin{array}{ccrcc}
  \A\models\phi(\bar b)
& \stackrel{\eqref{eq:fuenf},\eqref{eq:sieben}}{\iff}
& \A\models\phi(\bar a)
& \stackrel{\eqref{eq:neun}}{\iff}
& \NrA{\ta}\models\phi(\ta)
  \\
& \stackrel{\eqref{eq:zehn}}{\iff}
& \NrA{\tb}\models\phi(\ta)
& \stackrel{\eqref{eq:acht},\eqref{eq:sechs}}{\iff}
& \NrA{\tb}\models\phi(\tb).
\end{array}
\]
This proves that equivalence \eqref{eq:LocalityNotion:goal} holds.
Hence, the proof of \cref{lemma:LocalityNotion} is complete.
\end{proof}

\section{Details Omitted in Section~\ref{sec:FO-Gaifman}}
\label{appendix:fo-gaifman}

\GaifmanCounterExample*

\begin{claimproof}
Fix \(r \in \NN\) with \(r \geq 7\),
  and let \(\CN \deff \CN_r\).
  We assume that the vertex set of \(\CN\) is \(N \deff \set{1, \dots, r+9}\),
  where the vertices are numbered as indicated in \cref{fig:fo-gaifman};
  for general \(r\), the vertices on the path from node \(11\)
  to the leftmost orange node are \(11, \dots, r+9\).

  We call a pair \((w,w') \in N^2\) \emph{good} if Duplicator has a winning strategy
  for the \(1\)-round EF game on \((\CN,(v,w))\), \((\CN,(v',w'))\).
  We need to prove that for every \(w \in N\),
  there is a \(w' \in N\) such that \((w,w')\) is good,
  and for every \(w'\in N\), there is a \(w \in N\) such that \((w,w')\) is good.

  We give a list of good pairs witnessing this;
  for each pair on the list, it is straightforward to verify that it is good.
  \begin{itemize}
    \item \((x,x)\) for each \(x \geq 11\);
    \item \((1,5), (9,10), (2,4), (3,3)\)
      and the reversed pairs \((5,1), (10,9), (4,2)\);
    \item \((6,15), (7,16)\);
    \item \((2,6), (3,7), (8,8)\).
  \end{itemize}
  Note that indeed every \(x \in N\) occurs in the first and in the second component
  of a some pair in this list.
  This completes the proof of Claim~\ref{claim:GaifmanCounterExample}.
\end{claimproof}

Combining \cref{claim:GaifmanCounterExample}
with the Ehrenfeucht-\Fraisse\ Theorem (cf., \cite{Libkin-FMT,EF-FMT})
yields that for all \(\FO\) formulas \(\psi(x)\) of quantifier rank \(\leq 2\),
we have \(\CN_r \models \psi(v) \iff \CN_r \models \psi(v')\).
Thus, for all \(\FO\) formulas \(\psi(x)\)
of quantifier rank \(\leq 2\) that are \(r\)-local,
we have \(\CG_r \models \psi(v) \iff \CG_r \models \psi(v')\).

\begin{claim}
  The formula \(\phi(x)\) is not equivalent to
  a Boolean combination \(\phi'(x)\) of local formulas of quantifier rank at most \(2\)
  and of basic local sentences.
\end{claim}
\begin{claimproof}
  Suppose for contradiction it was.
  Then let \(L\) and \(S\) be the sets of local formulas
  and basic local sentences, respectively,
  such that \(\phi'(x)\) is a Boolean combination of the formulas in \(L \cup S\).
  Let \(r \in \NN\) with \(r \geq 7\) be such that each formula in \(L\) is \(r\)-local.
  Then for each \(\psi(x)\) in \(L\),
  we have \(\CG_r \models \psi(v)\iff\CG_r \models \psi(v')\).
  This implies that \(\CG_r \models \phi'(v)\iff\CG_r \models \phi'(v')\).
  However, this contradicts \(\phi'(x)\) being equivalent to \(\phi(x)\),
  because we know that \(\CG_r \models \phi(v)\)
  and \(\CG_r \not \models \phi(v')\).
\end{claimproof}
 
\section{Details Omitted in Section~\ref{sec:fow}}\label{appendix:fow}

\subsection{Detailed Definition of Syntax and Semantics of \texorpdfstring{$\FOWplus$}{FOW+}}

We use the same notion of weighted structures as \cite{vanBergeremSchweikardt_2021_FOWA}:
Let \(\sigma\) be a signature,
let \(\SC\) be a collection of rings and/or abelian groups,
and
let \(\Weights\) be a finite set of \emph{weight symbols},
where each $\weight\in\Weights$ comes with an
\emph{arity} \(\ar(\weight) \in \NN\)
and a \emph{type}, also called \emph{range},
\(\wtype(\weight)\in\SC\).
A \emph{\((\sigma, \Weights)\)-structure} is a \(\sigma\)-structure \(\A\)
that is enriched, for every \(\weight\in \Weights\), by an interpretation
\(\weight^\A \colon A^{\ar(\weight)} \to \wtype(\weight)\)
that satisfies the following \emph{locality condition}:
if \(\weight^\A(a_1, \dots, a_\ell) \neq \nullS\) for
\(S \deff \wtype(\weight)\),
\(\ell \deff \ar(\weight)\), and
\(a_1, \dots, a_\ell \in A\),
then $\ell = 0$ or \(a_1 = \cdots = a_\ell\)
or all of the elements \(a_1, \dots, a_\ell\) are contained in one tuple of a relation of \(\A\).
More formally, in the latter case, there exists
a relation symbol
\(R \in \sigma\)
and a
tuple
\((b_1, \dots, b_{\ar(R)}) \in R^{\A}\)
such that \(\set{a_1, \dots, a_\ell} \subseteq \set{b_1, \dots, b_{\ar(R)}}\).
Hence, if \(\weight^\A(a_1, \dots, a_\ell) \neq \nullS\),
then the elements in \(\set{a_1, \dots, a_\ell}\)
form a clique in the Gaifman graph of \(\A\).
All notions that were introduced for \(\sigma\)-structures
carry over to \((\sigma, \Weights)\)-structures in the obvious way.

We fix a countably infinite set \(\vars\) of variables that can be used
to build formulas.
The \emph{atomic formulas} of \(\FOWplus\) of signature \(\sigma\)
and with weights according to \(\SC\) and \(\Weights\)
are

\begin{itemize}
\item
the same \emph{relational atoms}
(rule~\eqref{def:fo-atomic}, including equalities, \(\true\), and \(\false\))
as in \(\FO\) of signature \(\sigma\),
\end{itemize}

\noindent
and, additionally,

\begin{itemize}
\item
  \emph{distance atoms} of the form \(\distAtom{x}{y}{d}\) for
  \(x, y \in \vars\) and \(d \in \NN\) (rule~\eqref{def:fowplus-dist});
  \smallskip

\item
\emph{weight atoms} of the form \(\bigl(s = \weight(y_1, \dots, y_\ell)\bigr)\),
where \(\weight\) is an \(\ell\)-ary weight symbol in \(\Weights\),
\(s \in \wtype(\weight)\),
and \(y_1, \dots, y_\ell \in \vars\) are pairwise distinct variables
(rule~\eqref{def:fow-wsimple});
and
\smallskip

\item
  \emph{term-comparison atoms} (rule~\eqref{def:fow-P}) that allow comparisons
  of values of \emph{\(\SC\)-terms}.
\end{itemize}

\noindent
The latter are defined as follows.

For every \(S \in \SC\), every \(s \in S\) is an \(\SC\)-term of type \(S\);
for every weight symbol \(\weight \in \Weights\) of type \(S\) and arity \(\ell\)
and all pairwise distinct \(y_1, \dots, y_\ell \in \vars\),
\(\weight(y_1, \dots, y_\ell)\) is an \(\SC\)-term of type \(S\);
and \(\SC\)-terms \(t_1\) and \(t_2\) of the same type \(S\)
can be combined into an \(\SC\)-term \((t_1 * t_2)\)
of type \(S\),
where \(*\) is a binary operation available for \(S\)
(i.e.\ \(+,-\) if \(S\) is a group, and \(+,-,\cdot\) if \(S\) is a ring).
We write \(\free(t)\) for the set of \emph{free variables}
of an \(\SC\)-term \(t\),
i.e.\ the set of variables occurring in \(t\),
and we write \(t(x_1, \dots,x_k)\) to stipulate that
\(\free(t) \subseteq \set{x_1, \dots, x_k}\).
Consider such a term \(t(x_1, \dots, x_k)\) of type \(S\).
For a \((\sigma,\Weights)\)-structure \(\A\) and elements \(a_1, \dots, a_k \in A\),
we write \(t^\A(a_1, \dots, a_k)\) to denote the element in \(S\)
to which \(t\) evaluates when replacing each \(x_i\) with \(a_i\),
interpreting the weight symbols \(\weight\)
with the weight functions \(\weight^\A\) of \(\A\),
and using the operations present in \(S\) to calculate the resulting element.

To equip the logic with a kind of atomic formulas that allow to
compare values of such terms,
we use the following notion of
\cite{vanBergeremSchweikardt_2021_FOWA}.
An \emph{\(\SC\)-predicate collection} is a 4-tuple \((\PP, \ar, \ptype, \sem{\cdot})\),
where \(\PP\) is a countable set of \emph{predicate names} and, to each \(\Pred \in \PP\),
\(\ar\) assigns an \emph{arity} \(\ar(\Pred) \in \NNpos\),
\(\ptype\) assigns a \emph{type} \(\ptype(\Pred) \in \SC^{\ar(\Pred)}\),
and \(\sem{\cdot}\) assigns a \emph{semantics} \(\sem{\Pred} \subseteq \ptype(\Pred)\).
Given such an \(\SC\)-predicate collection \((\PP, \ar, \ptype, \sem{\cdot})\),
the \emph{term-comparison atoms} of \(\FOWplus\) (rule~\eqref{def:fow-P})
are of the form \(\Pred(t_1, \dots, t_m)\),
where \(\Pred \in \PP\),
\(m = \ar(\Pred)\),
and \(t_1, \dots, t_m\) are \(\SC\)-terms such that
\(\ptype(\Pred) = \big(\wtype(t_1), \dots, \wtype(t_m)\big)\).
The set of \emph{free variables} of such an atom
is the union of the free variables of \(t_1, \dots, t_m\).
Given a \((\sigma,\Weights)\)-structure \(\A\)
and an assignment \(\ta = (a_1, \dots, a_k) \in A^k\)
to the free variables \(x_1, \dots, x_k\) of \(t_1, \dots, t_m\),
the formula \(\Pred(t_1, \dots, t_m)\) \emph{is satisfied}
by \((\A, a_1, \dots, a_k)\) if and only if \((s_1, \dots, s_m) \in \sem{\Pred}\),
where for each \(i \in [m]\),
\(s_i\) is the element in \(S_i = \wtype(t_i)\) to which \(t_i\) evaluates in \(\A\)
with variable assignment \(\ta\).
(i.e.\ \(s_i = t_i^\A(\ta)\)).

We have now presented all rules
(\eqref{def:fo-atomic},
\eqref{def:fowplus-dist},
\eqref{def:fow-wsimple},
and
\eqref{def:fow-P})
for building \emph{atomic formulas} of \(\FOWplus\).
More complex formulas can be built using the same rules
for \emph{Boolean combinations} (rule~\eqref{def:fo-bool})
and \emph{existential quantification} (rule~\eqref{def:fo-exists})
as in $\FO$.
Finally,
we have also available the following rule
\eqref{def:fow-finitegroup} called \emph{aggregation quantification}:
if \(\phi\) is a formula,
      \(\weight \in \Weights\),
      \(S = \wtype(\weight)\) is \emph{finite},
      \(s \in S\),
      \(\ell = \ar(\weight)\),
      and \(\ty = (y_1, \dots, y_\ell)\) is a tuple of
      \(\ell\) pairwise distinct variables,
      then \(\bigl( s = \sum \weight(\ty).\phi \bigr)\) is a formula
      (rule \eqref{def:fow-finitegroup}).
The set of free variables of such a formula is
$\free(\phi)\setminus\set{y_1,\ldots,y_\ell}$.
Let $\bar x=(x_1,\ldots,x_k)$ be a list of the free variables of the
formula $\psi\deff \bigl( s = \sum \weight(\ty).\phi \bigr)$, and
let $S\deff\wtype(\weight)$. For
a $(\sigma,\Weights)$-structure $\A$ and a tuple $\bar a \in A^k$, we have
$\A\models\psi(\bar a)$ if and only if the sum (using the addition
$+_S$ associated with $S$) of the weights $\weight^{\A}(\bar b)$ over
all tuples $\bar b$ in the set $\setc{\bar b\in
  A^\ell}{\A\models\phi(\bar a,\bar b)}$ evaluates to~$s$.
This completes the definition of the logic $\FOWplus$.

The logic $\FOW$ is defined as the restriction of $\FOWplus$ where
distance atoms (rule \eqref{def:fowplus-dist}) are not available. The logic $\FOWun$ of
\cite{vanBergeremSchweikardt_2021_FOWA} is the restriction of $\FOW$
where the \emph{term-comparison} rule \eqref{def:fow-P} can only be used
if $\bigcup_{i \in [m]}\free(t_i)$ contains at most one variable.

\bigskip

The following \cref{def:fow} provides a more formal, mathematically
precise definition of the syntax and semantics of $\FOWplus$.
A \emph{\((\sigma, \Weights)\)-interpretation} \(\I = (\A, \beta)\)
consists of a \((\sigma, \Weights)\)-structure \(\A\)
and an \emph{assignment} \(\beta \colon \vars \to A\).

\begin{definition}[\(\FOWplus\)]
  \label[definition]{def:fow}
  Let \(\sigma\) be a signature, \(\SC\) a collection of rings and/or abelian groups,
  \(\Weights\) a finite set of weight symbols, and
  \((\PP, \ar, \ptype, \sem{\cdot})\) an \(\SC\)-predicate
  collection.
  The set of \emph{formulas} and \emph{\(\SC\)-terms} for \(\FOWplusPS\)
  is built according to the following rules.

  \begin{enumerate}[(WA)]
    \renewcommand{\labelenumi}{\textbf{(\theenumi)}}
    \renewcommand{\theenumi}{DA}
    \item\label{def:fowplus-dist}
      \(\dist(x_1,x_2)\,{\leq}\, d\) are formulas,
      for \(x_1, x_2 \in \vars\) and \(d \in \NN\).
    \renewcommand{\theenumi}{RA}
    \item\label{def:fo-atomic}\label{def:fowplus-truefalse}
      \(\true\) and \(\false\) and
      \(x_1{=}x_2\) and \(R(x_1, \dots, x_k)\) are formulas for
      \(R \in \sigma\), \(k \deff \ar(R)\),
      and \(x_1, x_2, \dots, x_k \in \vars\).
      In particular, if \(\ar(R) = 0\), then \(R()\) is a formula.
    \renewcommand{\theenumi}{WA}
    \item\label{def:fow-wsimple}
      \(\bigl(s = \weight(\tx)\bigr)\) is a formula, for
      \(\weight \in \Weights\),
      \(S = \wtype(\weight)\), \(s \in S\), \(\ell = \ar(\weight)\),
      and \(\tx = (x_1, \dots, x_\ell)\)
      a tuple of \(\ell\) pairwise distinct variables.
    \renewcommand{\theenumi}{BC}
    \item\label{def:fo-bool}
      If \(\phi\) and \(\psi\) are formulas,
      then \(\neg \phi\) and \((\phi \lor \psi)\) are formulas.
    \renewcommand{\theenumi}{\(\exists\)}
    \item\label{def:fo-exists}
      If \(\phi\) is a formula and \(x \in \vars\),
      then \(\exists x\, \phi\) is a formula.
    \renewcommand{\theenumi}{\(\Sigma\)}
    \item\label{def:fow-finitegroup}
      If \(\phi\) is a formula,
      \(\weight \in \Weights\),
      \(S = \wtype(\weight)\) is \textbf{finite},
      \(s \in S\),
      \(\ell = \ar(\weight)\),
      and \(\tx = (x_1, \dots, x_\ell)\) is a tuple of
      \(\ell\) pairwise distinct variables,
      then \(\bigl( s = \sum \weight(\tx).\phi \bigr)\) is a formula.
    \renewcommand{\theenumi}{TA}
    \item\label{def:fow-P}
      \(\Pred(t_1, \dots, t_m)\) is a formula for
      \(\Pred \in \PP\),
      \(m = \ar(\Pred)\),
      and \(\SC\)-terms \(t_1, \dots, t_m\)
      with \(\ptype(\Pred) = \bigl(\wtype(t_1), \dots, \wtype(t_m)\bigr)\).
    \renewcommand{\theenumi}{T1}
    \item\label{def:fow-constterm}
      \(s\) is an \(\SC\)-term of type \(S\), for every \(S \in \SC\) and every \(s \in S\).
    \renewcommand{\theenumi}{T2}
    \item\label{def:fow-wsimpleterm}
      \(\weight(x_1, \dots, x_\ell)\) is an \(\SC\)-term of type \(S\),
      for every \(S \in \SC\), every \(\weight \in \Weights\) of type \(S\),
      and every tuple \((x_1, \dots, x_\ell)\)
      of \(\ell \deff \ar(\weight)\) pairwise distinct variables in \(\vars\).
    \renewcommand{\theenumi}{T3}
    \item\label{def:fow-plustimesterm}
      If \(t_1\) and \(t_2\) are \(\SC\)-terms of the same type \(S\),
      then \((t_1 \plus t_2)\) and \((t_1 \minus t_2)\)
      are also \(\SC\)-terms of type \(S\);
      furthermore, if \(S\) is a ring (and not just an abelian group),
      then also \((t_1 \mal t_2)\) is an \(\SC\)-term of type \(S\).
  \end{enumerate}

  \noindent
  Let \(\I = (\A, \beta)\) be a \((\sigma, \Weights)\)-interpretation.
  For a formula or \(\SC\)-term \(\xi\) from \(\FOWplusPS\),
  the semantics \(\sem{\xi}^\I\) is defined as follows.

  \begin{enumerate}[(WA)]
    \renewcommand{\labelenumi}{\textbf{(\theenumi)}}
    \renewcommand{\theenumi}{DA}
    \item
      $\sem{\dist(x_1,x_2)\,{\leq}\,d}^{\I}=1$ if
      $\dist^{\A}(\beta(x_1),\beta(x_2))\leq d$, and
      $\sem{\dist(x_1,x_2)\,{\leq}\,d}^{\I}=0$ otherwise.
    \renewcommand{\theenumi}{RA}
    \item
      $\sem{\true}^{\I}=1$; $\sem{\false}^{\I}= 0$;
      \(\sem{x_1{=}x_2}^\I = 1\) if \(\beta(x_1) = \beta(x_2)\),
      and \(\sem{x_1{=}x_2}^\I = 0\) otherwise;
      \(\sem{R(x_1, \dots, x_k)}^\I = 1\) if
      \(\bigl(\beta(x_1), \dots, \beta(x_k)\bigr) \in R^{\A}\), and
      \(\sem{R(x_1, \dots, x_k)}^\I = 0\) otherwise.
    \renewcommand{\theenumi}{WA}
    \item
      \(\sem{\bigl(s = \weight(\tx)\bigr)}^{\I} = 1\)
      if \(s= \weight^\A \bigl(\beta(x_1), \dots, \beta(x_\ell)\bigr)\), and
      \(\sem{\bigl(s = \weight(\tx)\bigr)}^\I = 0\) otherwise.
    \renewcommand{\theenumi}{BC}
    \item
      \(\sem{\neg \phi}^\I = 1 - \sem{\phi}^\I\)
      \ and \ \(\sem{(\phi \lor \psi)} = \max \bigset{\sem{\phi}^\I, \sem{\psi}^\I}\).
    \renewcommand{\theenumi}{\(\exists\)}
    \item
      \(\sem{\exists x\, \phi}^\I = \max \bigsetc{\sem{\phi}^{\I\frac{v}{x}}}{v \in A}\).
    \renewcommand{\theenumi}{\(\Sigma\)}
    \item
      \(\sem{\bigl(s=\sum \weight(\tx).\phi\bigr)}^\I = 1\)
      if \(s = \underset{\tv \in M}{\sum_S} \weight^\A(\tv)\)
      for \(M \deff \bigsetc{(v_1, \dots, v_\ell) \in A^\ell}
      {\sem{\phi}^{\I \frac{v_1, \dots, v_\ell}{x_1, \dots, x_\ell}}=1}\),
      and \(\sem{\bigl(s=\sum \weight(\tx).\phi\bigr)}^\I = 0\)
      otherwise.
      By convention, \(\underset{\tv \in M}{\sum_S} \weight^\A(\tv) = \nullS\) if \(M = \emptyset\).
    \renewcommand{\theenumi}{TA}
    \item
      \(\sem{\Pred(t_1, \dots, t_m)}^\I = 1\) if
      \(\bigl(\sem{t_1}^\I, \dots, \sem{t_m}^\I\bigr) \in \sem{\Pred}\),
      and \(\sem{\Pred(t_1,\ldots,t_m)}^\I = 0\) otherwise.
    \renewcommand{\theenumi}{T1}
    \item
      \(\sem{s}^\I = s\) for \(s \in S\) for some \(S \in \SC\).
    \renewcommand{\theenumi}{T2}
    \item
      \(\sem{\weight(x_1, \dots, x_\ell)}^\I =
      \weight^\A\bigl(\beta(x_1), \dots, \beta(x_\ell)\bigr)\).
    \renewcommand{\theenumi}{T3}
    \item \(\sem{(t_1 \ast t_2)}^\I = \sem{t_1}^\I \ast_S \sem{t_2}^\I\),
      for \(\ast \in \set{\plus, \minus, \mal}\).
  \end{enumerate}
\end{definition}

We write $\dist(x_1,x_2)\,{>}\,d$ as a shorthand for the formula $\nicht\,\dist(x_1,x_2)\,{\leq}\,d$.
An \emph{expression} is a formula or an \(\SC\)-term.
The set \(\varsof(\xi)\) of an expression \(\xi\) is defined as the set
of all variables in \(\vars\) that occur in \(\xi\).
The \emph{free variables} \(\free(\xi)\) of \(\xi\) are inductively defined as follows.

\begin{enumerate}[(WA)]
  \renewcommand{\labelenumi}{\textbf{(\theenumi)}}
  \renewcommand{\theenumi}{DA}
  \item \(\free(\dist(x_1,x_2)\,{\leq}\,d) = \set{x_1,x_2}\).
  \renewcommand{\theenumi}{RA}
  \item \(\free(\true)=\free(\false)=\emptyset\);
    \(\free(x_1{=}x_2) = \set{x_1, x_2}\);
    and
    \(\free\bigl(R(x_1, \dots, x_k)\bigr) = \set{x_1, \dots, x_k}\).
  \renewcommand{\theenumi}{WA}
  \item \(\free\Big(\bigl(s = \weight(x_1, \dots, x_\ell)\bigr)\Big) =
    \set{x_1, \dots, x_\ell}\).
  \renewcommand{\theenumi}{BC}
  \item \(\free(\neg \phi) = \free(\phi)\) and
    \(\free(\phi \lor \psi) = \free(\phi) \cup \free(\psi)\).
  \renewcommand{\theenumi}{\(\exists\)}
  \item \(\free(\exists x\, \phi) = \free(\phi) \setminus \set{x}\).
  \renewcommand{\theenumi}{\(\Sigma\)}
  \item \(\free\Big(\bigl(s = \sum\weight(x_1, \dots, x_\ell).\phi\bigr)\Big) =
    \free(\phi) \setminus \set{x_1, \dots, x_\ell}\),
  \renewcommand{\theenumi}{TA}
  \item \(\free\bigl(\Pred(t_1, \dots, t_m)\bigr) = \bigcup_{i=1}^m \free(t_i)\).
  \renewcommand{\theenumi}{T1}
  \item \(\free(s) = \emptyset\) for \(s \in S\) for some \(S \in \SC\).
  \renewcommand{\theenumi}{T2}
  \item \(\free\bigl(\weight(x_1, \dots, x_\ell)\bigr) = \set{x_1, \dots, x_\ell}\).
  \renewcommand{\theenumi}{T3}
  \item \(\free\bigl((t_1 \ast t_2)\bigr) = \free(t_1) \cup \free(t_2)\)
    for \(\ast \in \set{\plus, \minus, \mal}\).
\end{enumerate}

We write \(\xi(x_1, \dots, x_k)\) to indicate that
\(\free(\xi) \subseteq \set{x_1, \dots, x_k}\).
A \emph{sentence} is a formula without free variables.

The notation
$\xi(x_1,\ldots,x_k)$ is also convenient for substitutions; by
$\xi(y_1,\ldots,y_k)$, we denote the expression obtained from $\xi$ by
substituting $x_i$ with $y_i$, for all $i$ (and renaming bound
variables if necessary, cf.~\cite[Section~III.8]{EbbinghausFT21}).

For a formula \(\phi\) and a \((\sigma, \Weights)\)-interpretation \(\I\),
we write \(\I \models \phi\) to indicate that \(\sem{\phi}^\I = 1\).
Likewise, \(\I \not\models \phi\) indicates that
\(\sem{\phi}^\I = 0\).
For a formula \(\phi\) with \(\free(\phi) \subseteq \set{x_1, \dots, x_k}\),
a \((\sigma, \Weights)\)-structure \(\A\),
and a tuple \(\tv = (v_1, \dots, v_k) \in A^k\),
we write \(\A \models \eval{\phi}{\tv}\) or \((\A, \tv) \models \phi\)
to indicate that \((\A, \beta) \models \phi\) for one (and hence every) assignment \(\beta\)
with \(\beta(x_i) = v_i\) for all \(i \in [k]\).
Furthermore, we set \(\sem{\phi(\tv)}^\A \deff 1\) if \(\A \models \eval{\phi}{\tv}\),
and \(\sem{\phi(\tv)}^\A \deff 0\) otherwise.
Similarly, for an \(\SC\)-term \(t(\tx)\),
we write \(t^{\A}(\tv)\) to denote \(\sem{t}^\I\).

The \emph{quantifier rank} \(\qr(\xi)\) of an \(\FOWplusPS\) expression \(\xi\)
is defined as the maximum nesting depth of constructs using
rules~\eqref{def:fo-exists} and~\eqref{def:fow-finitegroup} in order to construct \(\xi\).

$\FOWPS$ is defined as the restriction of $\FOWplusPS$ where rule
\eqref{def:fowplus-dist} is not available.

Note that first-order logic \(\FO[\sigma]\)
is the restriction of \(\FOWplusPS\)
where only rules~\eqref{def:fo-atomic}, \eqref{def:fo-bool},
and \eqref{def:fo-exists} can be applied.
As usual, we write \((\phi \land \psi)\) and \(\forall x\,\phi\) as shorthands
for \(\neg(\neg\phi \lor \neg\psi)\) and \(\neg \exists x\,\neg\phi\).

Two $\FOWplusPS$ formulas $\phi$ and $\psi$ are \emph{equivalent} (for
short: $\phi\equiv\psi$) if $\sem{\phi}^{\I}=\sem{\psi}^{\I}$ holds
for all $(\sigma,\Weights)$-interpretations $\I$.

Let \(r \in \NN\).
For an \(\FOWplusPS\) formula \(\phi\) and a non-empty tuple \(\tx = (x_1,
\dots, x_k)\) of variables with
$\free(\phi)\subseteq\set{x_1,\ldots,x_k}$, we say that $\phi(\tx)$
is \emph{\(r\)-local} (around \(\tx\))
if for every \((\sigma,\Weights)\)-structure \(\A\) and all \(\ta \in A^k\)
we have that \(\A \models \phi(\ta)\) $\iff$ \(\NrA{\ta}
\models \phi(\ta)\);
note that this implies $r'$-locality for all $r'\geq r$.
A formula $\phi$ is \emph{$r$-local} if $\phi(\bar x)$ is $r$-local for some
non-empty  $\tx =(x_1,\ldots, x_k)$ with $\free(\phi)\subseteq
\set{x_1,\ldots,x_k}$; and
it is \emph{local} if it is \(r\)-local for some \(r \in \NN\).
In particular, for a relation symbol $R$ of
arity 0, the sentence $R()$ is 0-local.

To avoid notational clutter,
when $\sigma$, $\SC$, $\Weights$, $\PP$ are clear from the context, we
simply write $\FOWplus$ instead of $\FOWplusPS$, and we use similar
conventions for other logics.

The logic \(\ngFOWplus\) is obtained from \(\FOWplus\)
by replacing rule~\eqref{def:fow-P}
with its \emph{neighbourhood-guarded} version
\begin{enumerate}[(ngTA)]
  \renewcommand{\labelenumi}{\textbf{(\theenumi)}}
  \renewcommand{\theenumi}{ngTA}
  \item\label{def:ngfow-P}
    \(\Big(
       \Pred(t_1, \dots, t_m)
       \land
       \Land_{y \in \bigcup_{i=1}^{m} \free(t_i) \setminus \set{x}} \dist(x,y) \leq d
    \Bigr)\)
    is a formula
    for \(\Pred \in \PP\),
    \(m = \ar(\Pred)\),
    \(x \in \vars\),
    \(d \in \NN\),
    and
    \(\SC\)-terms \(t_1, \dots, t_m\)
    with \(\ptype(\Pred) = \bigl(\wtype(t_1), \dots, \wtype(t_m)\bigr)\).
  \end{enumerate}

  By $\ngFOW$,
  we denote the variant of $\ngFOWplus$  where distance atoms are not
  available but are replaced by equivalent $\FO$ formulas
  $\delta^{\sigma}_{\leq d}(x,y)$ (cf.\ the beginning of \cref{sec:FO-Gaifman}).

\subsection{Details Omitted in Example~\ref{example:noGNFforFOW}}

\begin{claim}
  Let $\phi(x)$ be the formula of Example~\ref{example:noGNFforFOW}.
  There exists no formula that is equivalent to
  $\phi(x)$ and that is a Boolean combination of sentences and
  local formulas of $\FOWplus$.
\end{claim}
\begin{claimproof}
  Towards a contradiction, suppose there is an \(\FOWplus\) formula \(\psi(x)\)
  that is a Boolean combination of sentences and local formulas
  such that \(\psi(x)\) is equivalent to \(\phi(x)\).

  Since \(S\) is infinite, there is an injective mapping \(\pi \colon \NNpos \to S\).
  For \(n \in \NNpos\) and \(i \in [n]\),
  we define the \((\sigma, \Weights)\)-structure \(\A_{n,i}\)
  with universe \(A_{n,i} \deff [n{+}1]\),
weights \(\weight^{\A_{n,i}}(j) \deff \pi(j)\) for all \(j \in [n]\),
  and \(\weight^{\A_{n,i}}(n{+}1) \deff \pi(i)\).

  Let $\mathcal{S}$ and $\mathcal{L}$ be the finite sets of sentences
  and local formulas, resp., such that $\psi(x)$ is a Boolean
  combination of the
  formulas in $\mathcal{S}\cup\mathcal{L}$.
  Set \(n \deff 2^{\abs{\mathcal{S}}} + 1\).
  By the pigeonhole principle,
  there are \(i, i' \in [n]\) with \(i \neq i'\)
  such that \(\A_{n,i}\) and \(\A_{n,i'}\) satisfy exactly the same
  sentences in $\mathcal{S}$.
  Let \(\psi'(x)\) be the formula obtained from \(\psi(x)\)
  by replacing all sentences in \(\psi\) by \(\true\) if they are satisfied by \(\A_{n,i}\)
  (and hence also by \(\A_{n,i'}\)), and
  replacing all other sentences in \(\psi\) by \(\false\).
Then \(\psi'(x)\) is a local formula in $\FOWplus$.
  Hence, $\psi'(x)$ is $r$-local for some $r\in \NN$.
  This implies that
  for all \(j \in [n{+}1]\),
  we have
\[
  \A_{n,i} \models \phi(j)
  ~ \iff ~
  \A_{n,i} \models \psi'(j)
  ~ \iff ~
  \Neighbr{\A_{n,i}}{j} \models \psi'(j),
  \]
  and the same holds for \(i'\) instead of \(i\).
  By the definition of \(\A_{n,i}\) and \(\A_{n,i'}\),  their Gaifman graphs contain no edge(s), and
  \(\Neighbr{\A_{n,i}}{i} = \A_{n,i}[\set{i}] \cong \A_{n,i'}[\set{i}] = \Neighbr{\A_{n,i'}}{i}\).
  Thus,
  \[
    \Neighbr{\A_{n,i}}{i} \models \psi'(i)
    ~ \iff ~
    \Neighbr{\A_{n,i'}}{i} \models \psi'(i),
   \]
  which implies that \(\A_{n,i} \models \phi(i)\) $\iff$ \(\A_{n,i'} \models \phi(i)\).
  On the other hand, however, by the definition of \(\phi(x)\) and since \(i, i' \in [n]\)
  with \(i \neq i'\), we have \(\A_{n,i} \models \phi(i)\) and \(\A_{n,i'} \not\models \phi(i)\).

  This is a contradiction,
  so there is no \(\FOWplus\) formula \(\psi(x)\)
  that is a Boolean combination of sentences and local formulas
  such that \(\psi(x)\) is equivalent to \(\phi(x)\).
\end{claimproof}
 
\section{Details Omitted in Section~\ref{sec:gaifman}}
\subsection{Proof of Lemma~\ref{lemma:LocalityLemma}}\label{appendix:LocalityLemma}

\LocalityLemma*
\begin{proof}
Note that $\ngFOW$ is a logic in which the \emph{coincidence lemma}
holds.
Thus, for any $\phi\in\ngFOW$, the following it true:
if $\phi(\bar x)$ is $r$-local for some non-empty tuple $\bar x$
that contains all the free variables of $\phi$,
then $\phi(\bar x)$ is $r$-local for \emph{every} non-empty tuple
$\bar x$ that contains all the free variables of $\phi$
(cf.~\cref{lemma:LocalityNotion} in \cref{appendix:LocalityNotion}).
We will tacitly use this throughout the following proof.
\medskip

For \(k = 0\), by definition,
we have \(\loc{q}{0}{d} = \emptyset\),
and there is nothing to prove.
We proceed by induction on $q$ to prove the assertion of the lemma for
$\loc{q}{k}{d}$ for all $k\geq 1$ and $d\geq 0$.
\medskip

For the induction base with $q=0$, consider arbitrary $k\geq 1$ and $d\geq 0$.
By definition, we have \(\loc{0}{k}{d} = \Fparam{0}{k}{d}\).
Concerning the claimed locality of the formulas,
recall from \cref{eq:def:radius_function} that \(\locRad{0}{k}{d} = \max \set{1, d}\).

By definition, $\Fparam{0}{k}{d}$ consists of quantifier-free formulas
that have free variables among $\set{x_1,\ldots,x_k}$ and that are
Boolean combinations of

\begin{enumerate}[(a)]
\item\label{item:LocalityLemma:RA}
relational atoms,
\item\label{item:LocalityLemma:DA}
distance atoms with maximal distance at most $d$,
\item\label{item:LocalityLemma:WA}
weight atoms of the form $\big(s=\weight(y_1,\ldots,y_\ell)\big)$,
with $y_1,\ldots,y_\ell\in\set{x_1,\ldots,x_k}$, $s\in \wtype(\weight)$
\item\label{item:LocalityLemma:ngTA}
  neighbourhood-guarded term-comparison atoms of the form\\
   \(\Big(
       \Pred(t_1, \dots, t_m)
       \land
       \Land_{y \in \bigcup_{i=1}^{m} \free(t_i) \setminus \set{x}} \dist(x,y) \leq d'
       \Bigr)\),
   \ with $d'\leq d$.
\end{enumerate}

\noindent
It is straightforward to see that
\begin{itemize}
\item
  atoms of the form \eqref{item:LocalityLemma:RA} are 0-local
\item
  atoms of the form \eqref{item:LocalityLemma:DA} are $d$-local
\item
  atoms of the form \eqref{item:LocalityLemma:WA} are $1$-local\footnote{We obtain 1-locality, but not 0-locality, due to the particular
   form of the \emph{locality condition} of weighted structures. E.g.,
   $\A$ with $A=\set{1,2,3,4}$, $R^{\A}=\set{(1,2,3,4)}$, and
   $\weight^{\A}(1,2)=s\in S$ with $s\neq 0_S$ is a weighted
   structure that witnesses that the formula
   $\big(s=\weight(x_1,x_2)\big)$ is not $0$-local.
}\item
  atoms of the form \eqref{item:LocalityLemma:ngTA} are
  $\max\set{d,1}$-local.\footnote{With the same reasoning as for the cases
  \eqref{item:LocalityLemma:WA} and \eqref{item:LocalityLemma:DA}, \enquote{$\Pred(t_1, \dots, t_m)$} is
  1-local (but not necessarily 0-local), and the distance atoms are $d$-local.
}\end{itemize}

As locality is preserved under Boolean combinations (without increasing
the radius), this proves that all formulas in $\loc{0}{k}{d}$ are $\locRad{0}{k}{d}$-local.
\medskip

For the induction step from $q$ to $q{+}1$, consider arbitrary
$q,d\in\NN$. The induction hypothesis states that the
assertion of the lemma holds for all formulas in $\loc{q}{k'}{d}$ for
all $k'\in\NN$. Now consider an arbitrary $k\in\NN$ with $k\geq 1$.
Along the particular definition of the sets $\loc{q{+}1}{k}{d}$ and
$\Fparam{q{+}1}{k}{d}$, and by using the induction hypothesis, it is
straightforward to see that $\loc{q{+}1}{k}{d}\subseteq\Fparam{q{+}1}{k}{d}$.

To prove the claimed locality of the formulas in $\loc{q{+}1}{k}{d}$,
consider an arbitrary formula $\phi\in \loc{q{+}1}{k}{d}$.
By the definition of the set $\loc{q{+}1}{k}{d}$, we have
$\free(\phi)\subseteq\set{x_1,\ldots,x_k}$, and $\phi$
is a Boolean combination of formulas of
the forms
\ref{itm:def:loc:kind_previous_rank},
\ref{itm:def:loc:kind_distance_atom},
\ref{itm:def:loc:kind_near},
or
\ref{itm:def:loc:kind_weight_aggregation}.
As locality is preserved under Boolean combinations (without
increasing the radius), it suffices to consider the case that $\phi$
itself is of one of the forms
\ref{itm:def:loc:kind_previous_rank},
\ref{itm:def:loc:kind_distance_atom},
\ref{itm:def:loc:kind_near},
or
\ref{itm:def:loc:kind_weight_aggregation}.

For the remainder of this proof let
\begin{equation*}k'\ \deff \ k{+}\wmaxar
  \quad\text{and}\quad
  \tilde{r}\ \deff\ \locRad{q{+}1}{k}{d}
  \quad\text{and} \quad
  r'\ \deff\ \locRad{q}{k'}{d}.
\end{equation*}
Our goal is to show that $\phi$ is $\tilde{r}$-local.

\begin{itemize}
\item
If $\phi$ is of the form \ref{itm:def:loc:kind_previous_rank},
then $\phi\in\loc{q}{k'}{d}$.
By the induction hypothesis, $\phi$ is $r'$-local.
By \cref{eq:def:radius_function} we have $\tilde{r}=4k r'\geq r'$, and
thus $\phi$ is also $\tilde{r}$-local.
\medskip

\item
If $\phi$ is of the form \ref{itm:def:loc:kind_distance_atom}, we have
\[
  \phi \ = \quad
  \dist(x_i,x_j)\,{\leq}\,\tilde{d}\,,
\]
with $i,j\in[k]$ and $\tilde{d}\leq \tilde{r}$.
Hence, $\phi$ obviously is $\tilde{r}$-local.
\medskip

\item
  If $\phi$ is of the form \ref{itm:def:loc:kind_near}, we have
  \[
    \phi \ = \quad
        \exists x_{k+1}
        \,
        \paren[\big]{
          \
          \dist(x_i,x_{k+1}) \,{\leq}\,
          \hat{d}
          \;\undp\;
          \lambda
          \
        }
  \]
      with $i\in[k]$, \(\lambda \in \loc{q}{k'}{d}\),
      $\hat{d}\leq\hat{r} \deff \tilde{r}-r'$, and $\free(\lambda)\subseteq\free(\phi)\cup\set{x_{k+1}}$.
      Let $\bar x=(x_{i_1},\ldots,x_{i_\ell})$ be a non-empty tuple of variables that contains all the free
      variables of $\phi$. We extend this tuple by the variable
      $x_{k+1}$ to $\bar
      x'\deff (x_{i_1},\ldots,x_{i_\ell},x_{k+1})$.

      The induction hypothesis yields that $\lambda(\bar x')$ is
      $r'$-local around $\bar x'$.

      The formula $\big(\dist(x_i,x_{k+1})\,{\leq}\,\hat{d}
      \;\und\;\lambda\big)$ enforces in its models that the
      $r'$-neighbourhood around $\bar x'$ is contained in the
      $(r'{+}\hat{d})$-neighbourhood around $\bar x$.
      Since $r'{+}\hat{d}\leq \tilde{r}$, this implies that $\phi$ is
      $\tilde{r}$-local.
      \medskip

\item
If $\phi$ is of the form \ref{itm:def:loc:kind_weight_aggregation}, we have
\[
  \phi \ = \quad
\Big(
  s =
  \sum \weight(\bar y).
  \big(\;
     \lambda \; \und \!\!
     \Oder_{i\in I, j\in[\ell]}
 \!\!        \dist(x_i,y_j)\,{\leq}\, r'
  \,\big)
\Big)\,, \quad\text{where}
\]
 $\weight\in\Weights$,
  $S=\wtype(\weight)$ is finite, $s\in S$, $\ell=\ar(\weight)\geq 1$,  $\bar
  y=(y_1,\ldots,y_\ell)=(x_{k+1},\ldots,x_{k+\ell})$,
  $\emptyset\neq I\subseteq [k]$,
  and
  $\lambda\in\loc{q}{k'}{d}$ with $\free(\lambda)\subseteq\free(\phi)\cup\set{x_{k+1},\ldots,x_{k+\ell}}$.

Recall that all weighted structures $\A$ satisfy the \emph{locality
  condition}, and hence $\weight^{\A}(\bar b)=0_S$ holds for all
tuples $\bar b$ that do not form a clique in the Gaifman graph of
$\A$. This implies that
$\phi$ is equivalent to the formula
$\phi'\deff \big(s=\sum\weight(\bar y).\lambda' \big)$ with
\[
  \lambda' \ \deff \quad
  \Big(\;
     \lambda \ \und \!\!
     \Oder_{i\in I, j\in[\ell]}
 \!\!        \dist(x_i,y_j)\,{\leq}\, r'
   \ \und \!\!
      \Und_{1\leq j<j'\leq \ell} \dist(x_{k+j},x_{k+j'})\,{\leq}\,1
 \,\Big)\,.
\]
      Let $\bar x=(x_{i_1},\ldots,x_{i_\ell})$ be a non-empty tuple of variables
      that contains all the free variables of $\phi$.
      We extend $\bar x$ by the variables
      $x_{k+1},\ldots,x_{k+\ell}$ to $\bar
      x'\deff (x_{i_1},\ldots,x_{i_\ell},x_{k+1},\ldots,\allowbreak x_{k+\ell})$.

      The induction hypothesis yields that $\lambda(\bar x')$ is
      $r'$-local around $\bar x'$.
      Since $r'\geq 1$,
      this implies that also $\lambda'(\bar x')$ is $r'$-local.

      Note that the formula $\lambda'$ enforces in its models that the
      $r'$-neighbourhood around $\bar x'$ is contained in the
      $(r'{+}r'{+}1)$-neighbourhood around $\bar x$.
      This implies that the formula $\phi'$ is
      $\paren{2r'{+}1}$-local.
      Since $\tilde{r}=4kr'\geq 2r'{+} 1$, this implies that $\phi'$ is
      $\tilde{r}$-local.
      As $\phi$ is equivalent to $\phi'$, we obtain that
      $\phi$ is $\tilde{r}$-local.
    \end{itemize}
    \smallskip

\noindent
This completes the proof of \cref{lemma:LocalityLemma}.
\end{proof}

\subsection{Properties of \texorpdfstring{$\sent{q}{k}{d}$}{S(q,k,d)}}
\label{appendix:PropertiesOfSqkd}

\begin{lemma}\label{lemma:PropertiesOfSqkd}
Let $q,k,d\in\NN$.
Every formula in $\sent{q}{k}{d}$ is
\begin{itemize}
\item
  an atomic sentence in $\sent{0}{0}{0}$, or
\item
  a local aggregation sentence for $\Fqkd$ of inner radius as most
  $\locRad{q{-}1}{k{+}\wmaxar}{d}$, or
\item
  a basic local sentence for $\Fqkd$ of
  inner radius at most $\locRad{q{-}1}{k{+}\wmaxar}{d}$,
  of width at most $k{+}1{+}(q{-}1)\wmaxar$,
  and
  of inner distance at most $\locRad{q}{k}{d}+2\locRad{q{-}1}{k{+}\wmaxar}{d}$.
\end{itemize}
\end{lemma}
\begin{proof}
We proceed by induction on $q$. For the induction base with $q=0$,
note that for all $k,d\in\NN$, the \emph{atomic sentences in $\Fparam{0}{k}{d}$} are exactly the sentences of the form
\begin{itemize}
\item
  $\true$, $\false$,
\item
  $R()$ for a relation symbol $R\in\sigma$ of arity 0,
\item
  $\big(s=\weight()\big)$ for a weight symbol $\weight\in\Weights$ of arity 0,
  and
\item
  $\Pred(t_1,\ldots,t_m)$ for $\Pred\in\PP$
and terms
$t_1,\ldots,t_m$ with $\bigcup_{i\in[m]}\free(t_i)=\emptyset$.
\end{itemize}
By definition, $\sent{0}{0}{0}$ consists of exactly these atomic sentences, and $\sent{0}{k}{d}=\sent{0}{0}{0}$ for all $k,d\in\NN$.
This proves the assertion of the lemma for the particular case that $q=0$.
\medskip

For the induction step from $q$ to $q{+}1$, consider arbitrary $q,d\in\NN$. The induction hypothesis states that the assertion of the lemma holds for $\sent{q}{k'}{d}$ for all $k'\in\NN$. Now consider an arbitrary $k\in\NN$ and an arbitrary formula $\phi\in\sent{q{+}1}{k}{d}$, and let $k'\deff k{+}\wmaxar$.
\\
We have to show that $\phi$ is
\begin{enumerate}[(a)]
\item\label{item:Sqkd:goal:a}
  an atomic sentence in $\sent{0}{0}{0}$, or
  \item\label{item:Sqkd:goal:b}
  a local aggregation sentence for $\Fparam{q{+}1}{k}{d}$ of inner radius as most
  $\locRad{q}{k'}{d}$, or
\item\label{item:Sqkd:goal:c}
  a basic local sentence for $\Fparam{q{+}1}{k}{d}$ of
  inner radius at most $\locRad{q}{k'}{d}$,
  of width at most $k{+}1{+}q\wmaxar$,
  and
  of inner distance at most $\locRad{q{+}1}{k}{d}+2\locRad{q}{k'}{d}$.
\end{enumerate}

According to the definition of $\sent{q{+}1}{k}{d}$, we know that one of the following cases applies:
\begin{enumerate}
\item\label{item:Sqkd:one}
  $\phi$ is a sentence in $\sent{q}{k'}{d}$.
\item\label{item:Sqkd:two}
  $\phi$ is a sentence of the form
\[
    \exists y_1\cdots\exists y_\ell\ \Big(
      \Und_{1 \leq j < j' \leq \ell} \dist(y_j,y_{j'}) > 2(2c{+}1)r' \undp
      \Und_{j \in [\ell]} \lambda(y_j)
    \ \Big)
\]
  with \(\ell \in [k{+}1]\), $\ell$
  variables $y_1,\ldots,y_\ell$,
  \(c \in [0,k]\),
  \(r' = \locRad{q}{k'}{d}\),
  and \(\lambda \paren{x_1} \in \loc{q}{k'}{d}\).
\item\label{item:Sqkd:three}
  $\phi$ is a sentence of the form
\[
\big(
  s = \sum \weight(y_1,\ldots,y_\ell).\lambda
\big),
\]
where
$\weight\in\Weights$,
$\ell=\ar(\weight)\geq 1$,
the range \(S\) of \(\weight\) is finite,
$s\in S$,
$y_1,\ldots,y_\ell$ are $\ell$
distinct variables,
and
$\lambda(y_1,\ldots,y_\ell)\in\loc{q}{k'}{d}$.
\end{enumerate}
\medskip

\noindent
In \textbf{Case~\ref{item:Sqkd:one}}, we can apply the induction hypothesis and obtain that $\phi$ is
\begin{enumerate}[(i)]
 \item\label{item:Sqkd:IH:i}
   an atomic sentence in $\sent{0}{0}{0}$, or
 \item\label{item:Sqkd:IH:ii}
   a local aggregation sentence for $\Fparam{q}{k'}{d}$ of inner radius at most
   $\locRad{q{-}1}{k'{+}\wmaxar}{d}$, or
 \item\label{item:Sqkd:IH:iii}
   a basic local sentence for $\Fparam{q}{k'}{d}$ of inner radius at most
   $\locRad{q{-}1}{k'{+}\wmaxar}{d}$, of width at most $k'{+}1+(q{-}1)\wmaxar$, and of inner distance at most $\locRad{q}{k'}{d}+2\locRad{q{-}1}{k'{+}\wmaxar}{d}$.
\end{enumerate}

\noindent
In case~\eqref{item:Sqkd:IH:i}, statement~\eqref{item:Sqkd:goal:a} holds and we are done.
\medskip
\\
In case~\eqref{item:Sqkd:IH:ii}, $\phi$ is of the form $\big(s=\sum\weight(\bar y).\lambda\big)$, where $\weight\in\Weights$, $\ell=\ar(\weight)\geq 1$, the range $S$ of $\weight$ is finite, $s\in S$, $\bar y=(y_1,\ldots,y_\ell)$ is a tuple of $\ell$ pairwise distinct variables, and $\lambda(\bar y)\in\Fparam{q{-}1}{k'{+}\wmaxar}{d}$ is $\tilde{r}$-local for some $\tilde{r}\leq \locRad{q{-}1}{k'{+}\wmaxar}{d}$.
Since $\lambda(\bar y)$ has at most $\ell\leq\wmaxar\leq k'$ free variables, we obtain that $\lambda(\bar y)\in \Fparam{q}{k'}{d}$. This implies that $\phi$ is a local aggregation sentence for $\Fparam{q{+}1}{k}{d}$. Its inner radius is at most $\tilde{r}\leq \locRad{q{-}1}{k'{+}\wmaxar}{d}\leq 4k'\locRad{q{-}1}{k'{+}\wmaxar}{d} = \locRad{q}{k'}{d}$. Thus, statement \eqref{item:Sqkd:goal:b} holds and we are done.
\medskip
\\
In case~\eqref{item:Sqkd:IH:iii}, $\phi$ is of the form
$\exists y_1\cdots\exists y_\ell\ \Big(\bigwedge_{1\le i<j\le\ell}
\distAtomNeg{y_i}{y_j}{2r}\ \wedge\ \bigwedge_{j\in [\ell]}\lambda(y_j)\Big)$,
where $\ell\geq 1$, $y_1,\ldots,y_\ell$ are $\ell$
variables, $r\geq \tilde{r}\geq 0$,
$\lambda(x_1)\in\Fparam{q{-}1}{k'{+}\wmaxar}{d}$ is
\emph{$\tilde{r}$-local} and has only one free variable, and
$\ell\leq k'{+}1+(q{-}1)\wmaxar$, $\tilde{r}\leq \locRad{q{-}1}{k'{+}\wmaxar}{d}$, and $2r\leq \locRad{q}{k'}{d}+2\locRad{q{-}1}{k'{+}\wmaxar}{d}$.
Since $\lambda(x_1)$ has only one free variable and $k'\geq \wmaxar\geq 1$, we obtain that $\lambda(x_1)\in\Fparam{q}{k'}{d}$. This implies that $\phi$ is a basic local sentence for $\Fparam{q{+}1}{k}{d}$.
\\
Its inner radius is at most $\tilde{r}\leq \locRad{q{-}1}{k'{+}\wmaxar}{d}\leq 4k'\locRad{q{-}1}{k'{+}\wmaxar}{d} = \locRad{q}{k'}{d}$.
\\
Its width is $\ell \leq k'{+}1+(q{-}1)\wmaxar= k{+}1+q\wmaxar$.
\\
Its inner distance is at most $\locRad{q}{k'}{d}+2\locRad{q{-}1}{k'{+}\wmaxar}{d}$.
Combining this with the following claim
proves that statement  \eqref{item:Sqkd:goal:c} holds.

\begin{claim}
$\locRad{q}{k'}{d}+2\locRad{q{-}1}{k'{+}\wmaxar}{d}
~ \leq ~
\locRad{q{+}1}{k}{d}+2\locRad{q}{k'}{d}$.
\end{claim}
\begin{claimproof}
We show a stronger statement, namely,  $\locRad{q}{k'}{d}+2\locRad{q{-}1}{k'{+}\wmaxar}{d}
~ \leq ~
2\locRad{q}{k'}{d}$, as follows:
\[
 \begin{array}{rcl}
  2\locRad{q}{k'}{d}
  ~ = ~
  \locRad{q}{k'}{d} + \locRad{q}{k'}{d}
  & \stackrel{\eqref{eq:def:radius_function}}{=} &
                                                   \locRad{q}{k'}{d} + 4k'\locRad{q{-}1}{k'{+}\wmaxar}{d}
\\
  & \stackrel{k'\geq 1}{\geq} &
  \locRad{q}{k'}{d} + 2\locRad{q{-}1}{k'{+}\wmaxar}{d}.
\end{array}
\]
\end{claimproof}
In summary, we obtain that statement \eqref{item:Sqkd:goal:c} holds, and we are done with \textbf{Case~\ref{item:Sqkd:one}}.
\bigskip

\noindent
In \textbf{Case~\ref{item:Sqkd:two}}, we obtain from \cref{lemma:LocalityLemma} that $\lambda(x_1)$ belongs to $\Fparam{q}{k'}{d}$ and is \(r'\)-local for $r'=\locRad{q}{k'}{d}$. This implies that
$\phi$ is a basic local sentence for $\Fparam{q{+}1}{k}{d}$ of inner radius at most
$r'$. Furthermore, its width is $\ell\leq k{+}1\leq k{+}1{+}q\wmaxar$, and its inner distance is $2(2c{+}1)r'\leq 2(2k{+}1)r' = 4kr'+2r' = \locRad{q{+}1}{k}{d}+2\locRad{q}{k'}{d}$.
Hence, statement \eqref{item:Sqkd:goal:c} holds, and we are done.
\bigskip

\noindent
In \textbf{Case~\ref{item:Sqkd:three}}, we obtain from \cref{lemma:LocalityLemma} that $\lambda(\bar y)$ belongs to $\Fparam{q}{k'}{d}$ and is \(r'\)-local for $r'=\locRad{q}{k'}{d}$. This implies that
$\phi$ is a local aggregation sentence for $\Fparam{q{+}1}{k}{d}$ of inner radius at most
$r'$.
Hence, statement \eqref{item:Sqkd:goal:b} holds, and we are done.
\bigskip

\noindent
This completes the proof of \cref{lemma:PropertiesOfSqkd}.
\end{proof}
 \subsection{Proof of Lemma~\ref{lemma:fv} (Feferman--Vaught Decomposition)}\label{sec:fv}

We start with an easy lemma that enables us to assume w.l.o.g.\ that the
first entries in $\Delta$ are mutually exclusive.

\begin{lemma}\label{lemma:fv_mux}
  Let \(q,\myk,d \in \NN\) with \(\myk \geq 2\).
Let \(\vec{x} \deff (x_1,\ldots,x_{\myk})\),
\(r' \deff \locRad{q}{\myk}{d}\),
\(K \subseteq \intsUpTo{\myk}\),
and \(I \subseteq K\) with \(\emptyset \neq I \neq K\).
For every finite and non-empty set \(\Delta\) of pairs of formulas
  \(\big(\alpha\paren{\projTup{x}{I}},\beta\paren{\projTup{x}{K \setminus I}}\big)\)
  such that \(\alpha\paren{\projTup{x}{I}},
  \beta\paren{\projTup{x}{K \setminus I}}\in \loc{q}{\myk}{d}\),
  there exists a finite and non-empty set $\tildeDelta$ of pairs of
  formulas
  \(\big(\tildealpha\paren{\projTup{x}{I}},\tildebeta\paren{\projTup{x}{K \setminus I}}\big)\)
  such that \(\tildealpha\paren{\projTup{x}{I}},
  \tildebeta\paren{\projTup{x}{K \setminus I}} \in
  \loc{q}{\myk}{d}\),
  they are Boolean combinations
  of formulas occurring in $\Delta$,
  \begin{equation}\label{eq:lemma:mux}
    \begin{array}{ll}
      & \displaystyle
    \Oder_{(\alpha,\beta) \in \Delta}
   \paren[\Big]{
      \alpha \paren{\projTup{x}{I}}
      \und
      \beta \paren{\projTup{x}{K \setminus I}}
      }
     \undp
     \dist(\projTup{x}{I};\projTup{x}{K \setminus I}) \,{>}\, r'
     \\ \equiv & \displaystyle
        \Oder_{(\tildealpha,\tildebeta) \in \tildeDelta}
   \paren[\Big]{
      \tildealpha \paren{\projTup{x}{I}}
      \und
      \tildebeta \paren{\projTup{x}{K \setminus I}}
      }
     \undp
                 \dist(\projTup{x}{I};\projTup{x}{K \setminus I}) \,{>}\, r'\,,
    \end{array}
  \end{equation}
  and the first entries in $\tildeDelta$ are \emph{mutually exclusive},
  i.e., for every two distinct $(\tildealpha,\tildebeta)$ and $(\tildealphaStrich,\tildebetaStrich)$
  in $\tildeDelta$, the formula $(\tildealpha\und\tildealphaStrich)$ is unsatisfiable.
  Furthermore, there is an algorithm that, given \(\Delta\), computes \(\tildeDelta\).

  Analogously, we can also compute a $\tildeDelta$ for which
  \eqref{eq:lemma:mux} holds, but where, instead of the first entries,
  the second entries in $\tildeDelta$ are mutually exclusive, i.e.,
  for every two distinct $(\tildealpha,\tildebeta)$ and $(\tildealphaStrich,\tildebetaStrich)$
  in $\tildeDelta$, the formula $(\tildebeta\und\tildebetaStrich)$ is unsatisfiable.
\end{lemma}
\begin{proof}
  Let $\myAlphas\deff\setc{\alpha}{\text{there exists a } \beta \text{ such that }
    (\alpha,\beta)\in \Delta}$. For every $\alpha\in\myAlphas$, let
  $\myBetas(\alpha)\deff\setc{\beta}{(\alpha,\beta)\in\Delta}$.
  For every $J\subseteq\myAlphas$, let
  \[
    \alpha_J\ \deff \ \Und_{\alpha\in J}\alpha \ \und  \Und_{\alpha\in
      \myAlphas\setminus J}\nicht\alpha
    \qquad\text{and}\qquad
    \beta_J \ \deff \ \Oder_{\alpha\in J}\Oder_{\beta\in\myBetas(\alpha)}\beta\,.
  \]
  Note that we have $\alpha_J(\tx_I)\in\loc{q}{\myk}{d}$ for every $J\subseteq \myAlphas$;
  and for $J\neq\emptyset$, we also have
  $\beta_J(\tx_{K\setminus I})\in\loc{q}{\myk}{d}$.
  Clearly, the formulas $(\alpha_J)_{J\subseteq\myAlphas}$ are
  mutually exclusive, i.e., for $J,J'\subseteq\myAlphas$ with $J\neq
  J'$, the formula $(\alpha_J\und\alpha_{J'})$ is unsatisfiable.
  Let
  \[
    \tildeDelta\ \deff \ \setc{(\alpha_J,\beta_J)}{J\subseteq\myAlphas,
    \ J\neq\emptyset}.
  \]
  It can easily be verified that \eqref{eq:lemma:mux} is satisfied.

  The analogous statement,
  aiming at a $\tildeDelta$ where the \emph{second} entries are mutually exclusive,
  follows by symmetry.
  This completes the proof of \cref{lemma:fv_mux}.
\end{proof}

\FV*

\begin{proof}
  We proceed by induction on \(q\).

\BaseCasesP
For the induction base with \(q = 0\),
consider arbitrary \(\myk,d \in \NN\) with \(\myk \geq 2\). Note that
$r'\deff \locRad{0}{k'}{d}\stackrel{\eqref{eq:def:radius_function}}{=}\max\set{d,1}$.
Let $\bar x\deff(x_1,\ldots,x_{k'})$, and consider arbitrary $K,I$
such that $K\subseteq [k']$, $I\subseteq K$, and $\emptyset\neq I\neq K$.
Note that any formula $\lambda(\bar x_K)\in
 \loc{0}{\myk}{d}=\Fparam{0}{\myk}{d}$
is a Boolean combination of \emph{atomic} \(\ngFOWplus\) formulas $\psi$
with \(\freeS{\psi} \subseteq\setc{x_j}{j\in K} \)
and \(\md{\psi} \leq d\).
We proceed by induction on the construction of \(\lambda(\bar x_K)\).

\begin{enumerate}[{Case} 1]
  \item
  \label{case:fv_decomp:base_case:all_free_vars_in_component}
  If \(\lambda\) is an atomic \(\ngFOWplus\) formula
  with \(\free(\lambda) \subseteq \setc{x_i}{i \in I}\),
  then we are done by letting
  \[
    \Delta
    \deffp
    \setS*{\paren{\lambda,\, \true}}
    .
  \]
\item\label{case:fv_decomp:base_case:all_free_vars_in_other_comp}
  Otherwise, if \(\lambda\) is an atomic \(\ngFOWplus\) formula with
  \(\freeS{\lambda}\subseteq \setc{x_j}{j \in K\setminus I}\),
  then we are done by letting
  \[
    \Delta
    \deffp
    \setS*{\paren{\true,\, \lambda}}
  .
  \]
\end{enumerate}
Note that atoms of the form \(\true\),
 \(\false\),
 \(R \tup{}\) for a \(0\)-ary relation
symbol \(R \in \sigma\), or of the form
\(\big(s = \weight \tup{}\big)\) for a \(0\)-ary weight symbol \(\weight \in \Weights\)
are already dealt with in Case~\ref{case:fv_decomp:base_case:all_free_vars_in_component}.
For the remaining cases, recall that
\(r' = \max\set{d,1}\).

\begin{enumerate}[{Case} 1]
    \setcounter{enumi}{2}
\item\label{case:fv_decomp:base_case:distance_atom}
  Consider the case that \(\lambda\)
  is of the form \(\distS{}{x_i,x_j} \leq d'\)
  or of the form \(\distS{}{x_j,x_i} \leq d'\)
  with \(i \in I\) and \(j \in K\setminus I\).
Since \(\lambda \in \loc{0}{\myk}{d}\), we know that \(d' \leq d\).
For all \(\tup{\sigma, \Weights}\)-structures \(\A\) and all tuples \(\vec{a} \in A^{\myk}\)
  with \(\distS{\A}{\projTup{a}{I};\projTup{a}{K\setminus I}} > r'\),
  we have \(\distS{\A}{a_{i},a_{j}} > r' \geq d\geq d'\),
so \(\A\not\models\eval{\lambda}{\vec{a}_K}\).
Thus, we are done by letting
  \[
    \Delta
    \deffp
    \setS*{\paren{\false,\ \false}}
    .
  \]
\item
  Consider the case that \(\lambda\) is of the form \(x_i {=} x_j\)
  or of the form \(x_j {=} x_i\)
  with \(i \in I\) and \(j \in K\setminus I\).
Then, for all \(\tup{\sigma, \Weights}\)-structures \(\A\)
  and all tuples \(\vec{a} \in A^{\myk}\)
  with \(\distS{\A}{\projTup{a}{I};\projTup{a}{K\setminus I}} > r'\),
  we have \(\distS{\A}{a_{i},a_{j}} > r' \geq 1\),
and therefore \(a_i\neq a_j\).
  Thus, \(\A\not\models\eval{\lambda}{\vec{a}_K}\),
and we are done by letting
  \[
    \Delta
    \deffp
    \setS*{\paren{\false,\ \false}}
    .
  \]
\item
  Consider the case that \(\lambda\)
  is of the form \(R \tupd{x_{\ell_1}}{x_{\ell_m}}\)
  with \(m = \arS{R}\geq 1\) and \(\enud{\ell_1}{\ell_m} \in K\)
  such that there exist \(i,j \in \intsUpTo{m}\)
  with \(\ell_i \in I\) and \(\ell_j \in K\setminus I\).
Then, for all \(\tup{\sigma, \Weights}\)-structures \(\A\)
  and all tuples \(\vec{a} \in A^{\myk}\)
  with \(\distS{\A}{\projTup{a}{I};\projTup{a}{K\setminus I}} > r'\),
  we have \(\distS{\A}{a_{\ell_i},a_{\ell_j}} > r' \geq 1\),
and therefore \((a_{\ell_1},\ldots,a_{\ell_m})\not \in R^{\A}\).
Hence,
  \(\A \not\models \eval{\lambda}{\vec{a}_K}\),
and we are done by letting
  \[
    \Delta
    \deffp
    \setS*{\paren{\false,\ \false}}
    .
  \]
\item
  Consider the case that \(\lambda\)
  is of the form \(\big(s = \weight \tupd{x_{\ell_1}}{x_{\ell_m}}\big)\)
  with
  \(\weight \in \Weights\),
  \(S = \wtypeS{\weight}\),
  \(s \in S\),
  \(m = \arS{\weight}\),
  and \(\enud{\ell_1}{\ell_m} \in K\)
  such that there exist \(i,j \in \intsUpTo{m}\)
  with \(\ell_i \in I\) and \(\ell_j \in K\setminus I\).
Then, for all \(\tup{\sigma, \Weights}\)-structures \(\A\)
  and all tuples \(\vec{a} \in A^{\myk}\)
  with \(\distS{\A}{\projTup{a}{I};\projTup{a}{K\setminus I}} > r'\),
  we have \(\distS{\A}{a_{\ell_i},a_{\ell_j}} > r' \geq 1\).
Hence,
  by the locality condition on weighted structures,
  we have \(\weight^{\A} \tupd{a_{\ell_1}}{a_{\ell_m}} = 0_S\).
Thus,
  \(\A \models \eval{\lambda}{\vec{a}_K}\) $\iff$ \(s = 0_S\),
and we are done by letting
  \[
    \Delta
    \deffp
    \setS*{\paren{\true,\ \true}}
    ~~
    \text{if } s = 0_S, \qquad
    \text{and}
    ~~
    \Delta
    \deffp
    \setS*{\paren{\false,\ \false}}
    ~~
    \text{otherwise.}
  \]
\item Consider the case that \(\lambda\)
    is of the form
    \[
      \paren[\Big]{
        \Pred \tupd{t_1}{t_m}
        \land\textstyle
        \bigLand_{y \in \bigcup_{\nu=1}^m \freeS{t_\nu} \setminus \set{z}}
        \distS{}{z, y} \leq d'
      }
    \]
    with
    \(\Pred \in \PP\),
    \(m = \arS{\Pred}\),
    \(\SC\)-terms \(\enud{t_1}{t_m}\)
    with \(\ptypeS{\Pred} = \tupd{\wtypeS{t_1}}{\wtypeS{t_m}}\),
    $d'\leq d$,
    $z\in\vars$,
    $\free(\lambda)\subseteq\setc{x_j}{j\in K}$, and there exist $i\in I$ and $j\in K\setminus I$
    such that $x_i,x_j\in\free(\lambda)$.
In particular, \(\abs{\free(\lambda)} \geq 2\), and this implies that the
    conjunction over the distance atoms actually is present in
    $\lambda$, and hence $z\in\free(\lambda)$.
    As \(z \in \free(\lambda)
    \subseteq \setc{x_\nu}{\nu \in I} \uplus \setc{x_\nu}{\nu \in K \setminus I}\),
    we can assume w.l.o.g.\ that \(z = x_i\) or \(z = x_j\).
    Thus, we either have
    \(x_i = z\) and \(x_j \in \bigcup_{\nu=1}^m\free(t_\nu) \setminus \set{z}\),
    or \(x_j = z\) and \(x_i \in \bigcup_{\nu=1}^m\free(t_\nu)\setminus\set{z}\).
Note that, for all \(\tup{\sigma, \Weights}\)-structures \(\A\)
    and all tuples \(\ta \in A^{\myk}\)
    with \(\distS{\A}{\projTup{a}{I};\projTup{a}{K\setminus I}} > r'\),
    we have \(\distS{\A}{a_{i},a_{j}} > r' \geq d \geq d'\).
    Hence, \(\A \not\models \lambda(\ta_K)\),
    and we are done by letting
    \[
      \Delta
      \deffp
      \setS*{\paren{\false,\ \false}}.
    \]
\end{enumerate}
So far, we have proved that the assertion of the lemma holds for all
\emph{atomic} formulas in \(\loc{0}{\myk}{d}\).
To handle Boolean combinations,
it suffices to consider disjunction and negation.

\begin{enumerate}[{Case} 1]
    \setcounter{enumi}{7}
  \item\label{case:proof:fv_decomp:base_case:disjunction}
  Consider the case that \(\lambda\) is of the form \((\lambda_1\oder\lambda_2)\).
By the induction hypothesis, for each
  \(\ell \in \set{1,2}\), we obtain a finite and non-empty set \(\Delta_\ell\) of pairs of
  formulas \(\paren{\alpha(\bar x_I),\beta(\bar x_{K\setminus I})}\)
  such that \(\alpha(\bar x_I), \beta(\bar x_{K\setminus I}) \in \loc{q}{\myk}{d}\),
  their quantifier rank is at most $\qr(\lambda_\ell)$, and
  \[
    \lambda_\ell\paren{\vec{x}_K}
    \undp
    \distS{}{\projTup{x}{I};\projTup{x}{K\setminus I}} > r'
    \equivp
    \bigLorClPad{\paren{\alpha,\beta} \in \Delta_\ell}{0.75 em}
    \paren[\Big]{
      \alpha\paren{\projTup{x}{I}} \und \beta\paren{\projTup{x}{K\setminus I}}
    }
    \undp
    \distS{}{\projTup{x}{I};\projTup{x}{K\setminus I}} > r'.
  \]
  It can easily be verified that we are done by choosing
  \(\Delta \deffp \Delta_1\cup\Delta_2\) and afterwards applying
  Lemma~\ref{lemma:fv_mux} to turn $\Delta$ into a $\tildeDelta$ where the
  first entries in $\tildeDelta$ are mutually exclusive.\bigskip

\item
\label{case:proof:fv_decomp:base_case:negation}
Next, consider the case that \(\lambda\) is of the form
\(\nicht\lambda'\).
By the induction hypothesis, we obtain a finite and non-empty set
\(\Delta'\) of pairs of formulas \(\paren{\alpha(\bar x_I),\beta(\bar
  x_{K\setminus I})}\) such that
\(\alpha(\bar x_I), \beta(\bar x_{K\setminus I}) \in \loc{q}{\myk}{d}\),
their quantifier rank is at most $\qr(\lambda')$, and
\begin{equation}
  \label{eq:proof:fv_decomp:base_case:negation}
  \begin{array}{ll}
 & \lambda'\paren{\vec{x}_K}
  ~ \undp ~
  \distS{}{\projTup{x}{I};\projTup{x}{K\setminus I}} \,{>}\, r'
  \medskip\\ \equivp & \displaystyle
  \Oder_{\paren{\alpha,\beta} \in \Delta'}
  \paren[\Big]{
    \alpha\paren{\projTup{x}{I}} \und \beta\paren{\projTup{x}{K\setminus I}}
  }
  \undp
  \distS{}{\projTup{x}{I};\projTup{x}{K\setminus I}} > r'.
 \end{array}
\end{equation}
Let \(m \deffp |\Delta'|\)
and let \((\alpha_1,\beta_1),\ldots,(\alpha_m,\beta_m)\)
be an enumeration of the elements in \(\Delta'\).
For every set \(H\subseteq\intsUpTo{m}\) with \(\emptyset\neq H\neq \intsUpTo{m}\) let

\newcommand{\landPad}{0.75 em}
\begin{alignat*}{5}
  \alpha_H
  &\deffp
  \bigLandClPad{h \in H}{\landPad} \nicht \alpha_h
  \qquad && \text{and} \qquad
  &
  \beta_H
  & \deffp
  \bigLandClPad{h' \in \intsUpTo{m}\setminus H}{\landPad} \nicht\beta_{h'}
  .
\intertext{We let
        }
  \alpha_\emptyset
  & \deffp
        \true
  \qquad &&\text{and} \qquad
  &
  \beta_\emptyset
  & \deffp
  \bigLandClPad{h' \in \intsUpTo{m}}{\landPad} \nicht \beta_{h'}
\shortintertext{and}
\alpha_{\intsUpTo{m}}
  & \deffp
  \bigLandClPad{h \in \intsUpTo{m}}{\landPad} \nicht \alpha_h
  \qquad &&\text{and} \qquad
  &
  \beta_{\intsUpTo{m}}
  & \deffp
        \true
  .
\shortintertext{We let}
\Delta
  & \deffp
  \mathrlap{
    \setc{\paren{\alpha_H,\, \beta_H}}{ H \subseteq \intsUpTo{m}}
    .
  }
\end{alignat*}
Then, \(\Delta\) is a finite, non-empty set of pairs of formulas
\((\alpha_H,\beta_H)\) such that \(\alpha_H(\bar x_I) \in \loc{q}{\myk}{d}\)
and \(\beta_H(\bar x_{K\setminus I}) \in \loc{q}{\myk}{d}\).
Consider an arbitrary \(\tup{\sigma, \Weights}\)-structure \(\A\)
and a tuple \(\vec{a} \in A^{\myk}\)
such that \(\distS{\A}{\projTup{a}{I};\projTup{a}{K\setminus I}} > r'\).
We have to show that
\[
  \A \models \eval{\lambda}{\vec{a}_K}
  ~~ \iff ~~
  \text{there is a
  \(H\subseteq \intsUpTo{m}\) such that \(\A \models \alpha_H(\projTup{a}{I})\) and
  \(\A \models \beta_H(\projTup{a}{K\setminus I})\)}.
\]
Note that the following holds.
\[
\begin{array}{rcl}
     \A \models \eval{\lambda}{\vec{a}_K}
  & \iff
  & \A\not\models\eval{\lambda'}{\vec{a}_K}\\
  & \stackrel{\eqref{eq:proof:fv_decomp:base_case:negation}}{\iff}
  & \text{for all } h \in \intsUpTo{m}
  ~ : ~
  \A\not\models\alpha_h(\projTup{a}{I})
  \text{ \ or \ }
  \A\not\models\beta_h(\projTup{a}{K\setminus I})
  .
\end{array}
\]
Therefore, it suffices to show that the following statements are
equivalent.
\begin{romanenumerate}
  \item
  \label{itm:proof:fv_decomp:base_case:negation:existence}
  There is an
  \(H\subseteq \intsUpTo{m}\) such that \(\A \models \alpha_H(\projTup{a}{I})\)
  and \(\A \models \beta_H(\projTup{a}{K\setminus I})\).

  \item
  \label{itm:proof:fv_decomp:base_case:negation:universality}
  For all \(h \in \intsUpTo{m}\)
  we have \(\A\not\models\alpha_h(\projTup{a}{I})\)
  or \(\A\not\models\beta_h(\projTup{a}{K\setminus I})\).
\end{romanenumerate}
For the direction
\enquote{(i)\(\Rightarrow\)(ii)},
consider an arbitrary \(h \in \intsUpTo{m}\).
If \(h \in H\),
then, from \(\A \models \alpha_H(\projTup{a}{I})\),
we obtain that \(\A\not\models\alpha_h(\projTup{a}{I})\).
If \(h\not \in H\),
then, from \(\A \models \beta_H(\projTup{a}{K\setminus I})\),
we obtain that \(\A\not\models\beta_h(\projTup{a}{K\setminus I})\).
Hence, (ii) is satisfied.

For the direction \enquote{(ii)\(\Rightarrow\)(i)},
let \(H \deffp \setc{h \in \intsUpTo{m}}{\A\not\models\alpha_h(\projTup{a}{I})}\).
Then, \(\A \models \alpha_H(\projTup{a}{I})\).
Furthermore, for every \(h' \in \intsUpTo{m}\setminus H\)
we have \(\A\not\models\beta_{h'}(\vec{a}_{K\setminus I})\).
Thus, \(\A \models \beta_H(\projTup{a}{K\setminus I})\).
Hence, \(H\) witnesses that (i) is satisfied.

This proves that the statements (i) and (ii) are equivalent.
\smallskip

Applying \cref{lemma:fv_mux}, we can turn $\Delta$ into a
$\tildeDelta$ where the first entries in $\tildeDelta$ are mutually exclusive.
This completes the case that \(\lambda\) is of the form \(\nicht\lambda'\).
\end{enumerate}

In summary, we have finished the induction base for \(q = 0\).

\IndStepP
For the induction step from \(q\) to \(q{+}1\), consider \(q,d \in
\NN\).
\medskip

The \textbf{induction hypothesis} states that the assertion of
\cref{lemma:fv} holds for all $k'\geq 2$, for
$r'=\locRad{q}{k'}{d}$, $\tx'=(x_1,\ldots,x_{k'})$, and for all
$K'\subseteq[k']$ and $I'\subseteq K'$ with $\emptyset\neq I'\neq
K'$. That is, for every $\lambda(\bar x'_{K'})\in\loc{q}{k'}{d}$, we can
compute a finite and non-empty set $\Delta'$ of pairs of formulas
$\big(\alpha(\tx'_{I'}),\beta(\tx'_{K'\setminus I'})\big)$
such that $\alpha(\bar x'_{I'}), \beta(\bar
x'_{K'\setminus I'})\in\loc{q}{k'}{d}$, their quantifier rank is at
most $\qr(\lambda)$,
\[
  \lambda\paren{\proj{\vec{x}'}{K'}}
  \undp
  \dist(\proj{\vec{x}'}{I'}; \proj{\vec{x}'}{K' \setminus I'})> r'
  ~~ \equiv ~~
  \Oder_{\mathclap{\paren{\alpha,\beta} \in \Delta'}}
  ~
  \paren[\Big]{
    \alpha \paren{\proj{\vec{x}'}{I'}}
    \und
    \beta \paren{\proj{\vec{x}'}{K' \setminus I'}}
  }
  \undp
  \dist(\proj{\vec{x}'}{I'}; \proj{\vec{x}'}{K' \setminus I'})> r'
  ,
\]
and the first entries in \(\Delta'\) are mutually exclusive.
\medskip

Now consider an arbitrary \(k \in \NN\) with \(k \geq 2\),
an arbitrary \(K \subseteq \intsUpTo{k}\) with \(|K| \geq 2\),
and an arbitrary \(I \subseteq K\) with \(\emptyset \neq I\neq K\).
Let $\bar x=(x_1,\ldots,x_k)$.
Let
\begin{equation}\label{eq:rs_fv}
 \tilde{r} \deffp \locRad{q{+}1}{k}{d}
  \quad \text{and}\quad
  r'  \deffp \locRad{q}{k{+}\wmaxar}{d}
  \quad
  \text{and}
  \quad
  \hat{r} \deffp \tilde{r} - r'\,.
\end{equation}

Note that the formulas \(\psi(\bar x_K)\in\loc{q{+}1}{k}{d}\)
are Boolean combinations
of formulas that have free variables in \(\setc{x_j}{j\in K}\)
and are of one of the forms
\ref{itm:def:loc:kind_previous_rank},
\ref{itm:def:loc:kind_distance_atom},
\ref{itm:def:loc:kind_near},
or
\ref{itm:def:loc:kind_weight_aggregation}.

Our goal is to construct a finite and non-empty set $\Delta$ of pairs
of formulas $(\alpha(\tx_I),\beta(\tx_{K\setminus I}))$ such that
$\alpha(\tx_I), \beta(\tx_{K \setminus I})\in\loc{q{+}1}{k}{d}$,
their quantifier rank is at most $\qr(\psi)$,
and
\begin{equation}\label{eq:fv-proof_goal_in_IndStep}
  \psi\paren{\proj{\vec{x}}{K}}
  \undp
  \dist(\proj{\vec{x}}{I}; \proj{\vec{x}}{K \setminus I})> \tilde{r}
  ~~ \equiv ~~
  \Oder_{\mathclap{\paren{\alpha,\beta} \in \Delta}}
  ~
  \paren[\Big]{
    \alpha \paren{\proj{\vec{x}}{I}}
    \und
    \beta \paren{\proj{\vec{x}}{K \setminus I}}
  }
  \undp
  \dist(\proj{\vec{x}}{I}; \proj{\vec{x}}{K \setminus I})> \tilde{r}\ ;
\end{equation}
once having achieved this, we can use \cref{lemma:fv_mux} to
turn $\Delta$ into a $\tildeDelta$ such that the first entries in
$\tildeDelta$ are mutually exclusive.
We proceed by induction on the construction of \(\psi(\tx_K)\).

\proofsubparagraph{Case~\ref{itm:def:loc:kind_previous_rank}:}
\(\psi\) is of the form \ref{itm:def:loc:kind_previous_rank}.
Then \(\psi\) belongs to \(\loc{q}{k{+}\wmaxar}{d}\)
(with \(\freeS{\psi} \subseteq \setc{x_j}{j\in K}\)),
and we use the induction hypothesis
with \(\lambda\deff\psi\), \(k' \deffp k{+}\wmaxar\),
\(K' \deffp \intsUpTo{k}\) and $I'\deff I$ to construct a \(\Delta'\).
Setting \(\Delta\deff\Delta'\) finishes this case, because
\(\tilde{r}\geq r'\).

\proofsubparagraph{Case~\ref{itm:def:loc:kind_distance_atom}:}
\(\psi\) is of the form \ref{itm:def:loc:kind_distance_atom}.
Then it is of the form \(\distS{}{x_i, x_j} \leq \tilde{d}\)
with \(i, j \in K\)
and \(\tilde{d} \leq \tilde{r}\).
This case can be dealt with in the same way
as the distance atoms in the base case (Cases
\ref{case:fv_decomp:base_case:all_free_vars_in_component},
\ref{case:fv_decomp:base_case:all_free_vars_in_other_comp},
\ref{case:fv_decomp:base_case:distance_atom}).

\proofsubparagraph{Case~\ref{itm:def:loc:kind_near}:}
\(\psi\) is of the form \ref{itm:def:loc:kind_near}.
Then we have
\begin{equation*}
  \psi \paren{\tx_K}
  ~ = ~
  \exists x_{k+1}\,
  \paren[\big]{
    \, \distS{}{x_i, x_{k+1}} \leq \hat{d}
    \, \und \,
    \lambda \paren{\tx_K, x_{k+1}} \,
  }
\end{equation*}
with \(i \in K\),
\(\lambda \in \loc{q}{k{+}\wmaxar}{d}\), and $\hat{d}\leq\hat{r}$.
We use the induction hypothesis for \(k' \deffp k{+}\wmaxar\),
\(K' \deffp K\cup\set{k{+}1}\),
the formula \(\lambda\),
and each of the sets \(I_1 \deffp I\cup\set{k{+}1}\) and \(I_2 \deffp I\).
This yields finite, non-empty sets \(\Delta_1\) and \(\Delta_2\)
such that, for each \(\ell \in \set{1,2}\), the set \(\Delta_\ell\)
consists of pairs of formulas \(\paren{\alpha,\beta}\)
with \(\alpha \paren{\projTup{x}{I_{\ell}}},
\beta \paren{\projTup{x}{K' \setminus I_{\ell}}} \in \loc{q}{k'}{d}\)
and of quantifier rank $\leq\qr(\lambda)$
such that, for \(\tx' \deffp (x_1,\ldots,x_{k'})\), we have
\begin{equation}
  \label{eq:proof:fv_decomp:case:kind_near:indhyp_on_disjunction}
  \begin{array}{ll}
    & \lambda \paren{\proj{\vec{x}'}{K'}}
    \undp
    \dist(\proj{\vec{x}'}{I_\ell}; \proj{\vec{x}'}{K' \setminus I_\ell}) > r'
    \smallskip\\
    \equivp & \displaystyle
    \Oder_{\mathclap{\paren{\alpha,\beta} \in \Delta_{\ell}}}
    ~
    \paren[\Big]{
      \alpha \paren{\proj{\vec{x}'}{I_\ell}}
      \und
      \beta \paren{\proj{\vec{x}'}{K' \setminus I_\ell}}
    }
    \undp
    \dist(\proj{\vec{x}'}{I_\ell}; \proj{\vec{x}'}{K' \setminus I_\ell}) > r'
    .
  \end{array}
\end{equation}
We use these two index sets \(I_1\) and \(I_2\),
because they allow us to distinguish between the two cases
\(i \in I\) and \(i \in K\setminus I\).
In the first case, we add the quantified variable $x_{k+1}$ to the component indexed by \(I\),
and in the second case, to the other component.
For each \(\paren{\alpha, \beta} \in \Delta_1\),
we define
\begin{align}
  \label{eq:proof:fv_decomp:kind_near__def_alpha_close_to_first_component}
  \alpha' \paren{\projTup{x}{I}}
  & \deffp
  \exists x_{k+1}
  \paren[\big]{
    \distS{}{x_i, x_{k+1}} \leq \hat{d}
    \, \und \,
    \alpha \paren{\projTup{x}{I}, x_{k+1}}
  },
  \intertext{
and for each \(\paren{\alpha, \beta} \in \Delta_2\)
    we define}
  \beta' \paren{\projTup{x}{K\setminus I}}
  & \deffp
  \exists x_{k+1}
  \paren[\big]{
    \distS{}{x_i, x_{k+1}} \leq \hat{d}
    \,\und\,
    \beta \paren{\projTup{x}{K\setminus I},x_{k+1}}
  }
  .
\end{align}
If \(i \in I\), we let
\begin{align*}
  \Delta
  & \deffp
  \Delta_1'
  \deffp
  \setcompS[\Big]{
    \paren[\big]{
      \alpha' \paren{\projTup{x}{I}}
      \;,\;
      \beta \paren{\projTup{x}{K \setminus I}}
      }
  }{
    \paren{\alpha,\beta} \in \Delta_1
  }
  \, ,
\shortintertext{and otherwise we let}
\Delta
  & \deffp
  \Delta_2'
  \deffp
  \setcompS[\Big]{
    \paren[\big]{
      \alpha \paren{\vec{x}_{I}}
      \;,\;
        \beta' \paren{\projTup{x}{K \setminus I}}
    }
  }{\paren{\alpha, \beta} \in \Delta_2}
  \, .
\end{align*}

For all \(\paren{\alpha, \beta} \in \Delta_1 \cup \Delta_2\),
we have \(\alpha, \beta \in \loc{q}{k'}{d}\).
Since they only have free variables in \(\setd{x_1}{x_k}\),
by formation rule~\ref{itm:def:loc:kind_previous_rank},
they are also contained in \(\loc{q {+} 1}{k}{d}\);
and, since \(\hat{d} \leq \hat{r}\),
we also have \(\alpha', \beta' \in \loc{q {+} 1}{k}{d}\)
by formation rule~\ref{itm:def:loc:kind_near}.
Furthermore, it holds that
\(\qrS{\alpha'} = \qrS{\alpha} + 1 \annLeq{ind.\,hyp.} \qrS{\lambda} + 1 = \qrS{\psi}\)
and
\(\qrS{\beta'} = \qrS{\beta} + 1 \annLeq{ind.\,hyp.} \qrS{\lambda} + 1 = \qrS{\psi}\).
Thus, \(\Delta\) is a finite, non-empty set of pairs of formulas
\(\paren{\hat{\alpha}, \hat{\beta}}\)
with \(\hat{\alpha} \in \loc{q{+}1}{k}{d}\)
and \(\hat{\beta} \in \loc{q{+}1}{k}{d}\) and of quantifier
rank at most $\qr(\psi)$.
All that remains to be done is to prove
that equivalence \eqref{eq:fv-proof_goal_in_IndStep} holds.

Consider an arbitrary \(\tup{\sigma, \Weights}\)-structure \(\A\)
and a tuple \(\vec{a} \in A^k\)
with \(\distS{\A}{\projTup{a}{I};\projTup{a}{K\setminus I}} > \tilde{r}\).
We want to show that
\[
  \A\modelsp\psi(\bar a_K)
  ~~ \iff ~~
  \text{there is a }
  (\hat{\alpha},\hat{\beta}) \in \Delta
  \text{ such that }
  \A \models \hat{\alpha}(\bar a_I)
  \text{ and }
  \A \models \hat{\beta}(\bar a_{K\setminus I}).
\]
First, suppose \(i \in I\), so $\Delta=\Delta'_1$.
For the direction \enquote{\(\Longrightarrow\)}, assume that
\(\A \models \eval{\psi}{\vec{a}_K}\).
This implies that there exists an \(a_{k+1} \in A\)
with \(\distS{\A}{a_i, a_{k+1}} \leq \hat{d}\)
and \(\A \models \eval{\lambda}{\vec{a}_K, a_{k+1}}\).
Then, for every \(j \in K \setminus I\),
we have
\[
  \tilde{r}\ <\ \distS{\A}{a_i,a_j}\ \leq\ \distS{\A}{a_i, a_{k+1}} + \distS{\A}{a_{k+1}, a_j}
  \ \leq\ \hat{d} + \distS{\A}{a_{k+1}, a_j}.
\]
Thus, \(\distS{\A}{a_{k+1}, a_j}\,>\,\tilde{r} - \hat{d}
\,\geq\,\tilde{r} - \hat{r}\,\stackrel{\eqref{eq:rs_fv}}{=}\, r'\).
Since this holds for all \(j \in K \setminus I\),
we have shown that \(\distS{\A}{\projTup{a}{K\setminus I}; a_{k+1}} > r'\).
Furthermore, since \(\tilde{r} \geq r'\),
we also have \(\distS{\A}{\projTup{a}{I};\projTup{a}{K\setminus I}} > r'\).
Hence, for \(\ta' \deffp \tupd{a_1}{a_{k+1}}\),
since \(I_1 = I \cup \set{k{+}1}\)
and \(K' \setminus I_1 = K \setminus I\),
we have
\(\distS{\A}{\proj{{\vec{a}}'}{I_1};
\proj{{\vec{a}}'}{K' \setminus I_1}} > r'\).
Thus, from equivalence \eqref{eq:proof:fv_decomp:case:kind_near:indhyp_on_disjunction}
(for \(\ell = 1\))
and from the fact that \(\A \models \eval{\lambda}{\proj{{\vec{a}}'}{K'}}\),
we obtain that there exists a tuple \(\paren{\alpha,\beta} \in \Delta_1\)
such that \(\A \models \eval{\alpha}{\proj{{\vec{a}}'}{I_1}}\)
and \(\A \models \eval{\beta}{\proj{{\vec{a}}'}{K' \setminus I_1}}\).
Let \(\alpha' \paren{\projTup{x}{I}}\) be the formula
associated with \(\alpha \paren{\proj{\vec{x}'}{I_1}}\)
according to \eqref{eq:proof:fv_decomp:kind_near__def_alpha_close_to_first_component},
and note that \(a_{k+1}\) witnesses that \(\A \models \eval{\alpha'}{\projTup{a}{I}}\).

In summary,
\((\alpha',\beta) \in \Delta'_1 = \Delta\),
and we have \(\A \models \eval{\alpha'}{\projTup{a}{I}}\)
and \(\A \models \eval{\beta}{\projTup{a}{K\setminus I}}\).
This completes the direction \enquote{\(\Longrightarrow\)}.
\medskip

For the direction \enquote{\(\Longleftarrow\)},
let \(\paren{\hat{\alpha}, \hat{\beta}} \in \Delta\)
such that \(\A \models \eval{\hat{\alpha}}{\projTup{a}{I}}\)
and \(\A \models \eval{\hat{\beta}}{\vec a_{K\setminus I}}\).
We have to show that \(\A \models \eval{\psi}{\vec{a}_K}\).
Since we are considering the case \(i \in I\),
we know that \(\Delta = \Delta'_1\).
Thus, there is a tuple \(\paren{\alpha,\beta} \in \Delta_1\)
such that \(\hat{\alpha}\) is the formula \(\alpha' \paren{\projTup{x}{I}}\)
associated with \(\alpha \paren{\projTup{x}{I},x_{k+1}}\)
according to \eqref{eq:proof:fv_decomp:kind_near__def_alpha_close_to_first_component},
and \(\hat{\beta}\) is the formula \(\beta \paren{\projTup{x}{K\setminus I}}\).

Since \(\A \models \eval{\alpha'}{\projTup{a}{I}}\),
there exists an \(a_{k+1} \in A\)
such that \(\distS{\A}{a_i, a_{k+1}} \leq \hat{d}\)
and \(\A \models \eval{\alpha}{\projTup{a}{I}, a_{k+1}}\),
so \(\A \models \eval{\alpha}{\proj{\vec{a}'}{I_1}}\),
where $\ta' \deff (a_1,\ldots,a_{k+1})$.
Furthermore,
since \(\A \models \eval{\hat{\beta}}{\projTup{a}{K\setminus I}}\),
and because \(\hat{\beta}\) is the formula \(\beta \paren{\projTup{x}{K\setminus I}}\),
and \(\projTup{x}{K\setminus I} = \proj{\vec{x}'}{K' \setminus I_1}\),
we have \(\A \models \eval{\beta}{\proj{{\vec{a}}'}{K' \setminus I_1}}\).
As above for the direction \enquote{\(\Longrightarrow\)},
having \(\distS{\A}{a_i, a_{k+1}} \leq \hat{d}\)
implies that \(\distS{\A}{\projTup{a}{K \setminus I}; a_{k+1}} > r'\).
Moreover, we already know that
\(\distS{\A}{\projTup{a}{I}; \projTup{a}{K \setminus I}} > \tilde{r} \geq r'\).
Combined, this shows that
\(\dist^{\A}(\proj{{\vec{a}}'}{I_1}; \proj{{\vec{a}}'}{K' \setminus I_1}) > r'\).

Applying equivalence~\eqref{eq:proof:fv_decomp:case:kind_near:indhyp_on_disjunction}
(for \(\ell = 1\)),
since \(\paren{\alpha,\beta} \in \Delta_1\),
we obtain that \(\A \models \eval{\lambda}{\proj{\vec{a}'}{K'}}\).
Since \(\distS{\A}{a_i, a_{k+1}} \leq \hat{d}\),
we obtain that \(a_{k+1}\) witnesses that \(\A \models \eval{\psi}{\vec{a}_K}\).
This completes the proof in case \(i \in I\).
\medskip

In the other case, we have \(i \in K \setminus I\), so $\Delta=\Delta'_2$.
This case follows analogously to the previous case,
with the roles of \(I\) and \(K \setminus I\) interchanged
and using \(\Delta_2\), \(\Delta_2'\)
instead of \(\Delta_1\), \(\Delta_1'\).
\medskip

Now, the induction step is completed for the case that \(\psi\) is of
the form~\ref{itm:def:loc:kind_near}. We proceed with the next case.

\proofsubparagraph{Case~\ref{itm:def:loc:kind_weight_aggregation}:}
\(\psi\) is of the form \ref{itm:def:loc:kind_weight_aggregation}.
Then, we have
\[
  \psi\paren{\projTup{x}{K}} \ = \ \Bigl(
    s =
    \sum \weight(\ty).\underbrace{
    \bigl(\;
      \lambda
      \; \und
      \Oder_{\substack{j\in J,\\ \nu\in[\ell]}}
      \dist(x_j,x_{k+\nu})\,{\leq}\, r'
     \;\bigr)}_{\ffed \ \lambda'(\tx'_{K'})}
  \;\Bigr),
\]
where $\weight\in\Weights$,
$S=\wtype(\weight)$ is finite,
$s\in S$, $\ell=\ar(\weight)\geq 1$,
$\ty=
(x_{k+1},\ldots,x_{k+\ell})$,
$\emptyset\neq J \subseteq K$,
and
$\lambda\in\loc{q}{k{+}\wmaxar}{d}$ is a formula with free variables
among $\setc{x_\kappa}{\kappa\in K'}$, where $K'\deff
K\cup\set{k{+}1,\ldots,k{+}\ell}$.
Let $k'\deff k{+}\wmaxar$ and $\tx'\deff (x_1,\ldots,x_{k'})$.

If $J\cap I\neq \emptyset$, then, for every $t\in S$, we set
\[
  \varrho_t(\tx_K) \ \deff\ \Bigl(
    t =
    \sum \weight(\ty).\underbrace{
      \bigl(\;
        \lambda
        \; \und \!\!
        \Oder_{\substack{j\in J\cap I,\\ \nu\in[\ell]}}
        \dist(x_j,x_{k+\nu})\,{\leq}\, r'
    \;\bigr)}_{\ffed \ \lambda_1(\tx'_{K'})}
  \Bigr).
\]
Note that this formula is similar to $\psi$ but replaces $s$ with $t$
and ensures that
one of the quantified variables
$x_{k+1},\ldots,x_{k+\ell}$
is close to an $x_j$ with $j\in J\cap I$.

Analogously, if $J\setminus I\neq \emptyset$, we set
\[
  \vartheta_t(\tx_K) \ \deff \ \Bigl(
    t =
    \sum \weight(\ty).\underbrace{
      \bigl(\;
        \lambda
        \; \und \!\!
        \Oder_{\substack{j\in J\setminus I,\\ \nu\in[\ell]}}
        \dist(x_j,x_{k+\nu})\,{\leq}\, r'
      \;\bigr)}_{\ffed \ \lambda_2(\tx'_{K'})}
  \Bigr);
\]
this formula ensures that one of the quantified variables
$x_{k+1},\ldots,x_{k+\ell}$
is close to an $x_j$ with $j\in J\setminus I$.
Note that
the formulas $\lambda_1(\tx'_{K'})$ and
$\lambda_2(\tx'_{K'})$ belong to $\loc{q}{k{+}\wmaxar}{d}$ and have
the same quantifier rank as $\lambda$.

Clearly, if $J\cap I=\emptyset$, then $J\setminus I=J\neq \emptyset$,
and $\vartheta_t$ is defined for all $t\in S$. Analogously, if
$J\setminus I=\emptyset$, then $J\subseteq I$ and $J\cap
I=J\neq\emptyset$, and $\varrho_t$ is defined for all $t\in S$.

Let $T\deff \setc{(t_1,t_2)\in S\times S}{t_1 +_S t_2 = s}$ (this set
is finite because $S$ is finite).
Let
\[
  \hat{\psi}
  ~ \deff ~
  \left\{
    \begin{array}{ll}
       \vartheta_s
    & \text{if } J\cap I=\emptyset,
    \\
        \varrho_s
    & \text{if } J\setminus I=\emptyset,
    \\
       \Oder_{(t_1,t_2)\in T} \big( \varrho_{t_1} \und \vartheta_{t_2} \big)
    & \text{otherwise.}
    \end{array}
  \right.
\]

We fix the following notation. For any
 \(\tup{\sigma, \Weights}\)-structure \(\A\)
and any tuple \(\ta \in A^k\)
with \(\distS{\A}{\projTup{a}{I};\projTup{a}{K\setminus I}} > \tilde{r}\),
we let

\begin{itemize}
\item
 $M^\A_{\ta} \ \ \deff \ \bigsetc{\,\tb \in A^\ell\,}
{\,\A \models \lambda'(\tc_K) \text{ and } \weight^\A(\tb) \neq 0_S,
  \text{ where } \tc \deff \ta \tb\,},$
\item
$M^\A_{\ta,1} \ \deff \ \bigsetc{\,\tb \in A^\ell\,}
{\,\A \models \lambda_1(\tc_K) \text{ and } \weight^\A(\tb) \neq 0_S,
  \text{ where } \tc \deff \ta \tb\,},$
if \(J \cap I \neq \emptyset\),
\item
$M^\A_{\ta,2} \ \deff \ \bigsetc{\,\tb \in A^\ell\,}
{\,\A \models \lambda_2(\tc_K) \text{ and } \weight^\A(\tb) \neq 0_S,
\text{ where } \tc \deff \ta \tb\,}$,
if \(J \setminus I \neq \emptyset\).
\end{itemize}

\noindent
Clearly, we have
\begin{equation*}
  \A \models \psi(\ta_K)
  ~ \iff ~
  s = \sum_{\tb \in M^\A_{\ta}} \weight^\A(\tb);
\end{equation*}
and for all \(t \in S\) we have
\begin{equation}
  \label{eq:fv:varrhot}
  \A \models \varrho_t(\ta_K)
  ~ \iff ~
  t = \sum_{\tb \in M^\A_{\ta,1}} \weight^\A(\tb)
\end{equation}
(if $J\cap I\neq \emptyset$), and
\begin{equation*}
  \A \models \vartheta_t(\ta_K)
  ~ \iff ~
  t = \sum_{\tb \in M^\A_{\ta,2}} \weight^\A(\tb)
\end{equation*}
(if $J\setminus I\neq \emptyset$).

\begin{claim}\label{claim:fv_agg_M}
  Let \(\A\) be a \(\tup{\sigma, \Weights}\)-structure, and
  let \(\ta \in A^k\)
  with \(\distS{\A}{\projTup{a}{I};\projTup{a}{K\setminus I}} >
  \tilde{r}\).
  \\
  If \(J \cap I \neq \emptyset\),
  then, for all \(\tb \in M^\A_{\ta,1}\)
  we have \(\distS{\A}{\projTup{a}{I}\tb; \projTup{a}{K \setminus I}}
  > r'\).
  \\
  If \(J \setminus I \neq \emptyset\),
  then, for all \(\tb \in M^\A_{\ta,2}\)
  we have \(\distS{\A}{\projTup{a}{I}; \projTup{a}{K \setminus I}\tb}
  > r'\).
  \\
  If \(J \cap I \neq \emptyset\) and \(J \setminus I \neq \emptyset\),
  then we have \(M^\A_{\ta} = M^\A_{\ta,1} \uplus M^\A_{\ta,2}\).
\end{claim}
\begin{claimproof}
  First, suppose \(J \cap I \neq \emptyset\),
  and let \(\tb \in M^\A_{\ta,1}\).
  Then there exist $i\in J\cap I$ and \(\nu \in [\ell]\) such that
  \(\distS{\A}{a_i, b_\nu} \leq r'\).
  Moreover, for all \(\mu \in [\ell]\), since \(\weight^\A(\tb) \neq 0_S\),
  we have \(\distS{\A}{b_\nu, b_\mu} \leq 1\) by the \emph{locality
    condition}
  that is satisfied by the weighted structure \(\A\).

  We have to show that
  \(\distS{\A}{\projTup{a}{I}\tb; \projTup{a}{K\setminus I}} > r'\),
  and we already know that
  \(\distS{\A}{\projTup{a}{I}; \projTup{a}{K\setminus I}}
  > \tilde{r} \geq r'\).
  All that remains to be done is to consider arbitrary
  \(\kappa \in K \setminus I\) and \(\mu \in [\ell]\)
  and show that
  \(\distS{\A}{b_\mu, a_\kappa} > r'\).
  Since \(i\in I\) and \(\kappa \in K\setminus I\),
  we have
  \(\distS{\A}{a_i, a_\kappa} > \tilde{r}\).
  Thus,
  \[
    \tilde{r}
    ~ < ~
    \dist^{\A}(a_i,a_\kappa)
    ~ \leq ~
    \dist^{\A}(a_i,b_\nu) + \dist^{\A}(b_\nu,b_\mu) +
    \dist^{\A}(b_\mu,a_\kappa)
    ~ \leq ~
    r'+1+\dist^{\A}(b_\mu,a_\kappa)\,.
  \]
  Hence, we have
  \(
    \distS{\A}{b_\mu, a_\kappa}
    >
    \tilde{r}-r'-1
    \overset{\eqref{eq:def:radius_function}}{=}
    4kr' - r' - 1
    \geq
    r'
  \).
  This holds for all \(\mu \in [\ell]\)
  and all \(\kappa \in K \setminus I\).
  In summary, we obtain that
  \(\distS{\A}{\projTup{a}{I}\tb; \projTup{a}{K \setminus I}} > r'\).

  For the case \(J \setminus I \neq \emptyset\),
  we can analogously prove
  \(\distS{\A}{\projTup{a}{I}; \projTup{a}{K \setminus I}\tb} > r'\)
  for all \(\tb \in M^\A_{\ta, 2}\).
  Lastly, for the case \(J \cap I \neq \emptyset\)
  and \(J \setminus I \neq \emptyset\),
  it is easy to see that \(M^\A_{\ta} = M^\A_{\ta,1} \cup M^\A_{\ta,2}\).
  Moreover, if \(\tb \in M^\A_{\ta,1}\),
  then \(\distS{\A}{\tb; \projTup{a}{K \setminus I}} > r'\),
  which implies that \(\tb \not\in M^\A_{\ta,2}\),
  Hence, \(M^\A_{\ta,1} \cap M^\A_{\ta,2} = \emptyset\),
  so \(M^\A_{\bar a} = M^\A_{\ta,1} \uplus M^\A_{\ta,2}\).
\end{claimproof}

\begin{claim}\label{claim:fv_agg_one}
  \(\psi(\tx_K)\; \und\; \dist(\projTup{x}{I}; \projTup{x}{K\setminus I})
  \,{>}\,\tilde{r}
  \ \ \equiv
  \ \ \hat{\psi}(\projTup{x}{K}) \;\und\; \dist(\projTup{x}{I};\projTup{x}{K\setminus I})
  \,{>}\,\tilde{r}\)\,.
\end{claim}
\begin{claimproof}
  If $J\cap I=\emptyset$, then $J=J\setminus I$ and $\psi=\vartheta_s=\hat{\psi}$.
  If $J\setminus I =\emptyset$, then $J=J\cap I$ and
  $\psi=\varrho_s=\hat{\psi}$.
  Hence, in both cases the claim holds trivially.

  Now consider the remaining case that $J\cap I\neq \emptyset$
  and $J\setminus I\neq \emptyset$.
  Then, for an arbitrary \(\tup{\sigma, \Weights}\)-structure \(\A\)
  and a tuple \(\ta \in A^k\)
  with \(\distS{\A}{\projTup{a}{I};\projTup{a}{K\setminus I}} > \tilde{r}\),
  we have \(\A \models \psi(\ta_K)\) $\iff$
  \(s = \sum_{\tb \in M^\A_{\ta}} \weight^\A(\tb)
  \overset{\text{\cref{claim:fv_agg_M}}}{=}
  \sum_{\tb \in M^\A_{\ta,1}} \weight^\A(\tb)
  \ +_S \ \sum_{\tb \in M^\A_{\ta,2}} \weight^\A(\tb)\),
  where \(M^\A_{\ta}\), \(M^\A_{\ta,1}\),
  and \(M^\A_{\ta,2}\) are as defined above.
  Thus, \(\A \models \psi(\ta_K)\) $\iff$
  there are \(t_1, t_2 \in S\) with \(s = t_1 +_S t_2\),
  \(t_1 = \sum_{\tb \in M^\A_{\ta,1}} \weight^\A(\tb)\), and
  \(t_2 = \sum_{\tb \in M^\A_{\ta,2}} \weight^\A(\tb)\).
  The latter holds if and only if
  \(\A \models \varrho_{t_1}(\ta_K)\)
  and
  \(\A \models \vartheta_{t_2}(\ta_K)\).
  According to our choice of $T$ and $\hat{\psi}$, we obtain that
  $\A\models\psi(\bar a_K)\iff\A\models\hat{\psi}(\bar a_K)$.
\end{claimproof}

The next two claims provide the main step needed for finishing the proof
for the case that $\psi$ is of the form
\ref{itm:def:loc:kind_weight_aggregation}.

\begin{claim}
  \label{claim:fv_agg_varrho}
  If $J\cap I\neq \emptyset$, then, for every $t\in S$,
  we can compute a finite, non-empty set $\Delta_t$ of pairs of formulas
  $\big(\alpha(\tx_I),\beta(\tx_{K\setminus I})\big)$
  such that $\alpha(\tx_I)\in\loc{q{+}1}{k}{d}$,
  $\beta(\tx_{K\setminus I})\in\loc{q{+}1}{k}{d}$, their quantifier
  rank is $\leq \qr(\varrho_t)$, and
  \[
    \varrho_t\paren{\proj{\vec{x}}{K}}
    \undp
    \dist(\proj{\vec{x}}{I}; \proj{\vec{x}}{K \setminus I})> \tilde{r}
    ~~ \equiv ~~
    \Oder_{\mathclap{\paren{\alpha,\beta} \in \Delta_t}}
    ~~
    \paren[\Big]{
      \alpha \paren{\proj{\vec{x}}{I}}
      \und
      \beta \paren{\proj{\vec{x}}{K \setminus I}}
    }
    \undp
    \dist(\proj{\vec{x}}{I}; \proj{\vec{x}}{K \setminus I})> \tilde{r}\,.
  \]
\end{claim}
\begin{claimproof}
  We use the induction hypothesis for
  $k'= k{+}\wmaxar$,
  $K'= K\cup\set{x_{k+1},\ldots,x_{k+\ell}}$, the formula
  $\lambda_1(\tx'_{K'})$, and
  the set $I_1\deff I\cup\set{x_{k+1},\ldots,x_{k+\ell}}$.
  This yields a finite, non-empty set \(\Delta'_1\)
  that
  consists of pairs of formulas \(\paren{\alpha,\beta}\)
  with \(\alpha \paren{\projTup{x}{I_{1}}} \in \loc{q}{k'}{d}\) and
  \(\beta \paren{\projTup{x}{K' \setminus I_{1}}} \in \loc{q}{k'}{d}\)
  and of quantifier rank $\leq \qr(\lambda_1)$,
  such that, for \(\vec{x}' = (x_1,\ldots,x_{k'})\), we have
\begin{equation}
    \label{eq:proof:fv_decomp:case:kind_weight_aggregation:indhyp_one}
    \begin{array}{ll}
    & \lambda_1 \paren{\proj{\vec{x}'}{K'}}
    \undp
    \dist(\proj{\vec{x}'}{I_1}; \proj{\vec{x}'}{K' \setminus I_1}) > r'
    \smallskip\\
    \equivp & \displaystyle
    \Oder_{\mathclap{\paren{\alpha,\beta} \in
                                  \Delta'_{1}}} \quad
    \paren[\Big]{
      \alpha \paren{\proj{\vec{x}'}{I_1}}
      \und
      \beta \paren{\proj{\vec{x}'}{K' \setminus I_1}}
    }
    \undp
    \dist(\proj{\vec{x}'}{I_1}; \proj{\vec{x}'}{K' \setminus I_1}) > r'
            .
          \end{array}
  \end{equation}
Note that,
  by our choice of $K'$ and $I_1$, we have
  $K'\setminus I_1=K\setminus I$,
  so $\tx'_{K'\setminus I_1} = \tx_{K\setminus I}$.
  Furthermore, we have $\tx'_{I_1} = \tx_I \ty$,
  where $\ty=(x_{k+1},\ldots,x_{k+\ell})$.
Our choice of $I_1$ reflects the fact that the formula $\varrho_t$
  aims at tuples $\ty$ where at least one of the variables in $\bar
  y$ is close to an $x_j$ with $j\in J\cap I\subseteq I$.
By \cref{lemma:fv_mux}, we can assume w.l.o.g.\ that
  the \emph{second} entries in $\Delta'_1$ are mutually exclusive.

  Now fix a $t\in S$.
  For each $(\alpha,\beta)\in\Delta'_1$, we define
  \[
    \alpha_t(\tx_I)
    ~ \deff ~
    \Big(
        t =
        \sum \weight(\ty).
        \big(\;
          \alpha(\tx'_{I_1})
          \; \und
          \Oder_{\substack{j\in J\cap I,\\ \nu\in[\ell]}} \dist(x_j,x_{k+\nu})\,{\leq}\, r'
        \;\big)
    \Big).
  \]

  Let \(\tilde{\Delta}_t \deff \bigsetc{(\alpha_t,\beta)}
  {(\alpha,\beta) \in \Delta'_1}\),
  \(\myBetas \deff \setc{\beta}{\text{there exists an } \alpha
  \text{ such that } (\alpha,\beta)\in \Delta'_1}\), and
  \[
    \Delta_t\ \deff\ \begin{cases}
    \tilde{\Delta}_t
    & \text{if }  t\neq 0_S,\\
    \tilde{\Delta}_t \,\cup\,\bigset{\big(\true\,,\,\Land_{\beta\in\myBetas} \neg\beta\big)}
    & \text{if } t=0_S.
    \end{cases}
  \]
  Note that $\Delta_t$ is a finite, non-empty set of pairs of formulas
  $(\alpha',\beta')$ such that $\alpha'(\tx_I)\in\loc{q{+}1}{k}{d}$ and
  $\beta'(\tx_{K\setminus I})\in\loc{q{+}1}{k}{d}$ and their
  quantifier rank is $\leq \qr(\varrho_t)$.

  All that remains to be done is to prove
  that the equivalence stated in the claim holds.
  Consider a \(\paren{\sigma,\Weights}\)-structure \(\A\)
  and a tuple \(\ta \in A^k\) such that
  \(\distS{\A}{\ta_I; \ta_{K\setminus I}} > \tilde{r}\).
  We have to show that
  \begin{equation}\label{eq:fv:varrho:goal}
    \A\models\varrho_t(\ta_K)
    \iff
    \text{there is an } (\alpha',\beta')\in\Delta_t \text{ with }
    \A\models\alpha'(\ta_I)
    \text{ and }
    \B\models\beta'(\ta_{K\setminus I}).
  \end{equation}

  Since the second entries in \(\Delta'_1\) are mutually exclusive,
  there exists at most one \(\beta \in \myBetas\)
  such that \(\A \models \beta(\projTup{a}{K\setminus I})\).
  Hence, one of the following two cases applies.

  \bigskip
  \noindent
  \emph{Case~1:}
  There is exactly one \(\beta \in \myBetas\)
  such that \(\A \models \beta(\projTup{a}{K \setminus I})\).
In this case,
  the particular definition of the notion of
  \enquote{the second entries in \(\Delta'_1\) being mutually exclusive}
  implies that there exists exactly one \(\alpha\)
  such that \((\alpha,\beta) \in \Delta'_1\).
  Using this \(\alpha\), we let
  \[\lambda_\alpha(\tx'_{I_1}) \ \deff \ \alpha(\tx'_{I_1})
  \land \Lor_{\substack{j  \in J \cap I,\\ \nu \in [\ell]}}
  \distAtom{x_j}{x_{k + \nu}}{r'}\]
  and
  \[N^\A_{\ta} \ \deff \ \bigsetc{\,\tb \in A^\ell\,}
  {\,\A \models \lambda_\alpha(\tc_{I_1}) \text{ and } \weight^\A(\tb) \neq 0_S,
    \text{ where } \tc \deff \ta \tb\,}.\]
In the following, we will show that
$N^{\A}_{\ta}=M^{\A}_{\ta,1}$.
Note that
  \begin{equation}
    \label{eq:fv:alphat}
    \A\models\alpha_t(\ta_I)
    ~ \iff ~
    t = \sum_{\tb \in N^\A_{\ta}} \weight^{\A}(\tb).
  \end{equation}
  Analogously to \cref{claim:fv_agg_M},
  we can show that
  \begin{equation}
    \label{eq:fv:agg_dist_N}
    \distS{\A}{\projTup{a}{I}\tb; \projTup{a}{K \setminus I}} > r',
    ~
    \text{ for all }
    \tb \in N^\A_{\ta}.
  \end{equation}
Now let \(\tb \in A^\ell\) with \(\weight^\A(\tb) \neq 0_S\),
  and set \(\tc \deff \ta \tb\).
  Then, by \cref{claim:fv_agg_M},
  we have
  \[\tb \in M^\A_{\ta, 1} \ \iff \ \A \models \lambda_1(\tc_K) \text{ and }
  \distS{\A}{\projTup{c}{I_1}; \projTup{c}{K' \setminus I_1}} > r'.\]
  According to our choice of \(\Delta_t\),
  the tuple \((\alpha_t, \beta)\) is the only \((\alpha', \beta') \in \Delta_t\)
  such that \(\A \models \beta'(\projTup{a}{K\setminus I})\).
  Therefore,
  by \cref{eq:proof:fv_decomp:case:kind_weight_aggregation:indhyp_one},
  we have
  \[
    \A \models \lambda_1(\tc_K) \text{ and }
    \distS{\A}{\projTup{c}{I_1}; \projTup{c}{K' \setminus I_1}} > r'
    ~ \iff ~
    \A \models \alpha(\tc_{I_1}) \text{ and }
    \distS{\A}{\projTup{c}{I_1}; \projTup{c}{K' \setminus I_1}} > r'.
  \]
  Moreover, if \(\A \models \lambda_1(\tc_K)\),
  then, in particular, by the definition of \(\lambda_1\),
  there are \(j \in J \cap I\) and \(\nu \in [\ell]\)
  such that \(\distS{\A}{a_j, b_\nu} \leq r'\).
  Hence,
  \[
    \A \models \lambda_1(\tc_K) \text{ and }
    \distS{\A}{\projTup{c}{I_1}; \projTup{c}{K' \setminus I_1}} > r'
    ~ \iff ~
    \A \models \lambda_\alpha(\tc_{I_1}) \text{ and }
    \distS{\A}{\projTup{c}{I_1}; \projTup{c}{K' \setminus I_1}} > r'.
  \]
  Finally, by \cref{eq:fv:agg_dist_N}, we have
  \[
    \A \models \lambda_\alpha(\tc_{I_1}) \text{ and }
    \distS{\A}{\projTup{c}{I_1}; \projTup{c}{K' \setminus I_1}} > r'
    \ \iff \ \tb \in N^\A_{\ta}.
  \]
  All in all, this shows that
  \(N^\A_{\ta} = M^{\A}_{\ta,1}\).

  Combined with \cref{eq:fv:alphat,eq:fv:varrhot},
  this proves the following equivalence:
  \[
    \begin{array}{rcl}
      \text{there is an \((\alpha',\beta') \in \Delta_t\)
        with \(\A \models \alpha'(\ta_I)\)
      and \(\A \models \beta'(\ta_{K\setminus I}\))}
      &
      \!\!\!\!\iff
      & \A\models\alpha_t(\ta_I)
      \\
      &
      \!\!\!\!\stackrel{\eqref{eq:fv:alphat}}{\iff}
      & t = \sum_{\tb \in N^{\A}_{\ta}} \weight^{\A}(\tb)
      \\
      &
      \!\!\!\!\!\!\stackrel{N^{\A}_{\ta}=M^{\A}_{\ta,1}}{\iff}\!\!
      & t = \sum_{\tb \in M^{\A}_{\ta,1}} \weight^{\A}(\tb)
      \\
      &
      \!\!\!\!\stackrel{\eqref{eq:fv:varrhot}}{\iff}
      & \A \models \varrho_t(\ta_K).
    \end{array}
  \]
  This proves that equivalence \eqref{eq:fv:varrho:goal} holds and
  completes Case~1.

  \bigskip
  \noindent
  \emph{Case~2:}
  For all \(\beta \in \myBetas\),
  we have
  \(\A \not\models \beta(\ta_{K\setminus I})\).
  As above, for every \(\tb \in A^\ell\) with \(\weight^\A(\tb) \neq 0_S\)
  and \(\tc \deff \ta\tb\),
  by \cref{claim:fv_agg_M},
  we have that
    \[\tb \in M^\A_{\ta, 1} \ \iff \ \A \models \lambda_1(\tc_K) \text{ and }
  \distS{\A}{\projTup{c}{I_1}; \projTup{c}{K' \setminus I_1}} > r'.\]
  By \cref{eq:proof:fv_decomp:case:kind_weight_aggregation:indhyp_one},
  since there is no \(\beta \in \myBetas\)
  with \(\A \models \beta(\projTup{a}{K \setminus I})\),
  this implies that
  \(M^\A_{\ta, 1} = \emptyset\).
  Using \cref{eq:fv:varrhot} yields that
  \[
    \A\models\varrho_t(\ta_K)\ \iff \ t=0_S.
  \]
  Now, recall our choice of $\Delta_t$. If $t=0_S$, then
  $\Delta_t$ contains the tuple $\big(\true\,,
  \Und_{\beta\in\myBetas}\nicht\beta\,\big)$, and we have
  $\A\models\true$ and $\A\models \Und_{\beta\in\myBetas}\nicht\beta\,
  (\ta_{K\setminus I})$.
  Hence, equivalence
  \eqref{eq:fv:varrho:goal} holds.
  On the other hand, if $t\neq 0_S$, then all tuples in $\Delta_t$ are
  of the form $(\alpha',\beta')$ for some $\beta'\in\myBetas$. Thus, since we
  are in Case~2, there exists no tuple $(\alpha',\beta')\in\Delta_t$
  with $\A\models\beta'(\ta_{K\setminus I})$.
  Again, equivalence \eqref{eq:fv:varrho:goal} holds,
  and Case~2 is completed.

  This completes the proof of \cref{claim:fv_agg_varrho}.
\end{claimproof}

\begin{claim}
  \label{claim:fv_agg_vartheta}
  If $J\setminus I\neq \emptyset$, then, for every $t\in S$,
  we can compute a finite, non-empty set $\altDelta_t$ of pairs of formulas
  $\big(\alpha(\tx_I),\beta(\tx_{K\setminus I})\big)$
  such that $\alpha(\tx_I)\in\loc{q{+}1}{k}{d}$,
  $\beta(\tx_{K\setminus I})\in\loc{q{+}1}{k}{d}$, their quantifier
  rank is $\leq \qr(\vartheta_t)$, and
  \[
    \vartheta_t\paren{\proj{\vec{x}}{K}}
    \undp
    \dist(\proj{\vec{x}}{I}; \proj{\vec{x}}{K \setminus I})> \tilde{r}
    ~~ \equiv ~~
    \Oder_{\mathclap{\paren{\alpha,\beta} \in \altDelta_t}}
    ~~
    \paren[\Big]{
      \alpha \paren{\proj{\vec{x}}{I}}
      \und
      \beta \paren{\proj{\vec{x}}{K \setminus I}}
    }
    \undp
    \dist(\proj{\vec{x}}{I}; \proj{\vec{x}}{K \setminus I})> \tilde{r}\,.
  \]
\end{claim}
\begin{claimproof}
  We proceed analogously as in our proof of \cref{claim:fv_agg_varrho},
  but we use the induction hypothesis for the formula \(\lambda_2\)
  and the set \(I_2 \deff I\),
  yielding a finite, non-empty set \(\Delta'_2\).
  Our choice of \(I_2\) reflects the fact that the formula \(\vartheta_t\)
  aims at tuples \(\ty\) where at least one of the variables in \(\ty\)
  is close to an \(x_j\) with \(j \in J \setminus I \subseteq K \setminus I\).
  By \cref{lemma:fv_mux}, we can assume w.l.o.g.\ that
  the \emph{first} entries in \(\Delta'_2\) are mutually exclusive.

  Now fix a \(t \in S\).
  For each \((\alpha,\beta) \in \Delta'_2\), we define
  \[
    \beta_t(\tx_{K\setminus I})
    ~ \deff ~
    \Big(
        t =
        \sum \weight(\ty).
        \big(\;
          \beta(\tx'_{K'\setminus I})
          \; \und
          \Oder_{\substack{j\in J\setminus I,\\ \nu\in[\ell]}} \dist(x_j,x_{k+\nu})\,{\leq}\, r'
        \;\big)
    \Big).
  \]
  Let \(\tilde{\altDelta}_t \deff
  \bigsetc{(\alpha,\beta_t)}{(\alpha,\beta) \in \Delta'_2}\),
  \(\myAlphas \deff \setc{\alpha}{\text{there exists a } \beta
  \text{ such that } (\alpha,\beta)\in \Delta'_2}\),
  and
  \[
    \altDelta_t\ \deff \ \begin{cases}
      \tilde{\altDelta}_t
       & \text{if }  t\neq 0_S,\\
       \tilde{\altDelta}_t \,\cup\,\bigset{\,\big(\Und_{\alpha\in\myAlphas} \neg\alpha\,,\,\true\,\big)}
       & \text{if } t=0_S.
    \end{cases}
  \]
  Note that \(\altDelta_t\) is a finite, non-empty set of pairs of formulas
  $(\alpha',\beta')$ such that $\alpha'(\tx_I)\in\loc{q{+}1}{k}{d}$ and
  $\beta'(\tx_{K\setminus I})\in\loc{q{+}1}{k}{d}$ and their
  quantifier rank is $\leq \qr(\vartheta_t)$.

  All that remains to be done is to prove that the equivalence stated in
  \cref{claim:fv_agg_vartheta} holds.
  The proof follows along the same lines as the proof of
  \cref{claim:fv_agg_varrho},
  and this completes the proof of \cref{claim:fv_agg_vartheta}.
\end{claimproof}

We are now ready to finish the proof for the case that \(\psi\)
is of the form \ref{itm:def:loc:kind_weight_aggregation}.
Revisit \cref{claim:fv_agg_one} and the particular formula \(\hat{\psi}\).
If $J\cap I=\emptyset$ or $J\setminus I=\emptyset$,
then $\hat{\psi}$ is one of the formulas handled by
\cref{claim:fv_agg_varrho,claim:fv_agg_vartheta}, and we are done.
In the remaining case, $\hat{\psi}$ is equivalent to
$\Oder_{(t_1,t_2)\in T}\nicht (\nicht \varrho_{t_1}\oder \nicht\vartheta_{t_2})$.
Starting with the sets $\Delta_{t_1}$ and $\altDelta_{t_2}$
obtained for all $(t_1,t_2)\in T$
from \cref{claim:fv_agg_varrho,claim:fv_agg_vartheta} for the formulas
$\varrho_{t_1}$ and $\vartheta_{t_2}$, we can use the same
constructions for disjunctions and negations that we already used in
the induction base, see the Cases
\ref{case:proof:fv_decomp:base_case:disjunction}
and
\ref{case:proof:fv_decomp:base_case:negation}.
This, finally, completes the induction step for the case that $\psi$ is of
the form \ref{itm:def:loc:kind_weight_aggregation}.

To entirely complete the induction step,
all that remains to be done is to consider Boolean combinations,
i.e., disjunctions and negations.
These can be handled in exactly the same way as we had handled them in the
induction base, see the Cases
\ref{case:proof:fv_decomp:base_case:disjunction}
and
\ref{case:proof:fv_decomp:base_case:negation}.
In summary, this completes the proof of \cref{lemma:fv}.
\end{proof}
 
\subsection{Detailed Proof of Theorem~\ref{thm:GNF_ngFOWplus_refined}
  (Rank-Preserving Gaifman Normal
  Form)}\label{appendix:DetailedGNFproof}

\GNFngFOWplusRefined*

\begin{proof}
We proceed by induction on \(q\).

For the induction base with \(q = 0\),
consider arbitrary \(k, d \in \NN\),
and consider an arbitrary \(\phi \in \Fparam{0}{k}{d}\).
If \(k \geq 1\),
then \(\Fparam{0}{k}{d} = \loc{0}{k}{d}\),
and we are done by choosing \(\phi' \deff \phi\).
If \(k = 0\),
then \(\phi\) is a Boolean combination
of \emph{atomic} formulas \(\psi\) in \(\ngFOWplus\)
with \(\free(\psi) = \emptyset\),
so \(\phi\) is a Boolean combination of sentences in \(\sent{0}{k}{d}\)
and we are done by choosing \(\phi' \deff \phi\).

For the induction step from \(q\) to \(q{+}1\),
consider arbitrary \(q, d \in\NN\).
The induction hypothesis states that the assertion of the theorem
holds for all \(k' \in \NN\) and all formulas in~\(\Fparam{q}{k'}{d}\).
Now consider an arbitrary \(k \in \NN\)
and a formula \(\phi \in \Fparam{q{+}1}{k}{d}\);
in particular, \(\free(\phi) \subseteq \set{x_1, \dots, x_k}\).
For the remainder of this proof,
let
\begin{equation}\label{eq:GNF-choice-of-r}
  \tilde{r}\ \deff\ \locRad{q{+}1}{k}{d}
  \quad\text{and} \quad
  r'\ \deff\ \locRad{q}{k{+}\wmaxar}{d}
  \quad\text{and}\quad
  \hat{r} \ \deff \ \tilde{r}-r'.
\end{equation}
By the definition of \(\Fparam{q{+}1}{k}{d}\),
it suffices to consider the following 5 cases.

\proofsubparagraph{Case~(\ref{def:fparam:previous}):}
\(\phi \in \Fparam{q}{k{+}\wmaxar}{d}\).
By the induction hypothesis,
we obtain an equivalent formula \(\phi'\)
with \(\free(\phi') = \free(\phi)\)
such that \(\phi'\) is a Boolean combination of formulas
in \(\sent{q}{k{+}\wmaxar}{d} \cup \loc{q}{k{+}\wmaxar}{d}\),
and the outer (inner) quantifier rank of \(\phi'\)
is at most \(\qr(\phi)\) (resp. \(\qr(\phi){-}1\)).
Since \(\free(\phi') = \free(\phi) \subseteq \set{x_1, \dots, x_k}\),
by the definition of \(\sent{q{+}1}{k}{d}\) and \(\loc{q{+}1}{k}{d}\),
we obtain that \(\phi'\) is a Boolean combination
of formulas in \(\sent{q{+}1}{k}{d} \cup \loc{q{+}1}{k}{d}\),
and we are done.

\proofsubparagraph{Case~(\ref{def:fparam:distance_atom}):}
\(\phi\) is of the form \(\distAtom{x_i}{x_j}{\tilde{d}}\)
with \(\tilde{d} \leq \tilde{r}\) and \(i,j \in [k]\).
Then \(\phi \in \loc{q{+}1}{k}{d}\),
and we are done by choosing \(\phi' \deff \phi\).

\proofsubparagraph{Case (\ref{def:fparam:existential_with_distance}):}
\(\phi\) is of the form
\(\exists x_{k+1}\bigl(\distAtom{x_i}{x_{k+1}}{\hat{d}} \land \phi_1\bigr)\)
with \(\phi_1 \in \Fparam{q}{k{+}\wmaxar}{d}\),
\(i \in [k]\),
\(\hat{d} \leq \hat{r}\),
and
\(\free(\phi_1) \subseteq \free(\phi) \cup \set{x_{k+1}}
\subseteq \set{x_1, \dots, x_k, x_{k+1}}\).

By applying the induction hypothesis to \(\phi_1\),
we obtain an equivalent formula \(\phi_1'\)
with \(\free(\phi_1') = \free(\phi_1)\)
such that \(\phi_1'\) is a Boolean combination of formulas
in \(\sent{q}{k{+}\wmaxar}{d} \cup \loc{q}{k{+}\wmaxar}{d}\),
the outer quantifier rank of \(\phi'_1\) is \(\leq \qr(\phi_1)=\qr(\phi)-1\),
and the inner quantifier rank of \(\phi'\) is
\(\leq \qr(\phi_1)-1 \leq \qr(\phi)-2\).
We transfer \(\phi_1'\) into disjunctive normal form
in order to obtain an equivalent formula of the form
\(\Lor_{j = 1}^n \bigl(\chi_j \land \lambda_j\bigr)\),
where \(n \geq 1\) and,
for all \(j \in \intsUpTo{n}\),
\(\chi_j\) is a Boolean combination of sentences in
\(\sent{q}{k{+}\wmaxar}{d}\) of inner quantifier rank \(\leq \qr(\phi)-2\),
and \(\lambda_j \in \loc{q}{k{+}\wmaxar}{d}\) with
\(\free(\lambda_j) \subseteq \free(\phi_1')\)
and \(\qr(\lambda_j) \leq \qr(\phi)-1\).

Then, \(\phi\) is equivalent to
\(
  \exists x_{k+1} \Bigl(
    \distAtom{x_i}{x_{k+1}}{\hat{d}}
    \,\land\,
    \Lor_{j = 1}^n \bigl(\chi_j \land \lambda_j\bigr)
  \Bigr),
\)
which, in turn, is equivalent to
\(
  \phi' \deff \Lor_{j = 1}^n
  \Bigl(\chi_j\, \land\, \exists x_{k+1}
    \bigl(\distAtom{x_i}{x_{k+1}}{\hat{d}}\, \land \lambda_j\bigr)
  \Bigr)
\).
By definition,
\(\sent{q}{k{+}\wmaxar}{d} \subseteq \sent{q{+}1}{k}{d}\).
Furthermore,
for each \(j \in [n]\),
\(\exists x_{k+1}\bigl(\distAtom{x_i}{x_{k+1}}{\hat{d}}
\,\land \lambda_j\bigr)\)
is a formula in \(\loc{q{+}1}{k}{d}\)
and of quantifier rank \(\leq \qr(\phi)\).
Hence, \(\phi'\) has the desired properties, and we are done.

\proofsubparagraph{Case~(\ref{def:fparam:existential_plain}):}
\(\phi\) is of the form \(\exists x_{k+1}\,\phi_1\) with
\(\phi_1 \in \Fparam{q}{k{+}\wmaxar}{d}\) and
\(\free(\phi_1) \subseteq \free(\phi) \cup \set{x_{k+1}}
\subseteq \set{x_1, \dots, x_k, x_{k+1}}\).

By applying the induction hypothesis to \(\phi_1\),
we obtain an equivalent formula \(\phi_1'\)
with \(\free(\phi_1') = \free(\phi_1)\)
such that \(\phi_1'\) is a Boolean combination of formulas
in \(\sent{q}{k{+}\wmaxar}{d} \cup \loc{q}{k{+}\wmaxar}{d}\),
the outer quantifier rank of \(\phi_1'\) is \(\leq \qr(\phi_1)=\qr(\phi)-1\),
and the inner quantifier rank of \(\phi'\) is
\(\leq \qr(\phi_1)-1 \leq \qr(\phi)-2\).
We transfer \(\phi_1'\) into disjunctive normal form
in order to obtain an equivalent formula of the form
\(\Lor_{j = 1}^n \bigl(\chi_j \land \lambda_j\bigr)\),
where \(n \geq 1\) and,
for all \(j \in \intsUpTo{n}\),
\(\chi_j\) is a Boolean combination of sentences in
\(\sent{q}{k{+}\wmaxar}{d}\) of inner quantifier rank \(\leq \qr(\phi)-2\),
and \(\lambda_j \in \loc{q}{k{+}\wmaxar}{d}\) with
\(\free(\lambda_j) \subseteq \free(\phi_1')\)
and \(\qr(\lambda_j) \leq \qr(\phi)-1\).

Then, \(\phi\) is equivalent to
\(\exists x_{k+1} \Lor_{j = 1}^n \bigl(\chi_j \land \lambda_j\bigr)\),
which, in turn, is equivalent to
\(\Lor_{j = 1}^n \bigl(\chi_j \land \exists x_{k+1} \lambda_j\bigr)\).
By definition,
\(\sent{q}{k{+}\wmaxar}{d} \subseteq \sent{q{+}1}{k}{d}\).
Therefore, all that remains to be done is to transform,
for each \(j \in [n]\),
the formula \(\exists x_{k+1} \lambda_j\) into an equivalent formula
that is a Boolean combination of formulas in
\(\sent{q{+}1}{k}{d} \cup \loc{q{+}1}{k}{d}\)
and that has outer (inner) quantifier rank at most
\(\qr(\phi)\) (resp.~\(\qr(\phi)-1\)).
Consider an arbitrary \(j \in [n]\),
and let \(\lambda \deff \lambda_j\).
Note that
\(\free(\lambda) \subseteq \free(\phi_1)
\subseteq \free(\phi) \cup \set{x_{k+1}}\)
and \(\qr(\lambda) \leq \qr(\phi)-1\).

If \(\free(\phi) = \emptyset\),
then \(\free(\lambda) \subseteq \set{x_{k+1}}\).
Thus, in this case,
\(\exists x_{k+1}\lambda(x_{k+1})\) is a sentence in \(\sent{q{+}1}{k}{d}\)
of inner quantifier rank at most \(\qr(\phi)-1\),
and we are done.\footnote{In fact, \(\exists x_{k+1}\lambda(x_{k+1})\)
  is a basic local sentence of width \(\ell=1\),
  which is why the assertion on the distance
  of the quantified variables vanishes.}
Henceforth, we consider the case that
\(\free(\phi) \neq \emptyset\);
in particular this means that \(k \geq 1\).
Let \(k'\deff k{+}\wmaxar\), \(\tx \deff (x_1, \dots, x_{k'})\),
and \(I \deff \setc{i \in [k]}{x_i \in \free(\phi)}\).

Clearly, \(\exists x_{k+1} \lambda\) is equivalent to
\(\psi \lor \Lor_{i \in I} \vartheta_i\)
with \(\vartheta_i \deff \exists x_{k+1}
\bigl(\distAtom{x_i}{x_{k+1}}{r'}\, \land \lambda\bigr)\) for all \(i \in I\),
and \(\psi \deff \exists x_{k+1}
\bigl(\dist(\tx_I;x_{k+1}) \, {>} \, r'\, \land \lambda\bigr)\).
By definition, for every \(i \in I\),
we have \(\vartheta_i \in \loc{q{+}1}{k}{d}\).
Hence, all that remains to be done is to transform \(\psi\)
into an equivalent formula that is a Boolean combination of formulas
in \(\sent{q{+}1}{k}{d} \cup \loc{q{+}1}{k}{d}\)
and that has outer (inner) quantifier rank at most
\(\qr(\phi)\) (resp.~\(\qr(\phi)-1\)).

By applying \cref{lemma:fv} with \(K \deff I \cup \set{k{+}1}\),
we obtain a finite and non-empty set \(\Delta\) of pairs of formulas
\(\bigl(\alpha(\bar x_I),\beta(x_{k+1})\big)\) such that
\(\alpha(\bar x_I) \in \loc{q}{k'}{d}\),
\(\beta(x_{k+1}) \in \loc{q}{k'}{d}\),
they have quantifier rank \(\leq \qr(\lambda) \leq \qr(\phi)-1\),
and
\[
  \lambda(\tx_K) \; \land\; \dist(\tx_I;x_{k+1})\,{>}\,r'
  ~~ \equiv ~~
  \Lor_{(\alpha,\beta) \in \Delta} \Bigl(
    \alpha(\bar x_I) \land \beta(x_{k+1})
  \Bigr)
  \;\land\;
  \dist(\tx_I;x_{k+1})\,{>}\,r'.
\]
This implies that \(\psi\) is equivalent to
\(\Lor_{(\alpha,\beta) \in \Delta}\Bigl(
  \alpha(\tx_I) \land \exists x_{k+1}\bigl(
    \dist(\tx_I;x_{k+1})\,{>}\,r' \land \beta(x_{k+1})
  \bigr)
\Bigr)\).

Note that \(\alpha(\tx_I) \in \loc{q{+}1}{k}{d}\),
and all that remains to be done is to consider an arbitrary \(\beta\)
such that \((\alpha, \beta) \in \Delta\) for some \(\alpha\),
and transform the formula
\[
  \mu(\tx_I)
  ~~ \deff ~~
  \exists x_{k+1} \bigl(
   \dist(\tx_I;x_{k+1})\,{>}\, r' \ \land \ \beta(x_{k+1})
  \bigr)
\]
into an equivalent formula \(\mu'(\tx_I)\)
that is a Boolean combination of formulas
in \(\sent{q{+}1}{k}{d} \cup \loc{q{+}1}{k}{d}\)
and that has outer (inner) quantifier rank at most
\(\qr(\phi)\) (resp.~\(\qr(\phi)-1\)).
Recall that \(\beta(x_{k+1}) \in \loc{q}{k'}{d}\)
with \(k'=k{+}\wmaxar\),
and \(\qr(\beta) \leq \qr(\phi)-1\).

For constructing \(\mu'(\tx_I)\), we use the following formulas:
\begin{itemize}
  \item
    \(\displaystyle
    \delta_i(\tx_I) \deff
    \exists x_{k+1} \bigl(
      \distAtom{x_i}{x_{k+1}}{\hat{r}}
      \landp
      \dist(\tx_I; \tx_{k+1})\, {>} \, r'
      \landp
      \beta(x_{k+1})
    \bigr)\),
    for all \(i \in I\),
  \item
    \(\displaystyle
    \psi_\ell^{(c)}(\tx_I) \deff
    \Lor_{\substack{J \subseteq I,\\ \abs{J} = \ell}}
    \paren[\Big]{
      \Land_{\substack{j,j' \in J,\\ j \neq j'}}
      \!\!\!\dist(x_j,x_{j'}) \, {>}\, 4cr'
      \; \land \;
      \Land_{j \in J} \exists x_{k+1} \,
      \paren[\big]{
        \distAtom{x_j}{x_{k+1}}{r'} \landp \beta(x_{k+1})
      }
    }\),
    for all \(\ell, c \in [\abs{I}]\), and
  \item
    \(\displaystyle
    \xi_\ell^{(c)} \deff
    \exists y_1 \cdots \exists y_\ell\Big(\!
      \Land_{1 \leq j < j' \leq \ell}
      \!\!\! \dist(y_j,y_{j'})\,{>}\,2(2c{+}1)r'
      \land
      \Land_{j \in [\ell]}\!\beta(y_j)
    \Big)\),
    for all \(\ell \in[\abs{I}{+}1]\) and \(c \in[0, \abs{I}]\).
\end{itemize}

Note that \(\delta_i(\tx_I) \in \loc{q{+}1}{k}{d}\) for all \(i \in I\),
and
\(\xi_\ell^{(c)} \in \sent{q{+}1}{k}{d}\)
for all \(\ell \in[\abs{I}{+}1]\) and \(c \in [0,\abs{I}]\).
Furthermore,
for all \(c \in [\abs{I}]\), we have \(c \leq k\),
so \(4cr' \leq \tilde{r}\) by \cref{eq:def:radius_function}.
Thus, the formula \(\distAtom{x_j}{x_{j'}}{4cr'}\) and its negation
\(\dist(x_j, x_{j'}) \,{>}\, 4cr'\) belong to \(\loc{q{+}1}{k}{d}\).
Moreover, \(r' \leq 4kr' - r' = \tilde{r} - r' = \hat{r}\),
so the formula
\(\exists x_{k+1} \, \paren[\big]{
  \distAtom{x_j}{x_{k+1}}{r'} \landp \beta(x_{k+1})
}\)
belongs to \(\loc{q{+}1}{k}{d}\).
This shows that
\(\psi_\ell^{(c)}(\tx_I) \in \loc{q{+}1}{k}{d}\)
for all \(\ell, c\in [\abs{I}]\).
We let
\[
  \mu'(\bar x_I) \quad \deff \quad
  \xi^{(0)}_{\abs{I}+1} \lorp \Lor_{c,\ell \in [\abs{I}]} \paren[\Big]{
    \xi^{(c)}_\ell \landp \neg\,\xi^{(c{-}1)}_{\ell+1}
    \landp
    \paren[\big]{
      \Lor_{i\in I} \delta_i(\bar x_I)
      \lorp
      \neg\,\psi^{(c)}_\ell(\bar x_I)
    }
  }.
\]
Clearly, \(\mu'\) is a Boolean combination of formulas in
\(\sent{q{+}1}{k}{d} \cup \loc{q{+}1}{k}{d}\),
\(\free(\mu') = \setc{x_i}{i \in I} = \free(\phi)\),
and \(\mu'\) has outer quantifier rank \(\leq \qr(\phi)\)
and inner quantifier rank \(\leq \qr(\phi)-1\).
All that remains to be done is to show that \(\mu'(\bar x_I)\)
is equivalent to \(\mu(\bar x_I)\).
To this end,
consider an arbitrary \((\sigma,\Weights)\)-structure \(\A\)
and a tuple \(\bar{a} \in A^k\).
We have to show that
\(\A \models \eval{\mu'}{\bar{a}_I} \iff \A \models \eval{\mu}{\bar{a}_I}\).

In the proof, the following observation will be used twice.

\distanceByDoubleTriangle*

\begin{claimproof}
  By applying the triangle inequality twice, we obtain
  \(
    s + 2t
    <
    \distS{\A}{d_1, d_2}
    \leq
    \distS{\A}{d_1, c_1} +
    \distS{\A}{c_1, c_2} +
    \distS{\A}{c_2, d_2}
    \leq
    t + \distS{\A}{c_1, c_2} + t
  \).
  Rearranging the inequality, we get the desired result
  \(\distS{\A}{c_1, c_2} > s + 2t - 2t = s\).
\end{claimproof}

We now focus on proving that
\(\A \models \eval{\mu'}{\ta_I} \iff \A \models \eval{\mu}{\ta_I}\).
Throughout the remainder of this proof,
we let \(B \deff \setc{b \in A}{\A \models \eval{\beta}{b}}\).

\proofsubparagraph{Case~(\ref{def:fparam:existential_plain}).1:}
\(\A \models \xi^{(0)}_{\abs{I}+1}\).
Then, \(\A \models \eval{\mu'}{\ta_I}\).
We have to show that \(\A \models \eval{\mu}{\ta_I}\).

Since \(\A \models \xi^{(0)}_{\abs{I}+1}\),
there exist \(\enud{b_1}{b_{\abs{I}+1}} \in B\) of pairwise distance \(> 2r'\).
Let \(\pi\) be a mapping that maps every
\(b \in \set{\enud{b_1}{b_{\abs{I}+1}}}\)
to an element \(a \in \setc{a_i}{i \in I}\) that is closest to \(b\),
that is, an element that minimises \(\distS{\A}{b,a}\),
where ties are broken arbitrarily.
Then, by the pigeonhole principle,
there are two distinct \(j, j' \in \intsUpTo{\abs{I}{+}1}\)
with \(\pi(b_j) = \pi(b_{j'}) = a_i\) for some \(i \in I\).
Since \(\distS{\A}{b_j, b_{j'}} > 2r'\),
by the triangle inequality,
we know that \(\distS{\A}{a_i, b_j} > r'\)
or \(\distS{\A}{a_i, b_{j'}} > r'\).
Say, w.l.o.g., \(\distS{\A}{a_i, b_j} > r'\).
By the choice of \(\pi\),
we obtain that \(\distS{\A}{\ta_I; b_j} > r'\).
This \(b_j \in B\) witnesses that \(\A \models \eval{\mu}{\ta}\).
 
\proofsubparagraph{Case~(\ref{def:fparam:existential_plain}).2:}
\(\A \not\models \xi^{(0)}_1\).
Then, there does not exist any \(b \in A\) such that
\(\A \models \eval{\beta}{b}\), so \(B = \emptyset\).
Hence, \(\A \not\models \eval{\mu}{\ta_I}\).
Furthermore, \(\A \not\models \xi^{(0)}_{\abs{I}+1}\) and, moreover,
for all \(c, \ell \in \intsUpTo{\abs{I}}\),
we have that \(\A \not\models \xi^{(c)}_\ell\).
Thus, \(\A \not\models \eval{\mu'}{\ta_I}\),
so we have that \(\A \not\models \eval{\mu'}{\ta_I}\)
and \(\A \not\models \eval{\mu}{\ta_I}\).
 
\proofsubparagraph{Case~(\ref{def:fparam:existential_plain}).3:}
We are neither in
Case~(\ref{def:fparam:existential_plain}).1
nor in
Case~(\ref{def:fparam:existential_plain}).2.
Then, we have that
\(\A \models \xi^{(0)}_1 \land \neg \xi^{(0)}_{\abs{I}+1}\),
so \(B \neq \emptyset\),
but \(B\) does not contain \(\abs{I}{+}1\) elements
of pairwise distance \(> 2r'\).

\rpGaifmanScatteredThreshold*
\begin{claimproof}
  For every \(c \in [0,\abs{I}]\),
  let \(\ell^{(c)} \in \NN \cup \set{\infty}\) be maximal
  such that there exists a set \(X \subseteq B\) of \(\ell^{\paren{c}}\) elements
  of pairwise distance \(> 2(2c{+}1)r'\).
From \(B \neq \emptyset\),
  we obtain that \(\ell^{\paren{c}} \geq 1\) for each \(c \in \intsUpTo{0,\abs{I}}\).
Since $B$ does not contain $\abs{I}{+}1$ elements of pairwise distance
  $>2r'$, we know that
  \(\ell^{\paren{c}} \leq \abs{I}\) for each \(c \in \intsUpTo{0,\abs{I}}\).
  Furthermore, for all \(c \geq 1\), we have \(\ell^{\paren{c}} \leq \ell^{(c-1)}\),
  because \(2(2c{+}1)r' \geq 2(2(c{-}1){+}1)r'\).
Hence, we have
  $\abs{I} \geq \ell^{\paren{0}} \geq \ell^{\paren{1}}  \geq
    \cdots \geq  \ell^{\paren{\abs{I}}} \geq 1$.
  By the pigeonhole principle,
  there exists a \(c \in \intsUpTo{\abs{I}}\)
  such that \(\ell^{(c-1)} = \ell^{\paren{c}}\).
We choose such a \(c\) and let \(\ell \deff \ell^{\paren{c}}\).
  Since \(\ell = \ell^{\paren{c}}\), we have that \(\A \models \xi^{(c)}_\ell\).
  Since \(\ell = \ell^{(c-1)}\) and \(\ell^{(c-1)}\) is maximal,
  we have that  \(\A \models \neg\,\xi^{(c-1)}_{\ell+1}\).
\end{claimproof}
 
Consider arbitrary \(c, \ell \in \intsUpTo{\abs{I}}\)
such that \(\A \models \xi^{(c)}_\ell \land \neg \xi^{(c-1)}_{\ell+1}\),
which exist due to \cref{claim:proof:rp_gaifman:scattered_threshold}.

\proofsubparagraph{Case~(\ref{def:fparam:existential_plain}).3.1:}
\(\A \models \eval{\delta_i}{\ta_I}\) for some \(i \in I\).
This implies that \(\A \models \eval{\mu'}{\ta_I}\)
and \(\A \models \eval{\mu}{\ta_I}\).
 
\proofsubparagraph{Case~(\ref{def:fparam:existential_plain}).3.2:}
\(\A \models \eval{\neg\,\psi^{(c)}_\ell}{\ta_I}\).
Then \(\A \models \eval{\mu'}{\ta_I}\).
We have to show that \(\A \models \eval{\mu}{\ta_I}\).

From \(\A \models \xi^{(c)}_\ell\),
we know that there exist \(\enud{b_1}{b_\ell} \in B\)
of pairwise distance \(> 2(2c{+}1)r'\).
If there is a \(\nu \in \intsUpTo{\ell}\)
such that \(\distS{\A}{\ta_I; b_\nu} \,{>}\, r'\),
then \(b_\nu\) serves as a witness certifying that
\(\A \models \eval{\mu}{\ta_I}\).
\emph{For contradiction},
assume that, for each \(\nu \in \intsUpTo{\ell}\),
we do \emph{not} have
\(\distS{\A}{\ta_I;b_\nu} \,{>}\, r'\).
Then, for every \(\nu \in \intsUpTo{\ell}\),
there is an \(i(\nu) \in I\)
such that \(\distS{\A}{a_{i(\nu)}, b_\nu} \,{\leq}\, r'\).
 
\rpGaifmanPwFarNearElements*
\begin{claimproof}
  For distinct \(\nu, \nu' \in [\ell]\),
  apply \cref{claim:distance_by_double_triangle}
  with \(s = 4cr'\),
  \(t = r'\),
  and \((c_1, c_2, d_1, d_2) = (a_{i(\nu)}, a_{i(\nu')}, b_\nu, b_{\nu'})\).
  This yields
  \(\distS{\A}{c_1, c_2} \,{>}\, s\),
  so
  \(\distS{\A}{a_{i(\nu)}, a_{i(\nu')}} \,{>}\, 4cr'\).
\end{claimproof}
 
\Cref{claim:proof:rp_gaifman:pw_far_near_elements} implies that
\(i(\nu) \neq i(\nu')\) holds for all distinct \(\nu, \nu' \in [\ell]\).
Thus, for \(J \deff \setc{i(\nu)}{\nu\in[\ell]}\),
we have \(\abs{J} = \ell\).
Furthermore, for each \(\nu \in \intsUpTo{\ell}\),
the element \(b_\nu\) serves as a witness certifying that
\(
(\A,\ta_I) \models \exists x_{k+1}\,\paren[\big]{
  \distAtom{x_{i(\nu)}}{x_{k+1}}{r'}
  \landp
  \beta(x_{k+1})
}
\).
Combining this with
\cref{claim:proof:rp_gaifman:pw_far_near_elements} proves
that \(J\) serves as a witness certifying that
\(\A \models \eval{\psi^{(c)}_\ell}{\ta_I}\).
However, recall that we currently consider the case that
\(\A \models \neg\,\eval{\psi^{(c)}_\ell}{\ta_I}\),
a contradiction!
Thus, there must exist a \(\nu \in \intsUpTo{\ell}\)
such that \(\distS{\A}{\ta_I; b_\nu} \,{>}\, r'\).
This \(b_\nu\) serves as a witness certifying that
\(\A \models \eval{\mu}{\ta_I}\).
 
\proofsubparagraph{Case~(\ref{def:fparam:existential_plain}).3.3:}
For all \(c, \ell \in \intsUpTo{\abs{I}}\)
with \(\A \models \xi^{(c)}_\ell \land \neg\xi^{(c-1)}_{\ell+1}\),
we have \(\A \not\models \eval{\delta_i}{\ta_I}\) for all \(i\in I\),
and \(\A \models \eval{\psi^{(c)}_\ell}{\ta_I}\).
This is the only remaining subcase of
Case~(\ref*{def:fparam:existential_plain}).3.
Since we are still in Case~(\ref*{def:fparam:existential_plain}).3,
the assumptions imply \(\A \not\models\eval{\mu'}{\ta_I}\).
Hence, we have to show that
\(\A \not\models \eval{\mu}{\ta_I}\).

Let us fix \(c, \ell \in \intsUpTo{\abs{I}}\)
according to \cref{claim:proof:rp_gaifman:scattered_threshold}.
Thus, we have \(\A \models \xi^{(c)}_\ell \land \neg\xi^{(c-1)}_{\ell+1}\),
\(\A \not\models \eval{\delta_i}{\ta}\) for all \(i \in I\),
and
\(\A \models \eval{\psi^{(c)}_\ell}{\ta}\).
\emph{For contradiction}, assume that \(\A \models \eval{\mu}{\ta_I}\).
Then there exists a \(b_0 \in B\) such that
\(\distS{\A}{\ta_I; b_0} \,{>}\, r'\).
Since, for all \(i \in I\),
we have \(\A \not\models \eval{\delta_i}{\ta_I}\),
we obtain
\begin{equation}
  \label{eq:proof:rp_gaifman:distance_between_tuple_and_far_element}
  \distS{\A}{\ta_I; b_0} > \hat{r}.
\end{equation}
From \(\A \models \eval{\psi^{(c)}_\ell}{\ta_I}\),
we obtain that there exists a set \(J \subseteq I\)
with \(\abs{J} = \ell\) such that
\begin{equation}
  \label{eq:proof:rp_gaifman:distance_within_tuple}
  \distS{\A}{a_j, a_{j'}}
  >
  4cr'
\end{equation}
holds for all distinct \(j, j' \in J\);
and for each \(j \in J\), there exists a \(b_j \in B\) with
\begin{equation}
  \label{eq:proof:rp_gaifman:distance_between_tuple_entry_and_partner}
  \distS{\A}{a_j, b_j}
  \leq
  r'.
\end{equation}

We would like to use the elements \(b_0, (b_j)_{j \in J}\) as witnesses
certifying that \(\A \models \xi^{(c-1)}_{\ell+1}\).
All that remains to be done
is to show that they have pairwise distance \(> 2(2(c{-}1){+}1)r'\).
Note that \(2(2(c{-}1){+}1)r'= 4cr'-2r'\).

For distinct \(j, j' \in J\),
we apply \cref{claim:distance_by_double_triangle}
with \(s = 4cr' - 2r'\),
\(t = r'\),
and \((c_1, c_2, d_1, d_2) = (b_j, b_{j'}, a_j, a_{j'})\).
The assumptions of the claim are met due to
\cref{eq:proof:rp_gaifman:distance_within_tuple,eq:proof:rp_gaifman:distance_between_tuple_entry_and_partner}.
This yields
\(\distS{\A}{c_1, c_2} \,{>}\, s\),
so
\(\distS{\A}{b_j, b_{j'}} \,{>}\, 4cr' - 2r'\).

All that remains to be done is to show that
\(\distS{\A}{b_j, b_0} > 4cr'-2r'\) for all \(j \in J\).
Consider an arbitrary \(j \in J\).
We have
\[
  \hat{r}
  ~ \stackrel{\eqref{eq:proof:rp_gaifman:distance_between_tuple_and_far_element}}{<} ~
  \dist^{\A}(a_j,b_0)
  ~ \leq ~
  \dist^{\A}(a_j,b_j) + \dist^{\A}(b_j,b_0)
  ~ \stackrel{\eqref{eq:proof:rp_gaifman:distance_between_tuple_entry_and_partner}}{\leq} ~
  r' + \dist^{\A}(b_j,b_0).
\]
Thus,
\(
\dist^{\A}(b_j,b_0)
>
\hat{r}-r'
\).
From \eqref{eq:GNF-choice-of-r}, we know that
\(\hat{r} - r' = \tilde{r} - 2r' = \locRad{q{+}1}{k}{d} - 2r'\),
which, by \eqref{eq:def:radius_function}, is equal to \(4kr' - 2r'\),
and this is \(\geq 4cr' - 2r'\) since \(c \leq \abs{I} \leq k\).

Finally, we have shown that \(b_0, (b_j)_{j \in J}\) serve as witnesses
certifying that \(\A \models \xi^{(c-1)}_{\ell+1}\).
However, this is a contradiction to the fact that
\(\A \models \neg\xi^{(c-1)}_{\ell+1}\).
Therefore, we obtain that \(\A \not\models \eval{\mu}{\ta_I}\).
This, finally, completes Case~(\ref*{def:fparam:existential_plain}).

\proofsubparagraph{Case~(\ref{def:fparam:aggregation}):}
\(\phi\) is of the form
\(\bigl(s = \sum\weight(\ty).\phi_1\bigr)\),
where
\(\weight \in \Weights\),
\(S = \wtype(\weight)\) is finite,
\(s \in S\),
\(\ell = \ar(\weight)\),
\(\ty = (x_{k+1}, \dots, x_{k+\ell})\),
and \(\phi_1 \in \Fparam{q}{k{+}\wmaxar}{d}\)
with \(\free(\phi_1)
\subseteq \free(\phi) \cup \set{x_{k+1}, \dots, x_{k+\ell}}
\subseteq \set{x_1, \dots,x_{k+\wmaxar}}\).

This case can be handled similarly as the analogous case in the
proof of \cite[Theorem~4.6]{vanBergeremSchweikardt_2021_FOWA},
but we rely on \cref{lemma:fv} instead of the Feferman--Vaught decompositions
of \cite{vanBergeremSchweikardt_2021_FOWA}.
The proof details are as follows.

Applying the induction hypothesis to \(\phi_1\),
we obtain an equivalent formula \(\phi_1'\)
with \(\free(\phi_1') = \free(\phi_1)\)
such that \(\phi_1'\) is a Boolean combination of formulas in
\(\sent{q}{k{+}\wmaxar}{d} \cup \loc{q}{k{+}\wmaxar}{d}\),
the formula \(\phi_1'\) has outer quantifier rank
\(\leq \qr(\phi_1) = \qr(\phi)-1\),
and it has inner quantifier rank \(\leq \qr(\phi_1)-1 = \qr(\phi)-2\).

We first consider the case that \(\ell = 0\),
that is,
\(\ar(\weight) = 0\) and \(\free(\phi_1') = \free(\phi)\).
In this case, if \(s = 0_S\),
then \(\phi\) is equivalent to
\(\bigl(\neg\phi_1' \lor \bigl(0_S = \weight()\bigr)\bigr)\);
and if \(s \neq 0_S\), then \(\phi\) is equivalent to
\(\bigl(\phi_1' \land \bigl(s = \weight()\bigr)\bigr)\).
We are done, because both formulas are Boolean combinations of formulas in
\(\loc{q{+}1}{k}{d} \cup \sent{q{+}1}{k}{d}\).

Henceforth, we consider the case that \(\ell \geq 1\).
We transfer \(\phi_1'\) into disjunctive normal form in order to obtain an
equivalent formula of the form
\(\Lor_{\nu = 1}^n \bigl(\chi_\nu \land \lambda_\nu\bigr)\),
where \(n \geq 1\) and,
for all \(\nu \in \intsUpTo{n}\),
\(\chi_\nu\) is a Boolean combination of sentences in
\(\sent{q}{k{+}\wmaxar}{d}\) of inner quantifier rank \(\leq \qr(\phi)-2\),
and \(\lambda_\nu \in \loc{q}{k{+}\wmaxar}{d}\) with
\(\free(\lambda_\nu) \subseteq \free(\phi_1')\)
and \(\qr(\lambda_\nu) \leq \qr(\phi)-1\).
For every \(J \subseteq [n]\), let
\[
  \chi_J \ \deff \ \Land_{j\in J}\chi_j \ \land
  \Land_{j \in [n]\setminus J} \neg\chi_j
  \qquad\text{and}\qquad
  \lambda_J \ \deff \ \Lor_{j\in J} \lambda_j\,.
\]
Clearly, \(\Lor_{\nu = 1}^n \bigl(\chi_\nu \land \lambda_\nu\bigr)\)
is equivalent to
\(\Lor_{\emptyset \neq J \subseteq [n]} \bigl(\chi_J \land \lambda_J\bigr)\),
the \((\chi_J)_{J \subseteq [n]}\) are mutually exclusive
(i.e.\ \(\chi_J \land \chi_{J'}\) is unsatisfiable
for any distinct \(J, J' \subseteq [n]\)),
and they are Boolean combinations of sentences in
\(\sent{q}{k{+}\wmaxar}{d} \subseteq \sent{q{+}1}{k}{d}\).
Furthermore,
for each \(J \subseteq [n]\) with \(J \neq \emptyset\),
we have
\(\lambda_J \in \loc{q}{k{+}\wmaxar}{d}\)
and
\(\free(\lambda_J) \subseteq \free(\phi)
\cup \set{x_{k+1}, \dots, x_{k+\ell}}\).
Let
\[
  \tilde{\phi}
  ~ \deff ~
  \Lor_{\emptyset \neq J \subseteq [n]} \paren[\Big]{
    \chi_J \land \paren[\big]{
      s = \sum \weight(\ty).\lambda_J
    }
  }.
\]
The following is straightforward to prove.

\begin{claim}
  \label{claim:GNF-tilde-phi}
  If \(s \neq 0_S\),
  then \(\phi\) is equivalent to \(\tilde{\phi}\).
  If \(s = 0_S\),
  then \(\phi\) is equivalent to
  \(\bigl(\tilde{\phi} \lor \chi_\emptyset\bigr)\).
\end{claim}
\begin{claimproof}
  Since \(\phi_1 \equiv \phi_1' \equiv \Lor_{\emptyset \neq J \subseteq [n]}
  (\chi_J \land \lambda_J)\),
  the formula \(\phi\) is equivalent to
  \[\hat{\phi} \ \deff \ \paren[\Big]{
    s = \sum \weight(\ty).\paren[\big]{
      \Lor_{\emptyset \neq J \subseteq [n]}
      (\chi_J \land \lambda_J)
    }
  }.\]
  Now let \(\A\) be a \(\paren{\sigma, \Weights}\)-structure,
  let \(\ta \in A^k\),
  and let \(J^* \deff \setc{j \in [n]}{\A \models \chi_j}\).
  Then \(\A \models \chi_{J^*}\),
  and \(\A \not\models \chi_J\) for all \(J \subseteq [n]\) with \(J \neq J^*\).

  If \(J^* \neq \emptyset\),
  then we have that
  \((\A, \ta) \models \hat{\phi}\)
  if and only if
  \((\A, \ta) \models \paren[\big]{
    s = \sum \weight(\ty).\lambda_{J^*}
  }\)
  if and only if
  \((\A, \ta) \models \Lor_{\emptyset \neq J \subseteq [n]}
  \paren[\Big]{
    \chi_J \landp
    \paren[\big]{
      s = \sum \weight(\ty).\lambda_J
    }
  }\).
  Furthermore, it holds that \(\A \not\models \chi_\emptyset\).
  Combined, this shows that
  \((\A, \ta) \models \phi\)
  if and only if
  \((\A, \ta) \models \tilde{\phi}\)
  if and only if
  \((\A, \ta) \models \bigl(\tilde{\phi} \lor \chi_\emptyset\bigr)\).

  If \(J^* = \emptyset\),
  then we have that
  \((\A, \ta) \models \hat{\phi}\)
  if and only if
  \(s = 0_S\),
  and it holds that \((\A, \ta) \not\models \tilde{\phi}\).
  Hence, if \(s = 0_S\),
  then we have
  \((\A, \ta) \models \phi\)
  and
  \((\A, \ta) \models \bigl(\tilde{\phi} \lor \chi_\emptyset\bigr)\).
  On the other hand,
  if \(s \neq 0_S\),
  then
  \((\A, \ta) \not\models \phi\),
  and
  \((\A, \ta) \not\models \tilde{\phi}\).

  All in all, this shows that
  \(\phi \equiv \tilde{\phi}\) if \(s \neq 0_S\),
  and
  \(\phi \equiv \bigl(\tilde{\phi} \lor \chi_\emptyset\bigr)\)
  if \(s = 0_S\).
\end{claimproof}

Based on \cref{claim:GNF-tilde-phi},
in order to complete the proof for Case~(\ref*{def:fparam:aggregation}),
it suffices to consider an arbitrary non-empty \(J \subseteq [n]\)
and the formula
\(\lambda \deff \lambda_J\)
(with \(\lambda \in \loc{q}{k{+}\wmaxar}{d}\),
\(\free(\lambda) \subseteq \free(\phi) \cup \set{x_{k+1}, \dots, x_{k+\ell}}\),
and \(\qr(\lambda)\leq\qr(\phi)-1\))
and show how to transform the formula
\(
  \psi \deff
  \bigl(
     s = \sum \weight(\ty).\lambda
  \bigr)
\)
into an equivalent formula \(\psi'\) that is a Boolean combination
of formulas in \(\loc{q{+}1}{k}{d} \cup \sent{q{+}1}{k}{d}\)
with \(\free(\psi') = \free(\psi)\) and of outer (inner) quantifier rank
at most \(\qr(\phi)\) (resp.~\(\qr(\phi)-1\)).

If \(\free(\psi) = \emptyset\),
then
\(\free(\lambda) \subseteq \set{x_{k+1}, \dots, x_{k+\ell}}
= \set{y_1, \dots, y_\ell}\),
and \(\psi\) is a sentence in
\(\sent{q{+}1}{k}{d}\)
of inner quantifier rank \(\leq \qr(\phi)-1\);
hence, we are done by setting \(\psi' \deff \psi\).

Henceforth, we consider the case that \(\free(\psi) \neq \emptyset\);
in particular, this means that \(k \geq 1\).
Let \(k' \deff k{+}\wmaxar\),
\(\tx \deff (x_1, \dots, x_{k'})\),
\(I \deff \setc{i \in [k]}{x_i \in \free(\psi)}\), and
\(K \deff I \cup \set{x_{k+1}, \dots, x_{k+\ell}}\).
Then, \(\free(\psi) = \setc{x_i}{i \in I}\),
and \(\tx_{K \setminus I} = \ty\).
Let \(T \deff \setc{(t_1,t_2) \in S \times S}{t_1 +_S t_2 = s}\),
which is a finite set, because \(S\) is finite.
Clearly, \(\psi\) is equivalent to
\(\Lor_{(t_1,t_2) \in T} \bigl(\psi'_{t_1} \land \psi''_{t_2}\bigr)\),
where
\begin{eqnarray*}
  \psi'_{t_1}
  & \deff
  & \paren[\Big]{\,
    t_1 =
    \sum \weight(\ty).\paren[\big]{\,
      \lambda \landp \neg\;\dist(\tx_I; \ty) \,{>}\, r'
    \,}
  }
  \qquad\text{and}
  \\
  \psi''_{t_2}
  & \deff
  & \paren[\Big]{\,
    t_2 =
    \sum \weight(\ty).\paren[\big]{\,
      \lambda \landp \dist(\tx_I; \ty) \,{>}\, r'
    \,}
  }.
\end{eqnarray*}
Throughout the remainder of this proof, consider an arbitrary tuple
\((t_1, t_2) \in T\).

\begin{claim}
  \label{claim:psi-prime-in-loc}
  \(\psi'_{t_1}\) is equivalent to a formula in \(\loc{q{+}1}{k}{d}\)
  with the same free variables and of the same quantifier rank.
\end{claim}
\begin{claimproof}
  Note that \(\neg\,\dist(\tx_I; \ty) \,{>}\, r'\)
  is equivalent to the formula
  \(\Lor_{i\in I, j\in[\ell]} \distAtom{x_i}{y_j}{r'}\).
  Therefore, \(\psi'_{t_1}\) is equivalent to the formula
  \[
  \paren[\Big]{\,
    t_1 =
    \sum \weight(\ty).
    \paren[\big]{\,
      \lambda
      \landp\!\!\!
      \Lor_{i\in I,j\in[\ell]}
      \distAtom{x_i}{y_j}{r'}
    \,}
  };
  \]
  and this formula belongs to \(\loc{q{+}1}{k}{d}\)
  and has the same quantifier rank as \(\psi'_{t_1}\).
\end{claimproof}

It remains to transform \(\psi''_{t_2}\)
into an equivalent Boolean combination of formulas in
\(\loc{q{+}1}{k}{d} \cup \sent{q{+}1}{k}{d}\)
and of outer (inner) quantifier rank at most
\(\qr(\phi)\) (resp.~\(\qr(\phi)-1\)).
To achieve this, we apply \cref{lemma:fv} to \(\lambda(\tx_K)\)
and obtain a finite and non-empty set \(\Delta\)
of pairs of formulas \(\paren{\alpha(\tx_I), \beta(\ty)}\)
such that \(\alpha(\tx_I) \in \loc{q}{k{+}\wmaxar}{d}\),
\(\beta(\ty) \in \loc{q}{k{+}\wmaxar}{d}\),
the formulas have quantifier rank \(\leq \qr(\lambda) \leq \qr(\phi)-1\),
the first entries in \(\Delta\) are \emph{mutually exclusive},
and
\begin{equation}
  \label{eq:GNF-aggregation-fv}
  \lambda(\tx_K)
  \landp
  \dist(\tx_I; \ty) \,{>}\, r'
  ~~ \equiv ~
  \Lor_{(\alpha,\beta) \in \Delta}
  \paren[\big]{
    \alpha \paren{\projTup{x}{I}}
    \land
    \beta \paren{\ty}
  }
  \landp
  \dist(\projTup{x}{I}; \ty) \,{>}\, r'.
\end{equation}
Let
\[
  \tilde{\psi}_{t_2}
  ~ \deff ~
  \Lor_{(\alpha,\beta) \in \Delta}\paren[\Big]{
    \alpha(\tx_I)
    \landp
    \paren[\Big]{
      t_2 = \sum \weight(\ty).\paren[\big]{
        \beta(\ty) \landp \dist(\tx_I; \ty) \,{>}\, r'
      }
    }
  }.
\]
Let \(\myAlphas \deff \setc{\alpha}
{\text{there is a \(\beta\) with } (\alpha,\beta) \in \Delta}\)
and \(\myBetas \deff \setc{\beta}
{\text{there is a \(\alpha\) with } (\alpha,\beta) \in \Delta}\).
Using the fact that the first entries in \(\Delta\)
are mutually exclusive,
the following is straightforward to prove.

\begin{claim}
  \label{claim:GNF-psi-t-2}
  If \(t_2 \neq 0_S\),
  then \(\psi''_{t_2}\) is equivalent to
  \(\tilde{\psi}_{t_2}\).
  If \(t_2 = 0_S\),
  then \(\psi''_{t_2}\) is equivalent to
  \(\bigl(\tilde{\psi}_{t_2} \lor
  \Land_{\alpha \in \myAlphas} \neg\alpha\bigr)\).
\end{claim}
\begin{claimproof}
  Let \(\A\) be a \(\paren{\sigma, \Weights}\)-structure,
  and let \(\ta \in A^k\).
  By \cref{eq:GNF-aggregation-fv},
  it holds that
  \((\A, \ta) \models \psi''_{t_2}\)
  if and only if
  \((\A, \ta) \models \paren[\Big]{
    t_2 = \sum \weight(\ty).\paren[\Big]{
      \Lor_{(\alpha, \beta) \in \Delta}
      \bigl(\alpha(\projTup{x}{I}) \land \beta(\ty)\bigr)
      \landp
      \dist(\projTup{x}{I}; \ty) \,{>}\, r'
    }
  }\).
  Since the first entries in \(\Delta\) are mutually exclusive,
  there is at most one \((\alpha, \beta) \in \Delta\)
  with \(\A \models \alpha(\ta_I)\).

  If there is some \((\alpha, \beta) \in \Delta\)
  with \(\A \models \alpha(\ta_I)\),
  then we have
  \((\A, \ta) \models \psi''_{t_2}\)
  if and only if
  \((\A, \ta) \models \Bigl(t_2 = \sum \weight(\ty).
  \bigl(\beta(\ty) \landp \dist(\projTup{x}{I}; \ty) \,{>}\, r'\bigr)\Bigr)\)
  if and only if
  \((\A, \ta) \models \tilde{\psi}_{t_2}\)
  if and only if
  \((\A, \ta) \models \bigl(\tilde{\psi}_{t_2} \lor
  \Land_{\alpha \in \myAlphas} \neg\alpha\bigr)\).

  If there is no \((\alpha, \beta) \in \Delta\)
  with \(\A \models \alpha(\ta_I)\),
  then \(\A \models \eval{\psi''_{t_2}}{\ta}\)
  if and only if \(t_2 = 0_S\),
  and we have \(\A \not\models \eval{\tilde{\psi}_{t_2}}{\ta}\)
  and \((\A, \ta) \models \Land_{\alpha \in \myAlphas} \neg\alpha\).
  Thus, if \(t_2 = 0_S\),
  then \((\A, \ta) \models \psi''_{t_2}\)
  and \((\A, \ta) \models \bigl(\tilde{\psi}_{t_2} \lor
  \Land_{\alpha \in \myAlphas} \neg\alpha\bigr)\).
  On the other hand, if \(t_2 \neq 0_S\),
  then \((\A, \ta) \not\models \psi''_{t_2}\)
  and \((\A, \ta) \not\models \tilde{\psi}_{t_2}\).

  All in all, this shows
  \(\psi''_{t_2} \equiv \tilde{\psi}_{t_2}\) if \(t_2 \neq 0_S\),
  and
  \(\psi''_{t_2} \equiv \bigl(\tilde{\psi}_{t_2} \lor
  \Land_{\alpha \in \myAlphas} \neg\alpha\bigr)\)
  if \(t_2 = 0_S\).
\end{claimproof}

Based on \cref{claim:GNF-psi-t-2},
in order to complete the proof
for Case~(\ref*{def:fparam:aggregation}),
it suffices to consider an arbitrary \(\beta \in \myBetas\)
and transform the formula
\[
  \mu(\tx_I)
  ~ \deff ~
  \Big(
    t_2 =
    \sum\weight(\ty).\big(
      \beta(\ty) \landp \dist(\tx_I; \ty) \,{>}\, r'
    \big)
  \Big)
\]
into an equivalent Boolean combination of formulas in
\(\loc{q{+}1}{k}{d} \cup \sent{q{+}1}{k}{d}\)
and of outer (inner) quantifier rank at most
\(\qr(\phi)\) (resp.~\(\qr(\phi)-1\)).
To achieve this, we let
\(U \deff \setc{(u_1,u_2) \in S \times S}{u_1 -_S u_2 = t_2}\).
This set is finite because \(S\) is finite.
Let
\(\hat{\mu}(\tx_I) \deff \Lor_{(u_1,u_2) \in U}\big(
  \mu'_{u_1} \land \mu''_{u_2}(\tx_I)
\big)\),
where
\begin{eqnarray*}
  \mu'_{u_1}
  & \deff
  & \Big(u_1 = \sum \weight(\ty).\beta(\ty)\Big)
  \qquad\text{and}
  \\
  \mu''_{u_2}(\tx_I)
  & \deff
  & \Big(u_2 = \sum \weight(\ty).\big(
      \beta(\ty) \landp \neg\;\dist(\tx_I; \ty) \,{>}\, r'
    \big)\Big).
\end{eqnarray*}
Here, \(\mu'_{u_1}\) checks that the sum of the weights of
all tuples satisfying \(\beta\) is equal to~\(u_1\),
and \(\mu''_{u_2}\) checks that the sum of the weights of
tuples satisfying \(\beta\) that \emph{are not} far from \(\tx_I\)
is equal to~\(u_2\).
Hence, if \(\mu'_{u_1}\) and \(\mu''_{u_2}\) are satisfied,
then the sum of the weights of tuples satisfying \(\beta\)
that \emph{are} far from \(\tx_I\) is equal to \(u_1 - u_2\).
Thus, using that $(S,+_S)$ is an \emph{abelian group},
and heavily relying on the existence of inverse elements in \(S\),
the following is straightforward to prove.
\begin{claim}
  \(\hat{\mu}(\tx_I)\) is equivalent to \(\mu(\tx_I)\).
\end{claim}
\begin{claimproof}
  Let \(\A\) be a \(\paren{\sigma, \Weights}\)-structure,
  let \(\ta \in A^k\),
  and let \(u', u'' \in S\)
  such that \(\A \models \mu'_{u'}\)
  and \((\A, \ta) \models \mu''_{u''}\)
  (and note that such \(u', u''\) always exist).
We let
  \begin{align*}
    B
    &\ \deff \ \bigsetc{\,\tb \in A^\ell\,}{\,\A \models \beta(\tb)\,},\\
    B_>
    &\ \deff \ \bigsetc{\,\tb \in B\,}{\,\distS{\A}{\ta_I; \tb} \,{>}\,
      r'\,}, \ \text{ and}\\
    B_\leq
    &\ \deff \ B \setminus B'.
  \end{align*}
  It is easy to see that \(u' = \sum_{\tb \in B} \weight^\A(\tb)\),
  \(u'' = \sum_{\tb \in B_\leq} \weight^\A(\tb)\),
  and \(u' = u'' +_S \sum_{\tb \in B_>} \weight^\A(\tb)\).
  Since \(S\) is an abelian group, we obtain that
  \(u' -_S u'' = \sum_{\tb \in B_>} \weight^\A(\tb)\).
  Furthermore, it can easily be verified that
  \((\A, \ta) \models \mu\)
  if and only if
  \(t_2 = \sum_{\tb \in B_>} \weight^\A(\tb)\).

  Hence, combined, this shows that
  \begin{align*}
    (\A, \ta) \models \mu
    & ~\iff~
    t_2 = \sum_{\tb \in B_>} \weight^\A(\tb)
    \\
    & ~\iff~
    t_2 = u' -_S u''
    \\
    & ~\iff~
    (u', u'') \in U
    \\
    & ~\iff~
    (\A, \ta) \models \hat{\mu}.
  \end{align*}
  Since \(\A\) and \(\ta\) were chosen arbitrarily,
  this shows that \(\mu \equiv \hat{\mu}\).
\end{claimproof}
Note that \(\mu'_{u_1}\) is a sentence in \(\sent{q{+}1}{k}{d}\)
of inner quantifier rank \(\leq \qr(\phi)-1\).
Furthermore, with the same reasoning as in the proof of
\cref{claim:psi-prime-in-loc}
(replacing \(\lambda\) with \(\beta\)),
we obtain that \(\mu''_{u_2}(\tx_I)\)
is equivalent to a formula in \(\loc{q{+}1}{k}{d}\)
with free variables \(\setc{x_i}{i \in I}\)
and of the same quantifier rank.
Moreover, \(\qr(\mu''_{u_2}) = \qr(\beta)+1 \leq \qr(\phi)\).
This, finally, completes Case~(\ref*{def:fparam:aggregation})
and finishes the proof of \cref{thm:GNF_ngFOWplus_refined}.
\end{proof}
 
\section{Application:\texorpdfstring{\newline}{} Deciding First-Order Properties
of Nowhere Dense Structures}\label{section:nowhere-dense}

Throughout this section, we fix a signature \(\sigma\).
By $\FO[\sigma]$ we denote the set of all first-order formulas of
signature $\sigma$.
For a signature \(\sigma' \supseteq \sigma\), and a $\sigma$-structure
$\A$, a \emph{\(\sigma'\)-expansion} of \(\A\) is a \(\sigma'\)-structure \(\A'\)
with the same universe as $\A$, and with
\(R^{\A'} = R^{\A}\) for all \(R \in \sigma\).

Throughout this section,
$\locRadFOOperator$ is the radius function defined in \cref{eq:specialRadiusFunction}
in \cref{sec:rpGNFforFOandFOMOD}.

As an application of our new locality theorem (specifically,
\cref{cor:GNF_FOplus} in \cref{sec:rpGNFforFOandFOMOD}),
we give a simplified proof of the following theorem.

\begin{theorem}[Grohe, Kreutzer, and Siebertz~\cite{GroheKreutzerSiebertz_2017_NowhereDense}]
  \label{thm:fo-mc}
  Let \(\CC\) be a nowhere dense class of \(\sigma\)-structures.
  For every \(\FO[\sigma]\) sentence \(\phi\) and every \(\epsilon > 0\),
  there is an algorithm that, given \(\A \in \CC\),
  decides if \(\A \models \phi\) in time
  \(\bigO_{\CC, \phi, \epsilon}(\abs{\A}^{1+\epsilon})\).
\end{theorem}

In this result, the subscripts indicate that the constants in the \(\bigO\)-notation
may depend on \(\CC\), \(\phi\), and \(\epsilon\).
Grohe, Kreutzer, and Siebertz~\cite{GroheKreutzerSiebertz_2017_NowhereDense}
also prove a uniform version of the theorem for effectively nowhere dense classes \(\CC\).
Here, to keep things simple, we focus on the non-uniform version (that is, \cref{thm:fo-mc}).
However, our simplified proof can easily be adapted to the uniform version as well.

While we review the main definitions necessary to understand the simplification
of the main algorithm that we achieve, this section is not self-contained;
we rely on several substantial results from \cite{GroheKreutzerSiebertz_2017_NowhereDense}
that we state here without proof.

For a \(\sigma\)-structure \(\A\),
we let \(\abs{\A} \deff \abs{A}\) be the \emph{order} of \(\A\),
and
\(\norm{\A} \deff \abs{A} + \sum_{R \in \sigma} \abs{R^\A}\)
is the \emph{size} of \(\A\).
Accordingly, for a graph \(G\),
we let \(\abs{G} \deff \abs{V(G)}\)
and \(\norm{G} \deff \abs{V(G)} + \abs{E(G)}\).

Note that the Gaifman graph \(G_\A\) of \(\A\) can be computed from \(\A\)
in time \(\bigO(\norm{\A})\) (assuming that \(\sigma\) is fixed).
This will allow us to always compute the Gaifman graphs of the structures
we are dealing with within the permitted time constraints.

\subsection{The Splitter Game and Nowhere Dense Classes}

To define nowhere dense graph classes,
we directly use a game characterisation of such classes introduced
in~\cite[Section~4]{GroheKreutzerSiebertz_2017_NowhereDense}.
Let \(G\) be a graph and \(\ell, r \in \NN\).
The \emph{\((\ell,r)\)-splitter game} on \(G\) is played by two players called
\emph{Connector} and \emph{Splitter} as follows.
We let \(G^{(0)} \deff G\).
For \(i \geq 1\), in round \(i\) of the game,
Connector selects a vertex \(v^{(i)} \in V(G^{(i-1)})\).
Then Splitter selects a set \(W^{(i)} \subseteq \neighbr{G^{(i-1)}}{v^{(i)}}\)
such that \(\sum_{j=1}^{i} \abs{W^{(j)}} \leq \ell\).
If \(\neighbr{G^{(i-1)}}{v^{(i)}} \setminus W^{(i)} = \emptyset\), then Splitter wins.
Otherwise, the play continues with
\(G^{(i)} \deff G^{(i-1)}\big[\neighbr{G^{(i-1)}}{v^{(i)}} \setminus W^{(i)}\big]\).
If the play never ends,
which is possible because Splitter may always choose \(W^{(i)} = \emptyset\),
then Connector wins.

It is convenient to allow Splitter to choose \(W^{(i)} = \emptyset\),
even though this does not help Splitter to make progress:
if \(W^{(i)} = \emptyset\), then Connector can always choose \(v^{(i+1)} = v^{(i)}\),
and the position remains unchanged.
Note that the sets \(W^{(i)}\) selected by Splitter are mutually disjoint.
Hence, in every play, there are at most \(\ell\) rounds \(i\)
where Splitter chooses a non-empty set \(W^{(i)}\).

A class \(\CC\) of graphs is \emph{nowhere dense}
if for every \(r \in \NN\),
there is an \(\ell(r) \in \NN\)
such that for all graphs \(G \in \CC\),
Splitter has a winning strategy for the \((\ell(r),r)\)-splitter game on \(G\).
Furthermore, a class \(\CC\) of \(\sigma\)-structures is \emph{nowhere dense}
if the class \(\setc{G_\A}{A \in \CC}\) of the Gaifman graphs
of all structures in \(\CC\) is nowhere dense.

It is a well-known fact that if \(\CC\) is a nowhere dense class of graphs,
then for every \(\epsilon > 0\) and all \(G \in \CC\),
it holds that \(\norm{G} \in \bigO(\abs{G}^{1+\epsilon})\)
(cf.~\cite{GroheKreutzerSiebertz_2017_NowhereDense}).
Hence, if \(\CC\) is a nowhere dense class of \(\sigma\)-structures,
then \(\norm{\A} \in \bigO(\abs{A}^{1+\epsilon})\) for all \(\A \in \CC\).
Furthermore, the closure of a nowhere dense class of graphs
under taking subgraphs is nowhere dense as well,
because if Splitter has a winning strategy for the \((\ell,r)\)-splitter game on a graph \(G\),
then this winning strategy induces a winning strategy for the game on any subgraph of \(G\).

We actually need to be able to compute winning strategies for Splitter efficiently.
Moreover, we must be able to do this while moving to subgraphs of the original graph.
We will describe a generic strategy for Splitter
that can be computed efficiently using just shortest path computations.
A \emph{partial play} of the \((\ell, r)\)-splitter game on a graph \(G\)
is a sequence
\(\pi \deff (v^{(1)}, W^{(1)}, v^{(2)}, W^{(2)}, \dots, v^{(k)}, W^{(k)})\),
where \(v^{(i)}\), \(W^{(i)}\) are the choices made by Connector and Splitter
in round \(i\) of the play.
The \emph{length} of \(\pi\) is \(k\).
We let \(G^{(0)} \deff G\) and
\(G^{(i)} \deff G^{(i-1)}\big[\neighbr{G^{(i-1)}}{v^{(i)}} \setminus W^{(i)}\big]\)
for \(i \in [k]\).
Observe that, by the rules of the game, for all \(i \in [k]\), we have
\(v^{(i)} \in V(G^{(i-1)})\),
which implies \(V(G^{(j)}) \neq \emptyset\) for \(j \in [0,k{-}1]\),
and \(W^{(i)} \subseteq  \neighbr{G^{(i-1)}}{v^{(i)}}\).

The following lemma is extracted from the proof of Theorem~4.2 and
Remark~4.3 of \cite{GroheKreutzerSiebertz_2017_NowhereDense}.
For the reader's convenience, we give a proof of the lemma.

\begin{lemma}[\cite{GroheKreutzerSiebertz_2017_NowhereDense}]
  \label{lem:splitter-strategy}
  Let \(\CC\) be a nowhere dense class of graphs.
  For every \(r \in \NN\),
  there is a \(t = t(\CC, r) \in \NNpos\) such that the following holds.
  Let \(G \in \CC\),
  and let
  \[\pi = (v^{(1)}, W^{(1)}, v^{(2)}, W^{(2)}, \dots, v^{(k)}, W^{(k)}),\]
  where \(k \in \NN\) and, for all \(i \in [k]\),
  \(v^{(i)} \in V(G)\) and \(W^{(i)} \subseteq V(G)\).
  Let
  \begin{equation}\label{eq:splitter-strategy}
    G^{(0)} \deff G \supseteq H^{(1)} \supseteq
    G^{(1)} \supseteq \cdots \supseteq H^{(k)} \supseteq G^{(k)}
  \end{equation}
  such that the following conditions are satisfied.
  \begin{romanenumerate}
    \item\label{item:splitter-strategy-i}
      For all \(i \in [k]\),
      it holds that \(v^{(i)} \in V(H^{(i)})\).
    \item\label{item:splitter-strategy-ii}
      \(W^{(1)} = \emptyset\)
      and for all \(i \in [2,k]\),
      it holds that \(W^{(i)} = \neighbr{H^{(i)}}{v^{(i)}} \cap \bigcup_{j=1}^{i-1} V(P^{(i)}_j)\),
      where \(P^{(i)}_j\) is a path of length at most \(r\)
      from \(v^{(j)}\) to \(v^{(i)}\) in the graph \(H^{(j)}\).
    \item\label{item:splitter-strategy-iii}
      For all \(i \in [k]\),
      it holds that
      \(G^{(i)} = H^{(i)}\big[\neighbr{H^{(i)}}{v^{(i)}} \setminus W^{(i)}\big]\).
  \end{romanenumerate}
  Then \(k \leq t\).
\end{lemma}
\begin{proof}
  We use the following characterisation,
  due to \cite{NesetrilOssonaDeMendez_2011_NowhereDense},
  of nowhere dense graph classes in terms of \emph{uniform quasi-wideness}
  (cf.~\cite{Dawar_2010_Quasiwide}).
  \emph{A class \(\CC\) of graphs is nowhere dense
  if and only if there are functions \(s \colon \NN \to \NN\)
  and \(N \colon \NN \times \NN \to \NN\) such that,
  for all \(r, m \in \NN\),
  all \(G \in \CC\),
  and all \(X \subseteq V(G)\) of size \(\abs{X} > N(r,m)\),
  there is a set \(S \subseteq V(G)\) of size \(\abs{S} \leq s(r)\)
  and a set \(Y \subseteq X\) of size \(\abs{Y} \geq m\)
  with \(\dist^{G \setminus S}(y,y') > r\) for all distinct \(y,y' \in Y\).}

  We choose the functions \(s, N\) for our nowhere dense class \(\CC\)
  according to this characterisation
  and let \(t \deff N\bigl(r, 2s(r) + 2\bigr)\).
  Suppose for contradiction that \(k > t\).

  Let \(X \deff \set{v^{(1)}, \dots, v^{(k)}}\).
  Choose \(S \subseteq V(G)\) of size \(\abs{S} \geq s(r)\)
  and a set \(Y \subseteq X\) of size \(\abs{Y} = 2s(r) + 2\)
  with \(\dist^{G \setminus S}(y,y') > r\) for all distinct \(y, y' \in Y\).
  Say, \(Y = \set{v^{(i_1)},\dots,v^{(i_{2s(r)+2})}}\)
  with \(1 \leq i_1 < i_2 < \cdots < i_{2r(s)+2} \leq k\).
  For all \(j \in [s(r) + 1]\),
  let \(P_j \deff P^{(i_{2j})}_{i_{2j-1}}\) be the path of length at most \(r\)
  from \(v^{(i_{2j-1})}\) to \(v^{(i_{2j})}\) in \(H^{(i_{2j-1})}\)
  from condition~\eqref{item:splitter-strategy-ii}.
  Then \(V(P_j) \cap V(G^{(i_{2j})}) = \emptyset\),
  and therefore \(V(P_j) \cap V(H^{(i_{2j+1})}) = \emptyset\)
  (if \(j < s(r) + 1\)).
  This means that the paths \(P_1, \dots, P_{2s(r)+1}\) are mutually vertex disjoint.
  Since, for every \(j\), the path \(P_j\) has length at most \(r\)
  and \(\dist^{G \setminus S}(v^{(i_{2j-1})}, v^{(i_{2j})}) > r\),
  it holds that \(V(P_j) \cap S \neq \emptyset\).
  Thus, \(\abs{S} \geq s(r) + 1\), which is a contradiction.
  This completes the proof of Lemma~\ref{lem:splitter-strategy}.
\end{proof}

Observe that \eqref{eq:splitter-strategy} and
\eqref{item:splitter-strategy-iii} imply that,
for \(i >j \geq 1\), we have
\(H^{(i)} \subseteq H^{(j)}[\neighbr{H^{(j)}}{v^{(j)}}]\),
and, as \(v^{(i)} \in V(H^{(i)})\) by \eqref{item:splitter-strategy-i},
there is a path \(P^{(i)}_j\) of length at most \(r\)
from \(v^{(j)}\) to \(v^{(i)}\) in \(H^{(j)}\).
Thus, the desired paths in \eqref{item:splitter-strategy-ii} always exist.

Furthermore, note that by \eqref{item:splitter-strategy-ii},
it holds that
\begin{equation}\label{eq:splitter2}
  \sum_{i=1}^{k} \bigabs{W^{(i)}} \leq \sum_{i=2}^k (1+r)(i-1)
  \leq (r + 1) \sum_{i=1}^{k-1} i \leq (r+1)k^2.
\end{equation}
and let
\begin{equation}\label{eq:ell-C-r}
  \ell \deff (r+1)t^2.
\end{equation}
Note that if \(H^{(i)} = G^{(i-1)}\) for all \(i \in [k]\),
then \(\pi\) is a partial play of the \((\ell,r)\)-splitter game.
Moreover, condition~\eqref{item:splitter-strategy-ii} describes a winning strategy for Splitter,
because if Splitter chooses the sets \(W^{(i)}\)
according to condition~\eqref{item:splitter-strategy-ii},
then Splitter wins no matter which vertices Connector chooses.
However, the lemma is stronger;
it even gives Splitter a winning strategy for a generalised game
where, in each round, Connector is allowed to restrict the game to a subgraph
\(H^{(i)}\) of \(G^{(i-1)}\).

\begin{remark}
  \label{rem:compute-splitter-strategy}
  Following \cite[Remark~4.7]{GroheKreutzerSiebertz_2017_NowhereDense},
  we can compute this winning strategy as follows:
  for each \(i \in [k]\),
  we compute a breadth-first search tree \(T^{(i)}\) for the graph \(H^{(i)}\)
  rooted in \(v^{(i)}\) in time \(\bigO\bigl(\norm{H^{(i)}}\bigr)\).
  Then, for each \(j \in [1, i{-}1]\),
  we use the tree \(T^{(j)}\) to compute the path \(P^{(i)}_j\)
  in time \(\bigO\bigl(\bigabs{P^{(i)}_j}\bigr)\),
  which allows us to compute the set \(W^{(i)}\) in time \(\bigO(ri)\).
\end{remark}

\subsection{Neighbourhood Covers}

Let \(G\) be a graph and \(r \in \NN\).
An \emph{\(r\)-neighbourhood cover} of \(G\) is a mapping
\(\X \colon V(G) \to 2^{V(G)}\)
such that for all \(v \in V(G)\),
the set \(\X(v) \subseteq V(G)\) is connected in \(G\)
and \(\neighbr{G}{v} \subseteq \X(v)\).
The sets \(\X(v)\) are called the \emph{clusters} of \(\X\).
For \(X \subseteq V(G)\), we write \(X \in \X\) to denote that \(X\) is a cluster of \(\X\),
\ie, \(X = \X(v)\) for some \(v \in V(G)\).

The \emph{radius} of a non-empty set \(X \subseteq V(G)\) is the least \(s \in \NN\)
such that there is a \(c \in X\),
called a \emph{centre} of \(X\),
with \(X \subseteq \neighb{s}{G[X]}{c}\).
The \emph{maximum radius} of an \(r\)-neighbourhood cover \(\X\)
is the maximum of the radii of the clusters \(X \in \X\).

The \emph{degree} of a vertex \(v \in V(G)\) in a neighbourhood cover \(\X\)
is the number of clusters \(X \in \X\) such that \(v \in X\).
The \emph{maximum degree} of \(\X\) is the maximum of the degrees of all \(v \in V(G)\).
Note that
\begin{equation}\label{eq:Delta}
  \sum_{X \in \X} \abs{X} \leq \abs{V(G)} \cdot \Delta,
\end{equation}
where \(\Delta\) denotes the maximum degree of \(\X\).

All these notions naturally generalise from graphs to \(\sigma\)-structures
by considering the Gaifman graph \(G_\A\) of a \(\sigma\)-structure \(\A\).

\begin{lemma}[\cite{GroheKreutzerSiebertz_2017_NowhereDense}]
  \label{lem:nb-covers}
  Let \(\CC\) be a nowhere dense class of graphs.
  For every \(\epsilon > 0\) and \(r \in \NN\),
  there is an algorithm that,
  given a graph \(G \in \CC\) with \(n \deff \abs{V(G)}\) vertices,
  computes in time \(\bigO(n^{1+\epsilon})\)
  an \(r\)-neighbourhood cover \(\X\) of \(G\)
  of radius at most \(2r\) and maximum degree \(\bigO(n^\epsilon)\)
  and a centre function \(c \colon \X \to V(G)\)
  such that \(X \subseteq \neighb{2r}{G[X]}{c(X)}\) for every \(X \in \X\).
\end{lemma}

\subsection{Distance Independent Sets}

Let \(d \in \NNpos\),
and let \(\A\) be a \(\sigma\)-structure.
A \emph{\(d\)-independent set} in \(\A\) is a set \(Y \subseteq A\)
such that \(\dist^\A(y,y') > d\) for all distinct \(y,y' \in Y\).

\begin{lemma}[\cite{GroheKreutzerSiebertz_2017_NowhereDense}]
  \label{lem:dist-is}
  Let \(\CC\) be a nowhere dense class of \(\sigma\)-structures.
  For all \(d, m \in \NN\) and \(\epsilon > 0\),
  there is an algorithm that,
  given a structure \(\A \in \CC\) and a set \(X \subseteq A\),
  decides in time \(\bigO(n^{1+\epsilon})\)
  whether \(\A\) has a \(d\)-independent set \(Y \subseteq X\) of order \(\abs{Y} \geq m\),
  where \(n \deff \abs{A}\).
\end{lemma}

\subsection{The Removal Lemma}
The main algorithm proceeds by following a strategy for Splitter in the splitter game.
This means that we repeatedly remove elements from a structure.
We have to translate formulas from the original structure
to the structure obtained by removing an element,
and this is the content of the \emph{removal lemma}.
The lemma is closely related to \cite[Lemma~8.2]{GroheKreutzerSiebertz_2017_NowhereDense}
and \cite[Lemma~7.5]{GroheSchweikardt_2018_FOC1}
but needs to be adapted to match the definition of \(\FOplusqkd\)
(see \cref{sec:rpGNFforFOplus}).

Recall that by \(\tz_I\)
we denote the projection of a tuple \(\tz = (z_1, \dots, z_k)\)
to the coordinates in \(I \subseteq [k]\),
We extend the notation by letting \(\tz_{\setminus I} \deff \tz_{[k]\setminus I}\).

Let \(r \in \NN\).
For every relation symbol \(R \in \sigma\),
we let \(\tilde{R}_\emptyset \deff R\),
and for \(s \deff \ar(R)\) and for every set \(J \subseteq [s]\) with \(J \neq \emptyset\),
we introduce a fresh \((s - \abs{J})\)-ary relation symbol \(\tilde{R}_J\).
We let \(\removalsigma{}{}\) be the union of \(\sigma\)
and the set of all these new relation symbols,\footnote{In particular, \(R_{[s]}\) is a 0-ary relation symbol in \(\tilde{\sigma}\).}
and we let \(\removalsigma{r}{}\) be the extension of \(\tilde{\sigma}\)
by fresh unary relation symbols \(S_i\) for all \(i \in [r]\).
For every \(\sigma\)-structure \(\A\) of order \(\abs{A} \geq 2\) and every \(c \in A\),
we let \(\A \lbag c\) be the \(\removalsigma{}{}\)-structure
with universe \(A \setminus \set{c}\)
and relations
\[
  \tilde{R}_J^{\A \lbag c} \deff \bigsetc{\ta_{\setminus J}}
  {\ta \in R^\A \text{ and } J=\setc{j \in [s]}{a_j=c}}
\]
for every \(R \in \sigma\) and every \(J \subseteq [\ar(R)]\).
Furthermore, we let \(\A \lbag_r c\) be the \(\removalsigma{r}{}\)-expansion of \(\A \lbag c\)
in which each \(S_i\) is interpreted by the set of all \(b \in A \setminus \set{c}\)
such that \(\dist^\A(c,b) \leq i\).
Note that, for fixed \(\sigma\) and \(r\),
we can compute \(\A \lbag_r c\) from \(\A\) and \(c\) in time \(\bigO(\norm{\A})\).
Also note that the removal operation ‘respects’ the Gaifman graphs in the sense that
\begin{equation}
  \label{eq:removalGaifmanGraph}
  G_{\A \lbag_r c} = G_\A[A \setminus \set{c}].
\end{equation}

\noindent
Recall that
$\locRadFOOperator$ is the radius function defined in \cref{eq:specialRadiusFunction} in
\cref{sec:rpGNFforFOandFOMOD}.

\begin{lemma}[Removal Lemma]
  \label{lem:fo-removal}
  Let \(q, k, d \in \NN\) with \(k\geq 1\),
  let \(\sigma\) be a signature,
  and let \(r \deff \locRadFO{q}{k}{d}\).
  For every formula \(\phi(\tx) \in \FOplusqkd\) of signature \(\sigma\)
  with \(\tx = (x_1, \dots, x_k)\)
  and for every set \(I \subseteq [k]\),
  there is a formula \(\tilde{\phi}_I(\tx_{\setminus I}) \in \FOplusqkd\)
  of signature \(\removalsigma{r}{}\)
  such that for every \(\sigma\)-structure \(\A\) of order \(\abs{A} \geq 2\),
  all \(c \in A\), and all \(\ta = (a_1, \dots, a_k) \in A^k\)
  such that \(I = \setc{i \in [k]}{a_i=c}\),
  we have
  \begin{equation}
    \label{eq:removal}
    \A \models \phi(\ta)
    \quad\iff\quad
    \A \lbag_r c \models \tilde{\phi}_I(\ta_{\setminus I}).
  \end{equation}
  Furthermore, there is an algorithm that computes \(\tilde{\phi}_I\)
  from \(\phi(\tx)\) and \(I\).
\end{lemma}
\begin{proof}
Note that, for all distance atoms $\dist(y,z)\,{\leq}\,d'$ that occur
in a formula in $\FOplusqkd$, it holds that $d'\leq\locRadFO{q}{k}{d}$;
and \(\removalsigma{\locRadFO{q}{k}{d}}{}\) contains unary relations $S_i$ for all $i$
with $1\leq i\leq \locRadFO{q}{k}{d}$.

Let us fix a $d\in \NN$.
We prove the lemma by induction on $q$, and for each fixed $q$, we show
that the assertion of the lemma holds for all $k\in \NNpos$.
\medskip

For the induction base with $q=0$, consider an arbitrary $k\in\NNpos$,
and let $r\deff\locRadFO{0}{k}{d}$.
We have to prove that the assertion of the lemma
holds for all formulas $\phi(\bar
x)\in\FOplusParam{0}{k}{d}$. First, consider the case that $\phi(\bar
x)$ is an atomic formula.

  If \(\phi(\tx)\) is of the form \(\dist(x_i, x_j) \leq d'\),
  then \(i, j \in [k]\) and \(d' \leq d \leq \locRadFO{0}{k}{d} = r\).
  If \(i, j \in I\),
  then we let
  \(\tilde{\phi}_I(\tx_{\setminus I}) \deff \true\).
  If \(i \in I\) and \(j \not\in I\),
  we let
  \(\tilde{\phi}_I(\tx_{\setminus I}) \deff S_{d'}(x_j)\).
  Analogously, if \(i \not\in I\) and \(j \in I\),
  we let
  \(\tilde{\phi}_I(\tx_{\setminus I}) \deff S_{d'}(x_i)\).
  Finally, if \(i, j \not\in I\),
  we let
  \[\tilde{\phi}_I(\tx_{\setminus I}) \ \deff \ \ \dist(x_i, x_j) \leq d'
  \ \lor \Lor_{\substack{1 \leq d_1, d_2 \leq d'-1,\\ d_1 + d_2 = d'}}
  \bigl(S_{d_1}(x_i) \land S_{d_2}(x_j)\bigr).\]

  If \(\phi(\tx)\) is of the form \(\true\) or \(\false\),
  then we also let \(\tilde{\phi}_I(\tx_{\setminus I}) \deff \true\)
  or \(\tilde{\phi}_I(\tx_{\setminus I}) \deff \false\), respectively.

  If \(\phi(\tx)\) is of the form \(x_i {=} x_j\),
  then \(i,j \in [k]\).
  If \(i,j \in I\), then we let \(\tilde{\phi}_I(\tx_{\setminus I}) \deff \true\).
  If \(i \in I\) and \(j \not\in I\) or if \(i \not\in I\) and \(j \in I\),
  then we let \(\tilde{\phi}_I(\tx_{\setminus I}) \deff \false\).
  Finally, if \(i,j \not\in I\),
  then we let \(\tilde{\phi}_I(\tx_{\setminus I}) \deff x_i {=} x_j\).

  If \(\phi(\tx)\) is of the form \(R(x_{i_1}, \dots, x_{i_s})\)
  for some \(R \in \sigma\) with \(s \deff \ar(R)\) and \(i_1, \dots, i_s \in [k]\),
  then we let \(\tilde{\phi}_I(\tx_{\setminus I}) \deff \tilde{R}_J(\ty)\),
  where \(J \deff \setc{j \in [s]}{i_j \in I}\)
  and \(\ty\) is the tuple obtained from \((x_{i_1}, \dots, x_{i_s})\)
  by dropping all entries \(x_{i_j}\) with \(j \in J\).

  This concludes the construction for atomic formulas.
  It can easily be verified that \(\tilde{\phi}_I(\tx_{\setminus I}) \in \FOplusParam{0}{k}{d}\)
  and that \cref{eq:removal} holds for all \(\sigma\)-structures \(\A\), \(c \in A\),
  and tuples \(\ta \in A^k\) as specified.

  Boolean combinations of such formulas are handled in the obvious
  way:
  if \(\phi = \neg\psi\), then \(\tilde{\phi}_I \deff \neg \tilde{\psi}_I\),
  and if \(\phi = \psi \lor \chi\),
  then \(\tilde{\phi}_I \deff \tilde{\psi}_I \lor \tilde{\chi}_I\).
  This completes the induction base with $q=0$.
\medskip
  
For the induction step from $q$ to $q{+}1$, consider arbitrary $q,k\in
\NN$ with $k\geq 1$.
We have to prove that the assertion of the lemma holds for all formulas in
$\FOplusParam{q{+}1}{k}{d}$ and $r\deff\locRadFO{q{+}1}{k}{d}$.
From the induction hypothesis, we obtain that the assertion of the lemma
holds for all formulas in $\FOplusParam{q}{k{+}1}{d}$ and $r'\deff\locRadFO{q}{k{+}1}{d}$.

By the definition of $\FOplusParam{q{+}1}{k}{d}$, it consists of
formulas that are Boolean combinations of formulas of the types
\eqref{item:FOplusParam_One}--\eqref{item:FOplusParam_Four} (cf.,
\cref{sec:rpGNFforFOplus}).
Consider a $\phi(\bar x)\in\FOplusParam{q{+}1}{k}{d}$. 

If $\phi(\bar x)$ is of type \eqref{item:FOplusParam_One}, it belongs to
$\FOplusParam{q}{k{+}1}{d}$. We are done by using the induction
hypothesis and the fact that $r\geq r'$.

If $\phi(\bar x)$ is of type \eqref{item:FOplusParam_Two}, it is a
distance atom and we can
handle it in the same way as we handled distance atoms in the base
case for $q=0$.

If $\phi(\bar x)$ is of type \eqref{item:FOplusParam_Three}, it is of
the form
  \(\exists x_{k+1}\, \bigl(\dist(x_i, x_{k+1}) \leq \hat{d} \land \psi(\tx')\bigr)\)
  for \(i \in [k]\), \(\psi(\tx') \in \FOplusParam{q}{k{+}1}{d}\),
  \(\tx' = (x_1, \dots, x_{k+1})\),
  and \(\hat{d} \leq r-r' \leq r\).
  Let \(I' \deff I \cup \set{k{+}1}\).
  If \(i \in I\), then we set
  \[\tilde{\phi}_I(\tx_{\setminus I}) \ \deff \ \  \tilde{\psi}_{I'}(\tx'_{\setminus I'})
  \ \lor\ \exists x_{k+1}\, \bigl(S_{\hat{d}}(x_{k+1}) \land \tilde{\psi}_I(\tx'_{\setminus I})\bigr).\]
  If \(i \not\in I\), we set
  \begin{align*}
    \tilde{\phi}_I(\tx_{\setminus I}) \ \deff \ \
    &\hphantom{\lor\ }\bigl(\,S_{\hat{d}}(x_i) \land \tilde{\psi}_{I'}(\tx'_{\setminus I'})\,\bigr)\\
    &\lor \exists x_{k+1}\,\bigl(\dist(x_i, x_{k+1}) \leq \hat{d}
    \ \land \ \tilde{\psi}_I(\tx'_{\setminus I})\bigr)\\
    &\lor \Lor_{\substack{1 \leq d_1, d_2 \leq \hat{d}-1\\ d_1 + d_2 = \hat{d}}}
    \exists x_{k+1}\,\bigl(S_{d_1}(x_i) \land S_{d_2}(x_{k+1})
    \land \tilde{\psi}_I(\tx'_{\setminus I})\bigr).
  \end{align*}

  If \(\phi(\tx)\) is of type \eqref{item:FOplusParam_Four}, it is of the form \(\exists x_{k+1} \psi(\tx')\)
  for \(\psi(\tx') \in \FOplusParam{q}{k{+}1}{d}\) and \(\tx' = (x_1, \dots, x_{k+1})\).
  Let \(I' \deff I \cup \set{k{+}1}\).
  We set
  \(\tilde{\phi}_I(\tx_{\setminus I}) \deff \tilde{\psi}_{I'}(\tx'_{\setminus I'})
  \lor \exists x_{k+1}\,\tilde{\psi}_I(\tx'_{\setminus I})\).

   It can easily be verified that, in all four cases, the constructed
   formula $\tilde{\phi}_I(\tx_{\setminus I})$ belongs to $\FOplusParam{q{+}1}{k}{d}$,
   and that \cref{eq:removal} is satisfied.

  Finally, Boolean combinations of formulas of the types
\eqref{item:FOplusParam_One}--\eqref{item:FOplusParam_Four} are
handled in the obvious way:
  if \(\phi = \neg\psi\), then \(\tilde{\phi}_I \deff \neg \tilde{\psi}_I\),
  and if \(\phi = \psi \lor \chi\),
  then \(\tilde{\phi}_I \deff \tilde{\psi}_I \lor \tilde{\chi}_I\).
 This completes the induction step, and it completes the proof of \cref{lem:fo-removal}.
\end{proof}

We will only apply the Removal Lemma to formulas with one free variable,
and we need an iterated version of the Removal Lemma for this case.
Let \(r \in \NN\).
We let \(\removalsigma{r}{0} \deff \sigma\), and, for \(s \in \NN\),
we let \(\removalsigma{r}{s+1} \deff \removalsig{\tau}{r}{}\)
for \(\tau \deff \sigmars\).
For a \(\sigma\)-structure \(\A\) of order \(\abs{A} > s\)
and distinct elements \(c_1, \dots, c_s \in A\), we let
\[
  \A\lbag_r c_1 \cdots c_s \deff ( \cdots ((\A \lbag_r c_1) \lbag_r c_2) \cdots )\lbag_r c_s.
\]
Then \(\A \lbag_r c_1 \cdots c_s\) is a \(\sigmars\)-structure
with universe \(A \setminus \set{c_1, \dots, c_s}\)
and Gaifman graph \(G_{\A}[A \setminus \set{c_1, \dots, c_s}]\)
(this follows inductively by using \cref{eq:removalGaifmanGraph}).
Furthermore, we have
\(
  \sigma = \removalsigma{r}{0}
  \subseteq \removalsigma{r}{1}
  \subseteq \removalsigma{r}{2}
  \subseteq \cdots
\).

\begin{corollary}
  \label{cor:iterated-removal}
  Let \(q, d \in \NN\), let \(k,s \in \NNpos\),
  let \(\sigma\) be a signature,
  and let \(r \deff \locRadFO{q}{k}{d}\).
  For every formula \(\phi(x) \in \FOplusqkd\) of signature \(\sigma\),
  there are sentences \(\tilde{\phi}_1, \dots, \tilde{\phi}_s \in \FOplusqkd\)
  and a formula
  \(\tilde{\phi}_0(x) \in \FOplusqkd\) of signature \(\sigmars\)
  such that for every \(\sigma\)-structure \(\A\) of order \(\abs{A} \geq s+1\),
  all distinct \(c_1, \dots, c_s \in A\), and all \(a \in A\),
  we have
  \begin{equation}
    \label{eq:removal2}
    \A \models \phi(a) \iff
    \begin{cases}
      \A \lbag_r c_1 \cdots c_s \models \tilde{\phi}_0(a)
      &\text{if } a \in A \setminus \set{c_1, \dots, c_s},\\
      \A \lbag_r c_1 \cdots c_s \models \tilde{\phi}_1
      &\text{if } a = c_1,\\
      \mkern64mu \vdots
      &\qquad \vdots\\
      \A \lbag_r c_1 \cdots c_s \models \tilde{\phi}_s
      &\text{if } a = c_s.
    \end{cases}
  \end{equation}
  Furthermore, there is an algorithm that computes
  \(\tilde{\phi}_0(x), \tilde{\phi}_1, \dots, \tilde{\phi}_s\)
  from \(\phi(x)\) and \(s\).
\end{corollary}

\subsection{The Main Algorithm}
\Cref{thm:fo-mc} immediately follows from the following lemma.
For a formula \(\phi(x)\) and a structure \(\A\),
it will be convenient to set
\(\phi(\A) \deff \setc{a \in A}{\A \models \phi(a)}\).

\begin{lemma}
  \label{lem:fo-mc}
  Let \(\CC\) be a nowhere dense class of \(\sigma\)-structures.
  Furthermore, let \(\phi(x)\) be an \(\FOplus[\sigma]\) formula.
  For every \(\epsilon > 0\), there is an algorithm that,
  given \(\A \in \CC\), computes the set \(\phi(\A)\)
  in time \(\bigO_{\CC, \phi, \epsilon}(\abs{A}^{1+\epsilon})\).
\end{lemma}
\begin{proof}
  Let \(\CC_G\) be the class of all Gaifman graphs of structures in \(\CC\)
  and all their subgraphs.
  Then \(\CC_G\) is nowhere dense.

  We assume without loss of generality that \(\free(\phi) = \set{x}\);
  in case that \(\free(\phi) = \emptyset\),
  we replace \(\phi\) with the formula \((\phi \land x{=}x)\).
  Let \(q, k, d \in \NN\) such that \(\phi \in \FOplusqkd\),
  and let \(r \deff \locRadFO{q}{k}{d}\), where $\locRadFOOperator$ is
  the function defined in \cref{eq:specialRadiusFunction} in
  \cref{sec:rpGNFforFOandFOMOD}.
  We choose \(t \deff t(\CC_G,2r)\) according to \cref{lem:splitter-strategy}
  and let \(\ell \deff (2r+1)t^2\) as in \cref{eq:ell-C-r}
  (for \(2r\) instead of \(r\)).
Let \(\epsilon > 0\).
  Without loss of generality,
  we assume \(\epsilon \leq \frac{1}{2}\)
  (which implies that \(\epsilon^2 \leq \epsilon/2\)).
  We let
  \begin{equation}
    \label{eq:eps}
    \epsilon' \deff \frac{\epsilon}{2t}.
  \end{equation}
  Furthermore, we choose a \(c_1 \in \NNpos\)
  such that for all graphs \(G \in \CC_G\) of order \(n_G \deff |G| > c_1\),
  the \(r\)-neighbourhood cover of \cref{lem:nb-covers},
  applied to \(G\) and \(\epsilon'/2\),
  has maximum degree at most \(n_G^{\epsilon'}\).
  Moreover, we let \(c_2 \deff c_1 + \ell\).

  Let \(\A\) be the input structure for our algorithm, and set \(n \deff \abs{A}\).
  If \(n \leq c_1\), then we compute the set \(\phi(\A)\) by brute force.
  Hence, in the following, suppose that \(n > c_1\),
  and let \(\X\) be an \(r\)-neighbourhood cover of \(G_\A\) of radius at most \(2r\)
  and with maximum degree at most \(n^{\epsilon'}\)
  computed via \cref{lem:nb-covers}.

  By \cref{cor:GNF_FOplus},
  the \(\FOplusqkd\) formula \(\phi(x)\) is equivalent to a Boolean combination \(\phi'(x)\)
  of \(r\)-local \(\FOplusqkd\) formulas \(\psi(x)\) and basic local sentences of the form
  \begin{equation}
    \label{eq:fo-mc-basic-local}
    \xi \deff \exists y_1\cdots\exists y_\ell\;
    \big(
      \Land_{1\leq j < j'\leq \ell} \dist(y_j,y_{j'})\,{>}\,2r'
      \ \land\ \Land_{j\in[\ell]} \lambda(y_j)
    \big),
  \end{equation}
  where \(1 \leq \ell \leq k+q\),
  \(y_1, \dots, y_\ell\) are \(\ell\) variables,
  \(r', \tilde{r} \in \NN\) with \(\tilde{r} \leq r' \leq \frac{1}{2}r + \locRadFO{q-1}{k+1}{d}\),
  \(\tilde{r} \leq \locRadFO{q-1}{k+1}{d} \leq r\),
  and \(\lambda(x_1) \in \FOplusParam{q-1}{k+1}{d}\) is \(\tilde{r}\)-local.
  In particular, \(\lambda(x_1)\) is \(r\)-local.

  For every \(r\)-local formula \(\psi(x)\) that appears in \(\phi'(x)\),
  either directly in the Boolean combination or as part of a basic local sentence,
  we shall compute the set \(\psi(\A)\).
  Then we can evaluate the basic local sentences \(\xi\)
  of the form~\eqref{eq:fo-mc-basic-local} as follows:
  we note that \(\A \models \xi\) if and only if \(\A\) has a \(2r'\)-independent set
  \(Y \subseteq \psi(\A)\) of order \(\abs{Y} \geq \ell\).
  We can use \cref{lem:dist-is} to decide if this is the case in time \(\bigO(n^{1+\epsilon'})\).
  Once we have computed the truth values of all basic local sentences \(\xi\) in \(\phi'\)
  as well as the sets \(\psi(\A)\) for all local formulas \(\psi(x)\) in \(\phi'\),
  we can easily evaluate the Boolean combination \(\phi'\)
  and compute the set \(\phi'(\A) = \phi(\A)\).
  Thus, it remains to compute \(\psi(\A)\) for all formulas \(\psi(x)\)
  in a finite set \(\Psi\) of \(r\)-local \(\FOplusqkd\) formulas.

  We shall recursively compute the set \(\psi(\A[X])\) for every \(X \in \X\)
  and every \(\psi(x) \in \Psi\).
  Observe that, if \(\nrA{a} \subseteq X\),
  then \(a \in \psi(\A) \iff a\in \psi(\A[X])\),
  because \(\psi(x)\) is \(r\)-local.
  Furthermore, \(\X\) is an \(r\)-neighbourhood cover,
  so this implies
  \begin{equation}
    \label{eq:evaluate-root}
    \psi(\A) = \bigsetc{a \in A}{a \in \psi\bigl(\A[\X(a)]\bigr)}.
  \end{equation}

  Let us move on to the recursive step.
  It will be convenient to think of our recursive algorithm
  in terms of the recursion tree it builds.
  Each node of the recursion tree will be indexed by a tuple \(\tX = (X_1, \dots, X_k)\),
  where \(k \in [0,t]\),
  and the \(X_i\) are subsets of \(A\).
  The root is indexed by the empty tuple \(()\).
  In the following, we will no longer distinguish between the nodes and their indices
  and refer to a tuple \(\tX\) as a node of the tree.

  For each node \(\tX\) that is not a leaf of the tree,
  we define a signature \(\sigma_{\tX}\)
  that will always be of the form \(\removalsigma{r}{i}\)
  for some \(i \geq 0\),
  a \(\sigma_{\tX}\)-structure \(\A_{\tX}\),
  its Gaifman graph \(G_{\tX}\),
  and \(r\)-neighbourhood cover \(\X_{\tX}\) of \(G_{\tX}\)
  of radius at most \(2r\) and maximum degree at most \(\abs{G_{\tX}}^{\epsilon'}\),
  a centre map \(c_{\tX}\) for \(\X_{\tX}\),
  and a finite set \(\Psi_{\tX}\) of \(r\)-local \(\FOplusqkd\) formulas
  with one free variable and of signature \(\sigma_{\tX}\).
  For the root \(()\),
  we let \(\sigma_{()} \deff \sigma = \removalsigma{r}{0}\),
  \(\A_{()} \deff \A\),
  \(G_{()} \deff G_\A\),
  \(\X_{()} \deff \X\),
  and \(\Psi_{()} \deff \Psi\).

  For nodes \(\tX \neq ()\) that are not leaves,
  in addition to the structure \(\A_{\tX}\), the graph \(G_{\tX}\),
  the neighbourhood cover \(\X_{\tX}\),
  and the set \(\Psi_{\tX}\) of formulas,
  we will define a graph \(H_{\tX}\),
  a vertex \(v_{\tX}\),
  and a set \(W_{\tX}\).
  They will be defined in such a way that,
  for each node \(\tX = (X_1, \dots, X_k)\),
  the vertices \(v^{(i)} \deff v_{(X_1, \dots, X_i)}\),
  sets \(W^{(i)} \deff W_{(X_1, \dots, X_i)}\),
  graphs \(G^{(0)} \deff G_{()}\)
  and \(G^{(i)} \deff G_{(X_1, \dots, X_i)}\),
  \(H^{(i)} \deff H_{(X_1, \dots, X_i)}\) for \(i \in [k]\) satisfy
  the conditions \eqref{eq:splitter-strategy}
  and \eqref{item:splitter-strategy-i}--\eqref{item:splitter-strategy-iii}
  of \cref{lem:splitter-strategy} with \(2r\) instead of \(r\).
  Furthermore, we define a breadth-first search tree \(T_{\tX}\)
  of the graph \(H_{\tX}\) rooted in \(v_{\tX}\).

  For a node \(\tX = (X_1, \dots, X_k)\),
  if \(k \geq 1\) and \(\abs{X_k} \leq c_2\), then \(\tX\) is a leaf of the tree.
  If \(k = 0\) or \(\abs{X_k} > c_2\),
  then the node \(\tX\) has a child \(\tX Y \deff (X_1, \dots, X_k, Y)\)
  for every \(Y \in \X_{\tX}\).

  Let \(\tX = (X_1, \dots, X_k)\) and \(\tY = (X_1, \dots, X_k, Y)\)
  for some \(Y \in \X_{\tX}\).
  Our goal at the node \(\tX\) is to compute \(\psi(\A_{\tX}[Y])\)
  for every \(\psi(x) \in \Psi_{\tX}\).
  From these sets \(\psi(\A_{\tX}[Y])\), we can compute
  \begin{equation}
    \label{eq:eval-from-children}
  \psi(\A_{\tX}) = \bigsetc{a \in A_{\tX}}{a \in \psi\bigl(\A_{\tX}[\X_{\tX}(a)]\big)},
  \end{equation}
  exploiting that \(\psi(x)\) is \(r\)-local
  and that \(\X_{\tX}\) is an \(r\)-neighbourhood cover
  of the Gaifman graph \(G_{\tX}\) of \(\A_{\tX}\).
  For the root \(\tX = ()\) of the recursion tree,
  this enables us to compute \(\phi(\A)\), as argued above.

  Consider a node \(\tY = \tX Y = (X_1, \dots, X_k, Y)\) for some \(k \in \NN\).
  If \(\abs{Y} \leq c_2\), then \(\tY\) is a leaf,
  and we compute \(\psi(\A_{\tX}[Y])\) for every \(\psi(x) \in \Psi_{\tX}\) by brute force.
  Now suppose that \(\abs{Y} > c_2\).
  For \(i \in [0,k]\), let \(\tX^{(i)} \deff (X_1, \dots, X_i)\)
  and \(G^{(i)} \deff G_{\tX^{(i)}}\).
  Furthermore, if \(i \geq 1\),
  let \(v^{(i)} \deff v_{\tX^{(i)}}\),
  \(W^{(i)} \deff W_{\tX^{(i)}}\),
  and \(H^{(i)} \deff H_{\tX^{(i)}}\).
  Observe that \(Y\) is a cluster in the neighbourhood cover \(\X_{\tX^{(k)}}\) of \(G^{(k)}\),
  and hence, \(Y \subseteq V(G^{(k)})\).
  We let
  \(v_{\tY} \deff v^{(k+1)} \deff c_{\tX^{(k)}}(Y)\)
  and \(H_{\tY} \deff H^{(k+1)} \deff G^{(k)}[Y]\).
  We choose \(W_{\tY} \deff W^{(k+1)} \subseteq Y\)
  according to \cref{lem:splitter-strategy}\eqref{item:splitter-strategy-ii}
  with \(2r\) instead of \(r\),
  and we let
  \(G_{\tY} \deff G^{(k+1)} \deff H^{(k+1)}
  \bigl[\neighb{2r}{H^{(k+1)}}{v^{(k+1)}} \setminus W^{(k+1)}\bigr]\)
  (according to \cref{lem:splitter-strategy}\eqref{item:splitter-strategy-iii}).
  Note that, actually, we have
  \(G^{(k+1)} = H^{(k+1)}[Y \setminus W^{(k+1)}]\),
  because the radius of the neighbourhood cover \(\X_{\tX^{(k)}}\) is at most \(2r\)
  and \(v^{(k+1)} = c_{\tX^{(k)}}(Y)\),
  so \(Y \subseteq \neighb{2r}{G^{(k)}[Y]}{v^{(k+1)}}\).

  If \(W^{(k+1)} = \emptyset\), then we proceed as follows.
  We let \(\sigma_{\tY} \deff \sigma_{\tX}\)
  and \(\A_{\tY} \deff \A_{\tX}[Y]\).
  We let \(G_{\tY} \deff G_{\A_{\tY}}\)
  and compute an \(r\)-neighbourhood cover \(\X_{\tY}\) of \(G_{\tY}\)
  of radius at most \(2r\) and maximum degree at most \(\abs{Y}^{\epsilon'}\)
  together with a centre map \(c_{\tY}\) using \cref{lem:nb-covers}.
  Moreover, we let \(\Psi_{\tY} \deff \Psi_{\tX}\).
  The children of \(\tY\) are \(\tY Z\) for \(Z \in \X_{\tY}\),
  and, for all \(\psi(x) \in \Psi_{\tY} = \Psi_{\tX}\),
  we can compute \(\psi(\A_{\tY}) = \psi(\A_{\tX}[Y])\)
  from the recursively computed sets \(\psi(\A_{\tY}[Z])\)
  for \(Z \in \X_{\tY}\) as in \cref{eq:eval-from-children}.

  Now suppose that \(W^{(k+1)} = \set{w_1, \dots, w_s} \neq \emptyset\)
  for pairwise distinct \(w_1, \dots, w_s\).
  In order to define \(\sigma_{\tY}\), suppose that \(\sigma_{\tX} = \removalsigma{r}{i}\).
  Then we let \(\sigma_{\tY} \deff \removalsigma{r}{i+s}\)
  and \(\A_{\tY} \deff \A_{\tX}[Y] \lbag_r w_1 \cdots w_s\).
  Set \(G_{\tY} \deff G_{\A_{\tY}}\),
  and note that \(V(G_{\tY}) = Y \setminus W^{(k+1)}\).
  Moreover, since \(\abs{Y} > c_2 = c_1 + \ell\) and \(s = \abs{W^{(k+1)}} \leq \ell\),
  we have \(\abs{G_{\tY}} > c_1\).
  We compute an \(r\)-neighbourhood cover \(\X_{\tY}\) of \(G_{\tY}\)
  of radius at most \(2r\) and maximum degree at most
  \(\abs{Y \setminus W^{(k+1)}}^{\epsilon'}\)
  together with a centre map \(c_{\tY}\) using \cref{lem:nb-covers}.
  The children of \(\tY\) are \(\tY Z\) for \(Z \in \X_{\tY}\).
For each \(\psi(x) \in \Psi_{\tX} \subseteq \FOplusqkd\),
  we compute sentences \(\tilde{\psi}_1, \dots, \tilde{\psi}_s \in \FOplusqkd\)
  and a formula \(\tilde{\psi}_0(x) \in \FOplusqkd\) by \cref{cor:iterated-removal}.
  The corollary shows that we can compute \(\psi(\A_{\tX}[Y])\)
  from the truth values of the sentences \(\tilde{\psi}_i\) in \(\A_{\tY}\)
  and the set \(\psi_0(\A_{\tY})\).
  We apply \cref{cor:GNF_FOplus} to each of the \(\tilde{\psi}_i\)
  and obtain a formula \(\tilde{\psi}'_i\) in Gaifman normal form.
  As explained above for \(\phi(x)\),
  from this formula in Gaifman normal form,
  we obtain a set \(\Psi_i\) of \(r\)-local formulas with one free variable in \(\FOplusqkd\)
  such that we can reduce the evaluation of \(\tilde{\psi}_i\) in \(\A_{\tY}\)
  to the evaluation of all formulas in \(\Psi_i\) in \(\A_{\tY}\).
  We let \(\Psi_\psi \deff \bigcup_{i=0}^s \Psi_i\).
  Then we can reduce the evaluation of \(\psi\) in \(\A_{\tX}[Y]\)
  to the evaluation of all formulas in \(\Psi_\psi\) in \(\A_{\tY}\).
  Finally, we let \(\Psi_{\tY} \deff \bigcup_{\psi \in \Psi_{\tX}} \Psi_\psi\).
  Then we can compute \(\psi(\A_{\tX}[Y])\) for all \(\psi \in \Psi_{\tX}\)
  from the recursively computed \(\chi(\A_{\tY}[Z])\) for \(Z \in \X_{\tY}\)
  and \(\chi(x) \in \Psi_{\tY}\).

  This completes the description of the recursive algorithm.
  It remains to analyse the algorithm.
  We first observe that the depth of the recursion is at most \(t\);
  this follows from \cref{lem:splitter-strategy}.
  Moreover, for each node \(\tX = (X_1, \dots, X_k)\),
  the sizes of the sets \(W_{\tX_i}\) for \(\tX_i \deff (X_1, \dots, X_i)\)
  add up to at most \(\ell\) by \cref{eq:splitter2}
  (for \(2r\) instead of \(r\)) and our choice of \(\ell\).
  Whenever we remove the elements of a non-empty set \(W_{\tX_i}\) of size \(s\),
  we extend the signature from \(\sigma_{\tX_{i-1}} = \removalsigma{r}{j}\)
  to \(\sigma_{\tX_{i}} = \removalsigma{r}{j+s}\).
  Thus, \(\sigma_{\tX} = \removalsigma{r}{j'}\) for \(j' = \sum_{i=1}^k \abs{W_{\tX_i}}\).
  As \(j' \leq \ell\),
  it follows that \(\sigma_{\tX} \subseteq \removalsigma{r}{\ell} \ffed \sigma^*\).
  Note that the size of \(\sigma^*\) only depends on the input formula \(\phi(x)\)
  (via \(\sigma\) and \(r\))
  and the class \(\CC\)
  (via \(t = t(\CC_G, 2r)\)),
  but not on the size of the structure $\A$.

  Thus, all formulas \(\psi(x) \in \Psi_{\tX}\) are of a signature
  that is contained in \(\sigma^*\).
  Furthermore, by the construction of \(\Psi_{\tX}\),
  all formulas in \(\Psi_{\tX}\) are contained in a finite set of \(\FOplusqkd\) formulas
  that only depends on \(\phi(x)\) and \(\CC\),
  but not on the size of the structure $\A$.
  Thus, the number and size of the formulas we need to evaluate at each node
  only contributes a constant factor to the running time.

  All structures \(\A_{\tX}\) have signature contained in \(\sigma^*\),
  which means that we can bound their size \(\norm{\A_{\tX}}\)
  and the size of their Gaifman graph \(G_{\tX}\)
  by \(\bigO\bigl(\abs{\A_{\tX}}^{1+\epsilon'}\bigr)\).
  Also note that for \(\tX = (X_1, \dots, X_k)\),
  we have \(\abs{\A_{\tX}} \leq \abs{X_k} \ffed n_{\tX}\).
  At each non-leaf node,
  we have to compute the neighbourhood cover \(\X_{\tX}\) and centre map \(c_{\tX}\),
  which can be done in time \(\bigO\bigl(n_{\tX}^{1+\epsilon'}\bigr)\) by \cref{lem:nb-covers}.
  Furthermore, we compute the breadth-first search tree \(T_{\tX}\) of \(H_{\tX}\),
  which can also be done in time
  \(\bigO(\norm{H_{\tX}}) \subseteq \bigO\bigl(n_{\tX}^{1+\epsilon'}\bigr)\).
  By using the trees \(T_{\tX_i}\) of the ancestors \(\tX_i\) of \(\tX\) in the tree,
  we can compute \(W_{\tX}\) in constant time (cf.~\cref{rem:compute-splitter-strategy}).

  Based on this,
  we can compute \(\A_{\tX}\) in time
  \(\bigO(\norm{\A_{\tX}}) \subseteq \bigO\bigl(n_{\tX}^{1+\epsilon'}\bigr)\).
  Finally, assuming \(\tX = \tX' X_k\),
  we have to compute \(\psi(\A_{\tX'}[X_k])\) for each \(\psi(x) \in \Psi_{\tX'}\)
  from the recursively computed \(\chi(\A_{\tX}[Y])\) for \(Y \in \X_{\tX}\)
  and \(\chi(x) \in \Psi_{\tX}\).
  Combining the sets \(\chi(\A_{\tX}[Y])\) to \(\chi(\A_{\tX})\) requires time
  \(\bigO\bigl(\sum_{Y \in \X_{\tX}} \abs{Y}\bigr)
  \subseteq \bigO\bigl(n_{\tX}^{1+\epsilon'}\bigr)\) by \cref{eq:Delta},
  since the maximum degree of \(\X_{\tX}\) is at most \(n_{\tX}^{\epsilon'}\).
  Computing \(\psi(\A_{\tX'}[X_k])\) for all \(\psi(x) \in \Psi_{\tX'}\)
  from the sets \(\chi(\A_{\tX})\) for \(\chi(x) \in \Psi_{\tX}\)
  involves a large constant factor to handle all formulas
  and may involve computing \(2r'\)-independent sets
  using \cref{lem:dist-is} to evaluate basic local sentences,
  but it is possible in time \(\bigO\bigl(n_{\tX}^{1+\epsilon'}\bigr)\).

  Overall, the running time at each non-leaf node \(\tX\) is in
  \(\bigO\bigl(n_{\tX}^{1+\epsilon'}\bigr)\)
  (not including the time for the recursive calls).
  The running time at a leaf node is constant.

  We obtain the following recurrence for the running time \(T(j,m)\)
  the algorithm spends in the subtree rooted at a node
  \(\tX = (X_1, \dots, X_{t-j})\) with \(n_{\tX} = m\):
  \begin{align*}
    T(0,m)
    &\in \bigO(1)
    &&\text{for all } m,\\
    T(j,m)
    &\in \bigO(1)
    &&\text{for all } j \in [t] \text{ and all } m \leq c_2,\\
    T(j,m)
    &\leq \sum_Y T(j-1, m_Y) + \bigO(m^{1+\epsilon'})
    &&\text{for all } j \in [t] \text{ and all } m > c_2.
  \end{align*}
  The sum ranges over all clusters \(Y\) in the neighbourhood cover \(\X_{\tX}\),
  and \(m_Y \deff n_{\tX Y}\).
  Recall that \(\sum_Y m_Y \leq m^{1+\epsilon'}\) by \cref{eq:Delta},
  using that the maximum degree of \(\X_{\tX}\) is at most \(m^{\epsilon'}\).
  Note that the clause \(T(0,m) \in \bigO(1)\) is justified by the fact
  that the height of the tree is at most \(t\).
  Therefore, all nodes of the form \((X_1, \dots, X_t)\) must be leaves.

  The same recurrence was obtained in \cite{GroheKreutzerSiebertz_2017_NowhereDense}.
  The straightforward analysis,
  carried out in \cite{GroheKreutzerSiebertz_2017_NowhereDense}
  (and using, among other things, the inequality
  \(\sum_i y_i^p \leq (\sum_i y_i)^p\)
  for \(y_i \geq 0\) and \(p \geq 1\),
  based on the fact that the \(\ell_p\) norm of a vector \((y_i)_i\)
  is bounded from above by its \(\ell_1\) norm),
  yields \(T(j,n) \in \bigO(n^{1 + 2j \epsilon'}) \subseteq \bigO(n^{1+\epsilon})\)
  for all \(j \in [0,t]\).
  In particular, this bounds the running time at the root
  to \(\bigO(n^{1+\epsilon})\)
  and completes the proof of \cref{lem:fo-mc}.
\end{proof}
 \end{document}